\documentclass[sigconf]{acmart}

\setcopyright{none} 
\acmConference[Accepted for ACM CCS]{ACM Conference on Computer and Communications Security}
\acmYear{2020}
\acmPrice{}
\acmISBN{}
\acmDOI{}
\copyrightyear{}

\usepackage{ amsthm }
\usepackage{pgfplots}
\usepackage{tikz}
\usetikzlibrary{shapes,snakes}
\usetikzlibrary{tikzmark}
\usepackage{xspace}
\usepackage{listings}
\usepackage{infer}
\usepackage{mathtools}
\usepackage{amsfonts}
\usepackage{amsmath}
\usepackage{color}
\usepackage{xcolor}
\definecolor{lightgray}{rgb}{.9,.9,.9}
\definecolor{darkgray}{rgb}{.4,.4,.4}
\definecolor{purple}{rgb}{0.65, 0.12, 0.82}
\definecolor{darkgreen}{rgb}{0, 0.5, 0}
\definecolor{turquoise}{rgb}{0, 0.5, 0.5}
\definecolor{plum}{rgb}{.4, .14, .37}
\usepackage{float}
\usepackage{subcaption}
\usepackage{multirow}
\usepackage{booktabs}
\usepackage{stmaryrd}
\usepackage{xargs}
\usepackage{url}

\usepackage{breakurl}
\usepackage{enumitem}
\usepackage{stackengine}
\stackMath
\usepackage{hyperref}
\usepackage[nameinlink]{cleveref}

\newif\ifauthors
\authorsfalse

\crefname{section}{\S}{\S\S}
\Crefname{section}{\S}{\S\S}
\crefname{lemma}{lemma}{lemmas}
\Crefname{lemma}{Lemma}{Lemmas}
\crefname{definition}{definition}{definitions}
\Crefname{definition}{Definition}{Definitions}

\renewcommand{\autoref}{\Cref}

\lstdefinelanguage{JavaScript}{
  keywords={typeof, new, true, false, catch, function, struct, mapping, return, null, catch, switch, var, if, in, while, do, else, case, break},
  keywordstyle=\color{blue}\bfseries,
  ndkeywords={class, export, boolean, throw, implements, import, this, contract, library},
  ndkeywordstyle=\color{turquoise}\bfseries,
  keywords={[3]bool,address,uint, uint256, string},
  keywordstyle=[3]\color{darkgreen}\bfseries,
  identifierstyle=\color{black},
  sensitive=false,
  comment=[l]{//},
  morecomment=[s]{/*}{*/},
  commentstyle=\color{purple}\ttfamily,
  stringstyle=\color{red}\ttfamily,
  morestring=[b]',
  morestring=[b]"
}

\lstdefinelanguage{HoRSt}{
  sensitive = true,
  keywords = [1]{for, in, op, sel, expect, test, query, datatype, pred, clause, rule, let },
  keywords = [2]{match, with, store, select, mod },
  keywords = [6]{true, false, @T, @V, @D, @ADD, @STOP,@INVALID, @SELFDESTRUCT, 0, 1, 2 },
  keywords = [4]{int, bool, array, AbsDom, Opcode, CallData },
  keywords = [5]{MState, Exc, Halt, ReturnData },
  keywordstyle=[4]\color{plum}\bfseries,
  keywordstyle=[5]\color{darkgreen}\bfseries,
  keywordstyle=[6]\color{darkgreen!90!black}\bfseries,
  comment=[l]{//},
  morecomment=[s]{/*}{*/}
}

\renewcommand{\subsubsection}[1]{\smallskip\noindent\textbf{#1.}}

\newcommand{\myparagraph}[1]{\smallskip \noindent \textit{#1}.}

\newtheorem*{theorem-non}{Theorem}

\newcommand*\circled[1]{\tikz[baseline=(char.base)]{
            \node[shape=circle,draw,inner sep=0.8pt] (char) {#1};}}

\newcommand\doubleplus{+\kern-1.3ex+\kern0.8ex}
\newcommand\mdoubleplus{\ensuremath{\mathbin{{+}\mkern-8mu{+}}}}
\newcommand\mcons{\ensuremath{\mathbin{:\mkern-8.5mu:}}}

\newcommand{\define}{:=}
\newcommand{\nil}{\epsilon}
\newcommand{\cons}[2]{{#1}\mcons{#2}}
\newcommand{\BB}{\mathbb{B}}
\newcommand{\NN}{\mathbb{N}}

\newcommand{\setof}[1]{\mathcal{P}(#1)}
\newcommand{\lam}[2]{\lambda {#1}. \, {#2}}

\newcommand{\access}[2]{{#1}. {#2}} 
\newcommand{\size}[1]{|{#1}|}

\newcommand{\fun}[2]{\lambda {#1}. \, {#2}}
\newcommand{\store}[3]{{#1}^{#2}_{#3}}
\newcommand{\select}[2]{{#1}[#2]}

\newcommand{\exstates}{\mathcal{S}}
\newcommand{\contracts}{\mathcal{C}}

\newcommand{\sstep}[3]{{#1} \vDash {#2} \, \rightarrow {#3}} 
\newcommand{\ssteps}[3]{{#1} \vDash {#2} \, \rightarrow^* {#3}} 
\newcommand{\nsteps}[4]{{#1} \vDash {#2} \, \rightarrow^{#4} {#3}}

\newcommand{\callstack}{S} 
\newcommand{\callstackb}{U}  
\newcommand{\transenv}{\Gamma} 

\newcommand{\memv}{\textit{m}}
\newcommand{\stackv}{\textit{s}}
\newcommand{\actwv}{\textit{i}}
\newcommand{\datav}{\textit{data}}
\newcommand{\instr}{\textit{inst}}

\newcommand{\accountv}{a}

\newcommand{\mstate}{\mu}
\newcommand{\exenv}{\iota}
\newcommand{\gstate}{\sigma}

\newcommand{\regstatefull}[4]{(#1, #2, #3, #4)}
\newcommand{\excstate}{\textit{EXC}}
\newcommand{\haltstate}[4]{\textit{HALT}(#1, #2, #3, #4)}
\newcommand{\haltstatefull}[4]{\textit{HALT}(#1, #2, #3, #4)}

\newcommand{\lgas}{\textit{gas}}
\newcommand{\callstackplain}{U}
\newcommand{\smstate}[5]{(#1, #2, #3, #4, #5)}

\newcommand{\accountstate}[4]{(#1, #2, #3, #4)}

\newcommand{\gas}{\textsf{gas}}

\newcommand{\pc}{\textsf{pc}}
\newcommand{\stack}{\textsf{s}}

\newcommand{\activeaccount}{\textsf{actor}}

\newcommand{\activecode}{\textsf{code}}

\newcommand{\transeffects}{\eta}
\newcommand{\exconf}[2]{(#1, #2)}

\newcommand{\concatstack}[2]{{#1}\mdoubleplus{#2} }

\newcommand{\getcontractcode}[1]{\textit{code} \, ({#1})}
\newcommand{\annotate}[2]{{#1}_{#2}}
\newcommand{\exstate}{s}

\newcommand{\ADD}{\textsf{ADD}}
\newcommand{\CALLCODE}{\textsf{CALLCODE}}
\newcommand{\DELEGATECALL}{\textsf{DELEGATECALL}}
\newcommand{\STATICCALL}{\textsf{STATICCALL}}
\newcommand{\CALL}{\textsf{CALL}}
\newcommand{\STOP}{\textsf{STOP}}
\newcommand{\RETURN}{\textsf{RETURN}}
\newcommand{\CREATE}{\textsf{CREATE}}

\newcommand{\JUMP}{\textsf{JUMP}}
\newcommand{\JUMPI}{\textsf{JUMPI}}

\newcommand{\CALLDATALOAD}{\textsf{CALLDATALOAD}}

\newcommand{\MLOAD}{\textsf{MLOAD}}
\newcommand{\MSTORE}{\textsf{MSTORE}}

\newcommand{\SSTORE}{\textsf{SSTORE}}
\newcommand{\JUMPDEST}{\textsf{JUMPDEST}}

\newcommand{\INVALID}{\textsf{INVALID}}

\newcommand{\pluseq}{\mathrel{+}=}
\newcommand{\minuseq}{\mathrel{-}=}
\newcommand{\cond}[3]{{#1} \, ? \,  {#2}\, : \,  {#3}}

\newcommand{\update}[3]{{#1}[{#2} \rightarrow {#3}]}

\newcommand{\arraypos}[2]{{#1} \, [#2]}

\newcommand{\updategstate}[3]{{#1} \big \langle {#2} \rightarrow {#3} \big \rangle}

\newcommand{\extract}[3]{{#1}_{[#2, #3]}}
\newcommand{\inc}[3]{{#1}[{#2}\pluseq{#3}]}
\newcommand{\dec}[3]{{#1}[{#2}\minuseq{#3}]}

\renewcommand{\exconf}[2]{#1}

\newcommand{\addr}{\textit{addr}}

\newcommand{\code}{\textit{code}}

\newcommand{\emptymemory}{\lam{x}{0}}
\newcommand{\emptystack}{\epsilon}

\newcommand{\anacontract}{c^*}

\newcommand{\absdom}{\hat{D}}

\newcommand{\pkec}[2]{\textsf{Kec} \, (#1, #2)}

\newcommand{\pmstate}[6]{\textsf{MState}_{\mathsf{#1}} \,((#2, #3), #4, #5)}

\newcommand{\hornclause}{H}

\newcommand{\lpc}{\textit{pc}}
\newcommand{\stor}{\textit{stor}}
\newcommand{\siz}{\textit{size}}

\newcommand{\cl}{\textit{c$\ell$}}

\newcommand{\abs}[1]{\hat{#1}} 

\newcommand{\absop}[1]{~\widehat{#1}~} 

\newcommand{\derives}{\vdash}
\newcommand{\vars}{\textit{Vars}}

\newcommand{\callinstructions}{\textit{Inst}_{\textit{call}}}

\newcommand{\binop}{\textit{op}_\textit{bin}}

\newcommand{\compop}{\textit{op}_\textit{comp}}

\newcommand{\aconcat}[1]{{||}_{#1}}

\newcommand{\config}{c}
\newcommand{\absconfig}{\Delta}
\newcommand{\smalls}{S_\textit{s}}

\newcommand{\configs}{\mathcal{C}}
\newcommand{\absconfigs}{\mathcal{A}}
\newcommand{\predsig}{\mathcal{S}}
\newcommand{\abss}{\Lambda} 
\newcommand{\hclauses}[1]{\mathcal{H}(#1)}
\newcommand{\prednames}{\mathcal{N}}
\newcommand{\absdoms}{\mathcal{D}}
\newcommand{\proj}[2]{\pi_{#1}(#2)}
\newcommand{\ord}[1]{\leq_{#1}}
\newcommand{\dro}[1]{\geq_{#1}}
\newcommand{\tsize}[1]{|{#1}|}
\newcommand{\configabs}[1]{\alpha}

\newcommand{\absts}{\Lambda}
\newcommand{\pred}{p}
\newcommand{\pnmstate}[1]{\textsf{MState}_{#1}}
\newcommand{\pnexception}{\textsf{Exc}}
\newcommand{\pnhalt}{\textsf{Halt}}

\newcommand{\mem}{m}
\newcommand{\predmstate}[6]{\textsf{MState}_{#1}((#2, #3), #4, #5, #6)}
\newcommand{\predexception}[1]{\textsf{Exc}(#1)}
\newcommand{\predhalt}[2]{\textsf{Halt}(#1, #2)}
\newcommand{\cstar}{c^*}
\newcommand{\exstateabs}{\alpha_{\exstate}}
\newcommand{\stacktoarray}{\textsf{stackToArray}}
\newcommand{\towordmem}{\textsf{toWordMem}}
\newcommand{\instabs}[1]{\llparenthesis #1 \rrparenthesis}

\newcommand{\arrayinit}[1]{\fun{x}{#1}}

\newcommand{\queryset}{\Delta_{\textit{query}}}
\newcommand{\prop}[2]{\mathcal{R}(#1, #2)}

\newcommand{\bad}{R}

\newcommand{\ncollfreestep}[2]{\stackon[-3pt]{\leadsto^{#1}}{\scriptstyle #2\,\,}}

\newcommand{\accessword}[2]{\textit{getWord}(#1, #2)}

\newcommand{\sstepfun}[1]{F_{#1}}
\newcommand{\derivefun}[1]{F'_{#1}}
\newcommand{\lfp}[1]{\textit{lfp}[#1]}
\newcommand{\compose}{\circ}

\newcommand{\isregular}[1]{\textit{isRegular}(#1)}
\newcommand{\ishalt}[1]{\textit{isHalt}(#1)}

\newcommand{\hc}[4]{\forall{#1}.~{#2}, {#3} \Rightarrow {#4}}
\newcommand{\constraints}{\Phi}
\newcommand{\premises}{P}
\newcommand{\conclusion}{c}
\newcommand{\derivesone}{\derives^1}
\newcommand{\valuation}{V}
\newcommand{\satisfies}{\vDash}

\newcommand{\xvdash}[1]{%
  \vdash^{\mkern-7mu\scriptscriptstyle\rule[-.5ex]{0pt}{0pt}#1}%
}

\newcommand\mydots{\hbox to 1em{.\hss.\hss.}}

\newcommand{\horst}{\emph{HoRSt}}
\newcommand{\souffle}{\emph{Souffl\'e}}
\newcommand{\java}{\emph{Java\textsuperscript{\texttrademark}}}
\newcommand{\smtlib}{\texttt{smt-lib}}

\newcommand{\zz}{\emph{z3}}
\newcommand{\ethor}{\emph{eThor}}

\newcommand{\hornc}{Horn\xspace}

\definecolor{seablue}{HTML}{004C99}
\definecolor{darkorange}{HTML}{CC6600}
\definecolor{mediumblue}{HTML}{0050EF}
\definecolor{darkblue}{HTML}{001DBC}
\definecolor{pink}{HTML}{D80073}
\definecolor{darkpink}{HTML}{A50040}
\definecolor{lilac}{HTML}{AA00FF}
\definecolor{darklilac}{HTML}{7700CC}
\definecolor{brown}{HTML}{663300}
\newcommand{\configpic}{\tikz{\filldraw[color=seablue] (0,0) circle [radius=2pt];}}
\newcommand{\badconfigpic}{\tikz{\node[circle, fill=seablue, draw=red, text width=4pt, inner sep=0pt]{};}}
\newcommand{\absconfigpic}{\tikz{\node[regular polygon, regular polygon sides =3, fill, color=darkorange, text width=3pt, inner sep=0pt]{};}}
\newcommand{\badabsconfigpic}{\tikz{\node[regular polygon, regular polygon sides =3, fill, color=brown, text width=3pt, inner sep=0pt]{};}}

\newcommand{\absrulepic}[3]{\tikz{\node[regular polygon, regular polygon sides =4, fill=#2, draw= #3, minimum width=2pt, inner sep=0pt, scale=0.7, text=white, anchor=base]{#1};}}

\newcommand{\absruletwo}[1]{\absrulepic{#1}{mediumblue}{darkblue}}
\newcommand{\absruleone}[1]{\absrulepic{#1}{pink}{darkpink}}
\newcommand{\absrulethree}[1]{\absrulepic{#1}{lilac}{darklilac}}

\begin{document}

\title{\ethor{}: Practical and Provably Sound Static Analysis \\ of Ethereum Smart Contracts}

\author{Clara Schneidewind \and Ilya Grishchenko \and Markus Scherer \and Matteo Maffei}
\affiliation{TU Wien}
\email{{clara.schneidewind, markus.scherer,ilya.grishchenko,matteo.maffei}@tuwien.ac.at}

\begin{abstract}
  Ethereum has emerged as the most popular smart contract development platform, with hundreds of thousands of contracts stored on the blockchain and covering a variety of application scenarios, such as auctions, trading platforms, and so on. Given their financial nature, security vulnerabilities may lead to catastrophic consequences and, even worse, they can be hardly fixed as data stored on the blockchain, including the smart contract code itself, are immutable.  An automated security analysis of these contracts is thus of utmost interest, but at the same time technically challenging for a variety of reasons, such as the specific transaction-oriented programming mechanisms, which feature a subtle semantics, and the fact that the blockchain data which the contract under analysis interacts with, including the code of callers and callees, are not statically known.

  In this work, we present \ethor{}, the first sound and automated static analyzer  for EVM bytecode, which is based on an abstraction of the EVM bytecode semantics based on \hornc clauses. In particular, our static analysis supports reachability properties, which we show to be sufficient for capturing interesting security properties for smart contracts (e.g., single-entrancy) as well as contract-specific functional properties. Our analysis is proven sound against a complete semantics of EVM bytecode and an experimental large-scale evaluation on real-world contracts demonstrates that \ethor{} is practical and outperforms the state-of-the-art static analyzers: specifically, \ethor{} is  the only one to provide soundness guarantees,  terminates on 95\% of a representative set of real-world contracts, and achieves an $F$-measure (which combines sensitivity and specificity) of 89\%.
\end{abstract}




\maketitle

\section{Introduction}
\label{sec:intro}

Smart contracts introduced a radical paradigm shift in distributed computation, promising security in an adversarial setting thanks to the underlying consensus algorithm. Software developers can implement sophisticated distributed, transaction-based computations by leveraging the scripting language offered by the underlying blockchain technology. While many cryptocurrencies have an intentionally limited scripting language (e.g., Bitcoin~\cite{nakamoto2008bitcoin}), Ethereum was designed from the ground up with a quasi Turing-complete language\footnote{While the language itself is Turing complete, computations are associated with a bounded computational budget (called gas), which gets consumed by each instruction thereby enforcing termination.}. Ethereum smart contracts have thus found  a variety of appealing use cases, such as  auctions~\cite{hahn2017smart}, data management systems~\cite{adhikari2017secure}, financial contracts~\cite{biryukov2017findel}, elections~\cite{McCorrySH17},  trading platforms~\cite{notheisen2017trading,mathieu2017blocktix}, permission management \cite{azaria2016medrec} and verifiable cloud computing~\cite{DWAMM::17}, just to mention a few. 
  Given their financial nature, bugs and vulnerabilities in smart contracts may lead to catastrophic consequences. For instance, the infamous DAO vulnerability~\cite{thedao} recently led to a 60M\$ financial loss and similar vulnerabilities occur on a regular basis~\cite{paritya,parityb}. Furthermore, many smart contracts in the wild are  intentionally fraudulent, as highlighted in a recent survey~\cite{survey}. Even worse, due to the unmodifiable nature of blockchains, bugs or vulnerabilities in deployed smart contracts cannot be fixed.

 A rigorous security analysis of smart contracts is thus crucial for the trust of the society in blockchain technologies and their widespread deployment. Unfortunately, this task is quite challenging for various reasons. First, Ethereum smart contracts are developed in an ad-hoc language, called Solidity, which resembles JavaScript but features non-standard semantic behaviours and transaction-oriented mechanisms, which complicate smart contract development and verification. Second, smart contracts are uploaded on the blockchain in the form of Ethereum Virtual Machine (EVM) bytecode, a stack-based low-level code featuring  very little static information, which makes it extremely difficult to analyze. Finally, most of the data available at runtime on the blockchain, including the contracts which the contract under analysis may interact with, may not be known statically, which requires ad-hoc abstraction techniques. As a result, while effective bug finding tools for smart contracts have been recently presented,  \emph{there exists at present no automated security analysis for EVM bytecode that provides formal security guarantees} (i.e., absence of false negatives, as proven against a formal semantics of EVM bytecode), as further detailed below.

\subsection{State-of-the-art in Security Analysis of Smart Contracts}
\label{related}

Existing approaches to smart contract analysis can be mainly classified as interactive frameworks for semantic-based machine-checked proofs~\cite{hirai2017defining,amani2018towards,hildenbrandt2017kevm,bhargavan2016formal,GMS::POST18,yang2019fether} and automated, heuristic-driven bug finding tools~\cite{luu2016making,zhou2018security,nikolic2018finding,Grech:2018:MSO,Krupp:2018}.

Some recent works try to fill the middle ground between these two approaches, aiming at the best of the two worlds, i.e.,  an automated, yet sound static analysis of Ethereum smart contracts that can prove generic security properties~\cite{kalra2018zeus,securify,lu2019neucheck,GMS::CAV18}. We conducted a thorough investigation, finding out that all of them  fail to provide  the intended soundness guarantees, which showcases the difficulty of this task. In the following, we further expand on this point, highlighting the particular challenges that occur in the process of designing a sound static analysis tool for Ethereum smart contracts.

\paragraph{Semantic foundations}
A first fundamental limitation of most existing static analysis tools is that they do not establish a formal connection with a semantic model of smart contract execution.
ZEUS~\cite{kalra2018zeus} leverages existing symbolic model checking tools for LLVM bitcode in order to analyze contracts written in Solidity. To this end, ZEUS first translates Solidity code into an abstract intermediate language and in a next step into LLVM bitcode.
However, upto now, there is no complete formal semantics of the Solidity language, hence making it impossible to prove the performed translation to be semantics-preserving and consequently to derive formal guarantees from the results of the LLVM model checking. This can easily lead to flaws, as confirmed by the theoretical investigation conducted in~\cite{GMS::CAV18} as well as by empirical evidence provided in~\cite{torres2018osiris}, which contradict the original soundness claim~\cite{kalra2018zeus}.
\cite{GMS::CAV18} reviews a theoretical approach to a static analysis technique based on \hornc clauses which is claimed to be provably sound, still we could find sources of unsoundness in the presented abstraction as detailed in \autoref{sec:ethertrust}.
Securify~\cite{securify} is an abstract interpreter working at the level of EVM bytecode that also aims to offer soundness guarantees: unfortunately, it does not come with any formal semantics or proof of soundness, which leads to both false positives and false negatives, as discussed below.

\paragraph{Formal security properties}
As hinted in the previous paragraph, for providing reliable guarantees, not only the analysis but also the security properties have to be formalized in the underlying semantic model. All reportedly sound tools do not accomplish that. While for ZEUS the intended properties are just informally described, Securify comes with an ad-hoc formalism for characterizing security properties of smart contracts. This, however, is not related to a formal EVM bytecode semantics, nor are the security patterns that are used in the analysis to determine  the fulfillment  and  violation of these properties\footnote{\cite{securify} introduces compliance and violation patterns for security properties where a contract matching a compliance pattern is meant to satisfy the property and a contract matching a violation pattern  to violate it.} provably related to such formal characterization.
This omission results in the lack of soundness and completeness guarantees, as we illustrate in \autoref{sec:securify} by providing counterexamples for the majority of the proposed patterns.
Similarly, the tool NeuCheck~\cite{lu2019neucheck} performs a purely syntactic analysis on Solidity source code and defines security properties by syntactic patterns on the smart contract's syntax tree. These patterns cannot be related to any semantic property due to a lacking formalism, and can be shown to be neither necessary nor sufficient for the corresponding security properties, see \autoref{sec:neucheck}.

\paragraph{Correct control flow reconstruction}
Analyzing EVM bytecode is particularly challenging as the underlying execution model allows for dynamic jump destinations. Most works~\cite{securify,luu2016making,albert2018ethir,trailofbits-manticore} reconstruct the control flow of a given smart contract before the analysis. However, recovering jump destinations is interconnected with the contract's execution, and hence, performing such a sound reconstruction is not trivial. For instance, ~\cite{securify} uses a custom algorithm -- whose correctness is never discussed -- for doing so. Indeed we found an example showing that this algorithm yields unsound results(see \autoref{sec:appendix-cfg}), undermining the soundness of the analysis.

\paragraph{Practicality}
A useful, automated analysis tool needs to be performant, not only in terms of overall execution time, but also in terms of precision.
This is particularly challenging as the soundness goal prevents the use of (potentially cheap and fast) heuristics to guide the analysis, but instead requires the chosen abstractions to provably over-approximate the set of all possible executions.
Appropriate abstractions hence need to be sound, but still efficiently encodable and precise enough to account for a contract's safety.

\paragraph{Benchmarking, testing, and community validation}
The previous problems which affect the design of reliable analysis tools are aggravated by the fact that there is no reliable and comparable benchmarking or testing infrastructure for Ethereum smart contract analysis tools.
One reason for that is the lack of clear definitions for the generic security properties targeted by the analysis tools in the first place. Another explanation is the difficulty of manually investigating the bytecode of real-world contracts for assessing their ground truth.
Even though the existing works evaluate their performance on real world smart contracts (fetched from the blockchain), the used ground truth is spurious: While~\cite{securify} reports quality metrics only on a dataset of 100 contracts which are not made available, ~\cite{kalra2018zeus} presents results on a dataset encompassing over 1500 contracts from the blockchain. When manually investigating this dataset, however, we found several issues ranging from non-existing contracts to deviating ground truths. These problems are detailed  in~\autoref{sec:eval}.

Inspired by the issues that we see in the state of the art, we introduce a principled approach to the design and implementation of a sound, yet performant, static analysis tool for EVM bytecode.

\subsection{Our Contributions}
The contributions of this work can be summarized as follows:

\begin{itemize}
	\item We design the first provably sound static analyzer for EVM bytecode, which builds on top of a reachability analysis realized by \hornc clause resolution. We show that a reachability analysis suffices to verify various interesting security properties for smart contracts as well as contract-specific functional properties via an encoding into Hoare-style reasoning. The design of such static analysis is technically challenging, since it requires careful abstractions of various EVM components (e.g., the stack-based execution model, the gas used to bound the smart contract execution, and the memory model) as well as a dedicated over-approximation of the data stored on the blockchain, which is not statically known and yet the contract under analysis can interact with (e.g., the code of other contracts which may act both as callers and callees);
	\item  We prove the soundness of our static analysis technique against the semantics of EVM bytecode formalized  by Grishchenko et al.~\cite{GMS::POST18};
	\item In order to facilitate future refinements of our analysis, as well as the design of similar static analyses for other languages, we design and implement \horst{}, a framework for the specification and implementation of static analyses based on \hornc clause resolution. Specifically, \horst{} takes as input a (pen-and-paper like) specification of the \hornc clauses defining the static analysis and produces a \texttt{smt-lib}~\cite{smtlib} encoding suitable for \zz{}~\cite{pdr}, which includes various optimizations such as \hornc clause and constant folding;
	\item We use \horst{} to implement the static analyzer \ethor{}. To gain confidence in the resulting implementation, we encode the relevant semantic tests (604 in total) of the official EVM suite as reachability properties, against which we successfully test the soundness and precision of \ethor{};
	\item  We conduct a large-scale experimental evaluation on real-world contract data comparing \ethor{} to the state-of-the-art analyzer ZEUS~\cite{kalra2018zeus} which claims to provide soundness guarantees.
	While ZEUS shows a remarkable specificity (i.e., completeness) of $99.8\%$, \ethor{} clearly outperforms ZEUS in terms of recall (i.e., soundness) -- $100\%$ vs. $11.4\%$ --  which empirically refutes ZEUS' soundness claim. With a specificity of  $80\%$, \ethor{}  results in an overall performance of $88.9\%$  (according to the F-measure) as compared to ZEUS' F-measure of $20.4\%$.
\end{itemize}

The remainder of this paper is organized as follows.~\autoref{sec:ethereum} reviews Ethereum and the semantics of EVM bytecode.~\autoref{sec:stat} introduces our static reachability analysis, specifies its soundness guarantee and discusses relevant smart contract properties in scope of the analysis.~\autoref{sec:horst} introduces the specification language \horst{}. \autoref{sec:eval} describes \ethor{} and presents our experimental evaluation.~\autoref{sec:conclusion} concludes by discussing interesting future research directions.
The appendix provides proofs and additional material. The source code of \ethor{} and \horst{} with the dataset used in  the experimental evaluation are available online~\cite{extended}.

\section{Ethereum}
\label{sec:ethereum}
We first introduce the required background on Ethereum~(\autoref{subsec:eth-background}) and then overview an existing semantics of EVM bytecode (\autoref{subsec:smallstep}), which this work builds on.

\subsection{Background}
\label{subsec:eth-background}
The Ethereum platform can be seen as a transaction-based state machine where transactions alter the global state of the system, which consists of accounts. There are two types of accounts: External accounts, which are owned by a user of the system, and contract accounts, which can be seen as  a distributed program. All accounts hold a balance in  the virtual currency \emph{Ether}. Additionally, contract accounts include persistent storage and the contract's code.
Transactions can either create new contract accounts or call existing accounts. Calls to external accounts can only transfer Ether to this account, but calls to contract accounts additionally execute the account's contract code. The contract execution might influence the storage of the account and might as well perform new transactions -- in this case, we speak of \emph{internal transactions}.
The effects of contract executions are determined by the \emph{Ethereum Virtual Machine} (EVM). This virtual machine characterizes the \emph{quasi Turing complete} execution model of Ethereum smart contracts where the otherwise Turing complete execution is restricted by an upfront defined resource \emph{gas} that effectively limits the number of execution steps.
A transaction's originator can specify an upper bound on the gas that she is willing to pay for the contract execution and also sets the gas price (the amount of Ether to pay for a unit of gas). The originator then prepays the specified gas limit and gets refunded according to the remaining gas in case of successful contract execution.

\subsubsection{EVM bytecode}
Contracts are published on the blockchain in form of \emph{EVM bytecode}--  an Assembler like bytecode language. The EVM is a stack-based machine and specifies the semantics of bytecode instructions. Consequently, EVM bytecode mainly consists of standard instructions for stack operations, arithmetics, jumps and local memory access. The instruction set is complemented with blockchain-specific instructions such as an opcode for the SHA3 hash and several opcodes for accessing information on the current (internal) transaction.
In addition, there are opcodes for accessing and modifying the storage of the executing account and distinct opcodes for initiating internal transactions.

Each instruction is associated with (a potentially environment-dependent) gas cost. If the up-front defined gas-limit is exceeded during execution, the transaction execution halts exceptionally and the effects of the current transaction on the global state are reverted.
For nested transactions, an exception only reverts the effects of the executing transaction, but not those of the calling transactions.

\subsubsection{Solidity}
In practice, Ethereum smart contracts are shipped and executed in EVM bytecode format but are, for a large part, written in the high-level language Solidity, which is developed by the Ethereum Foundation~\cite{solidity}.
The syntax of Solidity resembles JavaScript, enriched with additional primitives accounting for the distributed setting of Ethereum.
Solidity exhibits specific features that give rise to smart contract vulnerabilities, as will be in discussed in~\autoref{subsec:props}.
We will not give a full account of Solidity's language features here, but add explanations throughout the paper when needed.

\subsection{EVM Semantics}
\label{subsec:smallstep}
Our static analysis  targets a recently introduced small-step semantics for EVM bytecode~\cite{GMS::POST18}, which we shortly review below.\footnote{More recent changes to the EVM semantics such as the introduction of $\STATICCALL$, $\CREATE$, and $\CREATE2$, are not explicitly mentioned in this paper, but covered by our static analysis as specified in~\cite{extended}.}

The semantics of EVM bytecode is given by a small-step relation $\sstep{\transenv}{\exconf{\callstack}{\transeffects}}{\exconf{\callstack'}{\transeffects'}}$ that encompasses the possible steps that a callstack $\callstack$, representing the overall state of a contract execution,  can make under the transaction environment $\transenv$.
The transaction environment $\transenv$ summarizes static information about the transaction execution such as the information on the block that the transaction is part of and transaction-specific information such as gas price or limit.
We write $\ssteps{\transenv}{\exconf{\callstack}{\transeffects}}{\exconf{\callstack'}{\transeffects'}}$ for the reflexive transitive closure of the small-step relation and call the pair $(\transenv, \callstack)$ a \emph{configuration}.

\subsubsection{Configurations}
The most important formal components of EVM configurations are summarized in \autoref{fig:grammar'}.

\begin{figure}
{\small
\begin{mathpar}
\begin{array}{rlll}
\text{Callstacks} & \callstack & \define & \cons{\excstate}{\callstackplain} ~|~  \cons{\haltstatefull{\gstate}{\lgas}{d}{\transeffects}}{\callstackplain} ~|~ \callstackplain  \\
\text{Plain callstacks} & \callstackplain & \define & \cons{\regstatefull{\mstate}{\exenv}{\gstate}{\transeffects}}{\callstackplain} ~|~ \epsilon \\
\text{Machine states} & \mstate & \define & \smstate{\lgas}{\lpc}{m}{i}{s} \\
\text{Account states} & \accountv & \define & \accountstate{n}{b}{\textit{code}}{\textit{stor}}   \\
\end{array}
\end{mathpar}}
\caption{Grammar for calls stacks}
\label{fig:grammar'}
\end{figure}

\paragraph{Global State}
Ethereum's global state $\gstate$ is formally captured as a (partial) mapping from account addresses to account states.
An account state consists of a nonce $n$ that is incremented with every other account that the account creates, a balance $b$, a persistent storage $\stor$, and the account's $\code$. External accounts have no code and hence cannot access storage.

\paragraph{Callstacks}
The overall state of an external transaction is captured by a stack of execution states that we will refer to as \emph{callstack}. The individual execution states reflect the states of the pending internal transactions.
More formally, the elements of a callstack are either regular execution states of the form $\regstatefull{\mstate}{\exenv}{\gstate}{\transeffects}$ or terminal execution states $\haltstate{\gstate}{\lgas}{d}{\transeffects}$ and $\excstate$ which can only occur as stack top elements.
For terminated executions we differentiate between exceptional halting ($\excstate$), which will revert all effects of the transaction, and regular halting $\haltstate{\gstate}{\lgas}{d}{\transeffects}$, in which case the effects of the transaction are captured by the global state $\gstate$ at the point of halting, the $\lgas$ remaining from the execution, the return data $d$, and the transaction effects $\transeffects$ (effects that will only be applied after completing the external execution).

The state of a non-terminated internal transaction is described by a regular execution state of the form $\regstatefull{\mstate}{\exenv}{\gstate}{\transeffects}$. During execution, this state tracks the current global state $\gstate$ of the system, the execution environment $\exenv$ to the internal transaction (which specifies parameters such as the input to the transaction and the code to be executed) as well as the local state $\mstate$ of the stack machine, and the transaction effects $\transeffects$ that will be applied after transaction execution.

The local machine state $\mstate$ reflects the state of the stack machine that handles local computations.
It is represented by a tuple $\smstate{\lgas}{\textit{pc}}{\textit{m}}{\textit{i}}{\textit{s}}$ holding the amount of $\lgas$ available for execution, the program counter $\lpc$, the local memory $\textit{m}$, the number of active words in memory $\textit{i}$, and the machine stack $\textit{s}$.
As the stack machine models local computations, the execution of every new (internal) transaction starts again in a fresh machine state at program counter zero with an empty stack and zero-initialized memory. Only the gas value is initialized as specified by the initiator of the transaction. 

\subsubsection{Small-step Rules}
We illustrate the mechanics of the EVM bytecode semantics by an example and refer to ~\cite{GMS::POST18} for a full definition.

Local instructions, e.g., \ADD{}, only operate on the machine state:
\begin{mathpar}
\small\infer{
\arraypos{\access{\exenv}{\code}}{\access{\mstate}{\pc}}= \ADD \\
\access{\mstate}{\stack} = \cons{a}{\cons{b}{s}} \\
\access{\mstate}{\gas} \geq 3 \\
\mstate'= \dec{\inc{\update{\mstate}{\stack}{\cons{(a+b)}{s}}}{\pc}{1}}{\gas}{3}}
{\sstep{\transenv}{\cons{\regstatefull{\mstate}{\exenv}{\gstate}{\transeffects}}{\callstack}}{\cons{\regstatefull{\mstate'}{\exenv}{\gstate}{\transeffects}}{\callstack}}}
\end{mathpar}
Given a stack that contains at least two values and given a sufficient amount of gas (here $3$ units), an $\ADD$ instruction takes two values from the stack and pushes their sum. These effects, as well as the advancement in the program counter and the substraction of the gas cost, are reflected in the updated machine state $\mstate'$.

A more evolved semantics is exhibited by the class of transaction initiating instructions  ($\CALL$, $\CALLCODE$ and $\DELEGATECALL$, $\CREATE$). Intuitively, $\CALL$ executes the callee's code in its own environment, $\CALLCODE$ executes the callee's code in the caller's environment, which might be useful to call libraries  implemented in a separate  contract,  and $\DELEGATECALL$ takes a step further by preserving not only the caller's environment but even part of the environment of the previous call (e.g., the sender information), which effectively treats the callee's code as an internal function of the caller's code. Finally, the $\CREATE$ instruction initiates an internal transaction that creates a new account.

Instructions from this set are particularly difficult to analyze, since their arguments are dynamically evaluated and the execution environment has to be tracked and properly modified across different calls. Furthermore, it can well be that the code of a called function is not accessible at analysis time, e.g., because 
the contract allows for money transfers to a dynamic set of contracts (like in the DAO contract as will be discussed in the next section).

\subsection{Security Properties of Smart Contracts}
\label{subsec:props}

Ethereum smart contracts have undergone several severe attacks in the past that where enabled by major bugs in the contract code, most prominently the DAO hack~\cite{thedao}.
Interestingly, this bug can be traced back to the violation of a generic security property of the attacked contract, called single-entrancy. We will shortly present the class of reentrancy attacks and the corresponding security property.

\subsubsection{Preliminary Notions}
In order to present security properties in a concise fashion, the previously presented small-step semantics is augmented with 
an annotation to callstack elements that reflects the currently executed contract.
We say that an execution state $\exstate$ is \emph{strongly consistent} with annotation $c$ if $\exstate$ executes $c$ (according to the execution environment) and $c$ is present in the global state of $\exstate$.
Further, for arguing about EVM bytecode executions, we are only interested in those initial configurations that might result from a valid external transaction in a valid block. We call these configurations \emph{reachable} and refer to~\cite{GMS::POST18} for a detailed definition.

\subsubsection{Single-entrancy}
For motivating the definition of single-entrancy, we introduce a class of bugs in Ethereum smart contracts called \emph{reentrancy bugs}~\cite{luu2016making,survey}.
Reentrancy attacks exploit that a contract which hands over control to another contract by calling it can be called back (reentered) before completing the original internal transaction. At the point of reentering the contract can then be in an inconsistent state which allows for unintended behavior. In the DAO hack, the attacker stole all funds of the contract  by  reentering the contract and   sending money to itself.
We exemplify this kind of attack by the  \lstinline|Bank| contract in \autoref{fig:reentrancy}: this  has a basic reentrancy protection in place which however can easily be circumvented.

The \lstinline|Bank| contract implements a simple banking functionality, keeping the balance of all users (identified by their addresses) in the mapping \lstinline|bal.| We only discuss the contract function \lstinline|drain| which allows a user to transfer all its money from its bank account to the provided beneficiary address \lstinline|a|.
For protecting against reentrancy, the \lstinline|drain| function implements a simple locking functionality: it is only entered in case the lock is not taken . Otherwise it takes the lock (using function \lstinline|take|), transfers the remaining balance of the function callee (denoted by \lstinline|msg.sender|) to the beneficiary address \lstinline|a|, updates the user's balance, and releases the lock again.
One needs to note that the \lstinline|call| construct (being translated to a $\CALL$ instruction in EVM bytecode) does not only trigger the value transfer, but also invokes the execution of the callee's so-called \emph{fallback function} (written as a function without name or argument in Solidity as depicted in the \lstinline|Mallory| contract in \autoref{fig:reentrancy}). Hence, the use of a \lstinline|call| can cause the the executed contract to be reentered during execution, potentially undermining the contract integrity. The locking mechanism should prevent this problem by causing an exception in case the contract is reentered (indicated by the lock being taken).
However, since the locking functionality is publicly accessible, a reentrancy attack (as depicted in \autoref{fig:reentrancy}) is still possible: An attacker calling the \lstinline|drain| function (via \lstinline|Mallory|) with \lstinline|Mallory|'s address as argument ({\color{darkgreen}\circled{1}}) transfers all of \lstinline|Mallory|'s money back to her and executes her fallback function ({\color{darkgreen}\circled{2}}). \lstinline|Mallory| then first calls the public \lstinline|release| function to release the lock ({\color{red}\circled{3}}) and next calls the \lstinline|drain| function of \lstinline|Bank| again ({\color{red}\circled{4}}). Since the attacker's balance has not been updated yet at this point, another transfer of the prior amount to \lstinline|Mallory| can be performed ({\color{red}\circled{5}}).
These steps can be continued until running out of gas or reaching the callstack limit. In both cases the last value transfer is rolled back, but the effects of all former internal transactions persist, leaving the contract \lstinline|Bank| drained-out.

\begin{figure}
\includegraphics[width=\columnwidth]{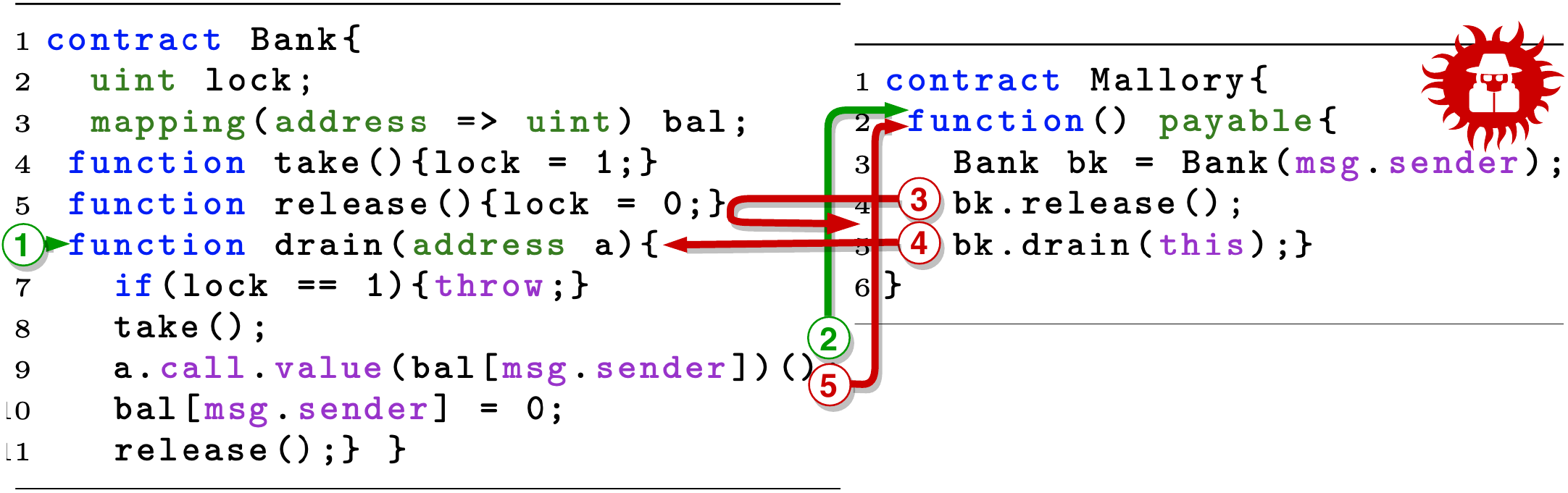}
\caption{Reentrancy Attack.}
\label{fig:reentrancy}
\end{figure}
The security property ruling out these attacks is called \emph{single-entrancy} and is formalized below. Intuitively, a contract is single-entrant if it cannot perform any more calls after reentering.

\begin{definition}[Single-entrancy~\cite{GMS::POST18}]
\label{def:single-entrancy}
	A contract $c$ is single-entrant if for all reachable configurations $(\transenv,\cons{\annotate{\exstate}{c}}{\callstack})$, for all $s', s'',\callstack' $
\par
\nobreak
	{\small
	\noindent
	\begin{align*}
	&\ssteps{\transenv}{\cons{\annotate{\exstate}{c}}{\callstack}}{\concatstack{\cons{\annotate{\exstate}{c}'}{\callstack'}}{\cons{\annotate{\exstate}{c}}{\callstack}}} \\ 
&\implies
	\neg \exists \exstate'', c'. \, \ssteps{\transenv}{\concatstack{\cons{\annotate{\exstate}{c}'}{\callstack'}}{\cons{\annotate{\exstate}{c}}{\callstack}}}{\cons{\annotate{s}{c'}''}{\concatstack{\cons{\annotate{\exstate}{c}'}{\callstack'}}{\cons{\annotate{\exstate}{c}}{\callstack}}}} 
	\end{align*}
	}%
\end{definition}%
where $\mdoubleplus$ denotes concatenation of callstacks.
The property expresses that after reentering a contract $c$ (in state $\exstate'_c$)  while executing a call initiated by the very same contract, it is not possible anymore to perform another internal transaction (which would result in adding another element $\exstate''_{c'}$ to the call stack). Note that the call stack records the sequence of calling states, hence the suffix $\cons{\annotate{\exstate}{c}}{\callstack}$ indicates a pending call initiated by the execution $\exstate$ of contract $c$.

Single-entrancy is particularly interesting in that it constitutes a generic robustness property of smart contracts. In contrast, other prominent vulnerabilities~\cite{paritya,parityb} are caused by functional correctness issues that are particular to a specific contract. For spotting such issues, contract-specific correctness properties need to be defined and verified. We discuss the formalization of such properties in~\autoref{subsec:reachability}.

\section{Static Analysis of EVM Bytecode}
\label{sec:stat}
Starting from the small-step semantics presented in \autoref{subsec:smallstep}, we design a sound reachability analysis that supports (among others) the validation of the single-entrancy property.
We follow the verification chain depicted in \autoref{fig:formal-verification-chain}:
For showing the executions of a contract to satisfy some property $\Phi$, we
formulate a \emph{\hornc-clause based abstraction} that abstracts the contract execution behavior and argue about an \emph{abstracted property} over \emph{abstract executions} instead. This reasoning is sound given that all concrete small-step executions are modeled by some abstract execution and given that the abstracted property over-approximates $\Phi$.

A \hornc-clause based abstraction for a small-step semantics $\rightarrow$ is characterized by an abstraction function $\configabs{\predsig}$ that translates concrete configurations (here \configpic) into \emph{abstract configurations} (here \absconfigpic). Abstract configurations are sets of predicate applications where predicates (formally characterized by their signature $\predsig$) range over the values from abstract domains. These abstract arguments are equipped with an order $\ord{}$ that can be canonically lifted to predicates and further to abstract configurations, hence establishing a notion of precision on the latter. Intuitively, $\configabs{\predsig}$ translates a concrete configuration into its most precise abstraction.
The \emph{abstract semantics} is specified by a set of Constrained \hornc clauses  $\abss$ over the predicates from $\predsig$ and describes how abstract configurations evolve during abstract execution. A Constrained \hornc clause is a logical implication that can be interpreted as an inference rule for a predicate, consequently an abstract execution consists of logical derivations from an abstract configuration using $\abss$.
A \hornc-clause based abstraction constitutes a sound approximation of small-step semantics $\rightarrow$ if every concrete (multi-step) execution $\configpic \rightarrow^* \configpic'$ can be simulated by an abstract execution: More precisely, from the abstract configuration $\alpha(\configpic)$ one can logically derive using $\abss$ an abstract configuration $\absconfigpic$ that constitutes an over-approximation of $\configpic'$ (so is at least as abstract as $\configabs{}(\configpic')$). A formal presentation of the soundness statement is given in~\autoref{subsec:soundness} while a characterization in abstract interpretation terminology is deferred to \autoref{sec:appendix-key}.
A sound abstraction allows for the provable analysis of \emph{reachability properties}: Such properties can be expressed as sets of problematic configurations (here \badconfigpic). Correspondingly, a sound abstraction for such a property is a set of bad abstract configurations (here \badabsconfigpic) which contains all possible over-approximations of the bad concrete states. The soundness of the abstract semantics then guarantees that if no bad abstract configuration from this set can be entered, also no bad configuration can be reached in the concrete execution.

\begin{figure}
\includegraphics[width=\columnwidth]{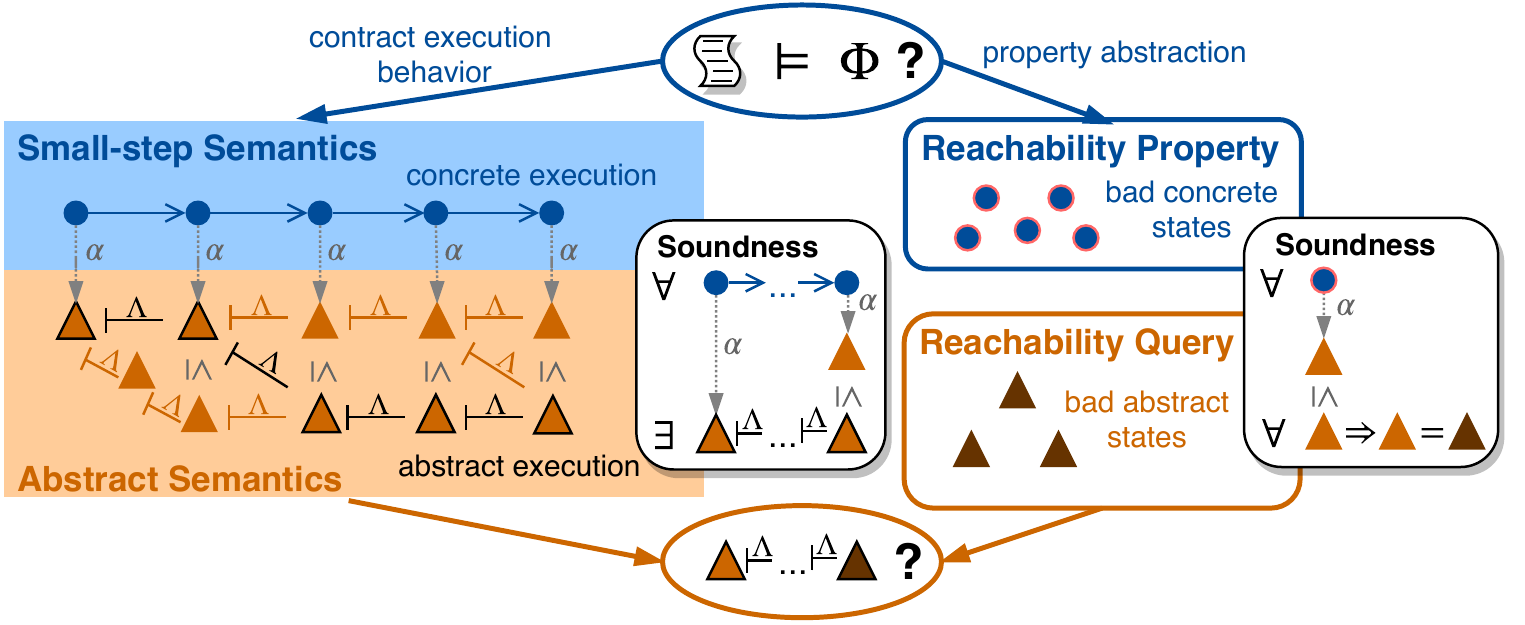}
\caption{Formal verification chain of \ethor{}. \normalfont $\absconfig \xvdash{\abss} \absconfig'$ denotes that the abstract configuration $\absconfig'$ can be logically derived from $\absconfig'$ (within one step) using the \hornc clauses in $\abss$.}
\label{fig:formal-verification-chain}
\end{figure}

\subsection{Main Abstractions}
Our analysis abstracts from several details of the original small-step semantics. In the following we overview the main abstractions:

\myparagraph{Blockchain environment}
The analysis describes the invocation of a contract (in the following denoted as $\cstar$) in an arbitrary blockchain environment, hence is not modeling the execution environment as well as large fractions of the global state.
Indeed, most of this information is not statically known as the state of the blockchain at contract execution cannot be reliably predicted.
As a consequence, the analysis has to deal with a high number of unknown environment inputs in the abstract semantics. Most prominently, the behavior of other contracts needs to be appropriately over-approximated, which turns out to be particularly challenging since such contracts can interact with $\cstar$ in multitudinous ways.

\myparagraph{Gas modelling}
The contract gas consumption is not modeled. The gas resource, which is meant to bound the contract execution, is set by the transaction initiator and hence not necessarily known at analysis time. For this reason, our analysis takes into account that a contract execution might exceptionally halt at any point due to an out-of-gas exception, which does not affect the precision of the analysis for security properties that consider arbitrary contract invocations (and hence arbitrary gas limits).

\myparagraph{Memory model}
In the EVM the local memory is byte-indexed, while the machine stack holds words (encompassing 32 bytes). Consequently loading a machine word from memory requires to assemble the byte values from 32 consecutive memory cells.
However, as already described in~\cite{park2018formal}, in practice reasoning about this conversion between words and bytes is hard. Therefore, we model memory in our abstraction as a word array: this enables very cheap accesses in case that memory is accessed at the start of a new memory word, and otherwise just requires the combination of two memory words.

\myparagraph{Callstack}
The callstack is captured by a two-level abstraction distinguishing only between  the original execution of $\cstar$ (call level $0$) and reentrancies of $\cstar$ ultimately originating from the original execution (call level $1$). This abstraction reflects that given the unknown blockchain environment, the state of the callstack when reentering is obscure: it is unknown who initiated the reentering call and which other internal transactions have been executed before.

\subsection{Analysis Definition}
\label{subsec:analysis-definition}
In the following we formally specify our analysis by defining the underlying \hornc-clause based abstraction.
\begin{figure}
\centering
\small
\begin{mathpar}
\begin{array}{rclr}
  \predsig_{\cstar} &\ni \pred & \define \\
  & | & \pnmstate{\pc} :
 (\NN \times (\NN \to \absdom)) \times (\NN \to \absdom) \times (\NN \to \absdom) \times \BB \to \BB  \\
  & | &\pnexception: \BB \to \BB\\
  & | &\pnhalt: (\NN \to \absdom) \times \BB \to \BB \\
&  \multicolumn{2}{l}{\pc \in \{0, \dots, \size{\access{\cstar}{\code}}-1 \}} \\
   \absdom & \define & \NN \cup \{ \top \} 
\end{array}
\begin{array}{llcl}
\end{array}
\end{mathpar}
  \caption{Definition of the predicate signature $\predsig_{\cstar}$ and the abstract domain $\absdom$.
  \normalfont
}
\label{fig:pred-signature}
\end{figure}
An abstract configuration is a set of predicate applications representing one or several concrete configurations. Since we are interested in analyzing executions of the contract $\cstar$, we consider EVM configurations representing such executions, which are call stacks having an execution state of contract $\cstar$ as a bottom element. 
We abstract such a call stack by the set of all its elements that describe executions of $\cstar$, reflecting the stack structure only by indicating whether a relevant execution state represents the original execution of $\cstar$ (call level $0$) or a reentering execution that hence appears higher on the call stack (call level $1$). The individual execution states are abstracted as predicate applications using the predicates listed in~\autoref{fig:pred-signature}: A predicate application of the form $\predmstate{\pc}{\siz}{\stackv}{\mem}{\stor}{\cl}$ describes a regular execution of $\cstar$ at program counter $\pc$ that has a local stack of size $\siz$ with elements as described by the mapping $\stackv$ (from stack positions to elements) and a local memory $\mem$, and the global storage of contract $\cstar$ at this point being $\stor$. 
Accordingly, the predicate application $\predexception{\cl}$ denotes that an execution of $\cstar$ exceptionally halted on call level $\cl$ and $\predhalt{\stor}{\cl}$ represents an execution that halted regularly on call level $\cl$ with the global storage of $\cstar$ being $\stor$.
Since during the abstract execution, a precise modeling of all the described state components is not always possible, the argument domains of the predicates encompass the abstract domain $\absdom$ that enriches $\NN$ with the join element $\top$ over-approximating any natural number. Formally, the described abstractions of EVM configurations are captured by the abstraction function $\configabs{}$ in \autoref{fig:absfun} that maps call stacks into the corresponding sets of predicates, yielding an abstract configuration.

Note that $\configabs{}$ is parametrized by $\cstar$ and that only the callstack elements modeling executions of $\cstar$ are translated.
\begin{figure*}[!t]
\small
\begin{align*}
\configabs{\cstar}_{\cstar}(\callstack) &:=
\begin{cases}
\emptyset & \callstack = \nil \\
\exstateabs(s, \access{\cstar}{\addr}, \cl) \cup \configabs{}_{\cstar}(\callstack') & \callstack = \cons{s_{\cstar}}{\callstack'} ~\land~ \cl = \cond{(\callstack' = \nil)}{0}{1} \\
\configabs{}_{\cstar}(\callstack')  & \callstack=\cons{s_{c}}{\callstack'} ~\land~ c \neq \cstar \\
\end{cases}
  &\\
\exstateabs(s, a, \cl) &:=
\begin{cases}
\{ \predmstate{\lpc}{\size{\stackv}}{\stacktoarray(\stackv)}{\towordmem(\memv)}{\access{\gstate(a)}{\stor}}{\cl} \} & s = ((\lgas, \lpc, \memv, \actwv, \stackv), \exenv, \gstate, \transeffects) \\
\{ \predexception{\cl} \} & s = \excstate \\
\{ \predhalt{\access{\gstate(a)}{\stor}}{\cl} \} & s = \haltstate{\gstate}{\lgas}{\datav}{\transeffects} \\
\emptyset & \textit{otherwise} 
\end{cases}
  &\\
\stacktoarray(s) &:=
\begin{cases}
\fun{x}{0} & s = \nil \\
\store{(\stacktoarray(s'))}{\size{s'}}{x} & s = \cons{x}{s'}
\end{cases}
  \hspace{0.8cm}
\towordmem(\mem) \define
\fun{x}{\select{\mem}{x \cdot 32} \aconcat{1} \select{\mem}{x \cdot 32 + 1} \cdots \aconcat{1} \select{\mem}{x \cdot 32 + 31}}
\end{align*}
\caption{Configuration abstraction function.
\normalfont
Here $v \aconcat{n} w$ denotes the value obtained by concatenating $v$'s and $w$'s byte representation, assuming that $w$ is represented by $n$ bytes.
}
\label{fig:absfun}
\end{figure*}
The transitions between abstract configurations (as yielded by $\configabs{}$) are described by an abstract semantics in the form of Constrained \hornc clauses. The abstract semantics is also specific to the contract $\cstar$: Depending on the EVM instructions that appear in $\cstar$, it contains \hornc clauses that over-approximate the execution steps enabled by the corresponding instructions. We hence formulate the abstract semantics as a function $\delta$ that maps a contract $\cstar$ to the union over the \hornc clauses that model the individual instructions in the contract:
{\small
\begin{align*}
\delta(\cstar) := \bigcup_{0 \leq i < \size{\access{\cstar}{\code}}}{\instabs{\access{\cstar}{\code}[i]}_i}
\end{align*}}
The core of the abstract semantics is defined by the \emph{instruction abstraction function} $\instabs{\cdot}_i$ that maps a contract instruction at position $i$ to a set of \hornc clauses over-approximating the semantics of the corresponding instruction. We will discuss the translation of the $\ADD$, $\MLOAD$, and $\CALL$ instruction depicted in~\autoref{fig:absrules} to illustrate the main features of the abstract semantics. 

\begin{figure*}[!t]
{\small
\begin{minipage}{\textwidth}
\begin{align*}
\instabs{\ADD}_\pc
\define \{ \,
& \predmstate{\pc}{\siz}{\stackv}{\mem}{\stor}{\cl}
~\land~ \siz > 1
~\land~ \abs{x} = \select{\stackv}{\siz-1}
~\land~ \abs{y} = \select{\stackv}{\siz-2} \\
&\implies
\predmstate{\pc +1}{\siz-1}{\update{\stackv}{\siz-2}{\abs{x} \absop{+} \abs{y}}}{\mem}{\stor}{\cl},  \tag{A1} \label{rule:add-1}\\
& \predmstate{\pc}{\siz}{\stackv}{\mem}{\stor}{\cl}
\implies \predexception{\cl} \, \}  \label{rule:add-2}  \tag{A2} \\
\instabs{\MLOAD}_\pc
\define \{ \,
&\predmstate{\pc}{\siz}{\stackv}{\mem}{\stor}{\cl}
~\land~ \siz > 1
~\land~ \abs{o} = \select{\stackv}{\siz-1}
~\land \abs{v} = \cond{(\abs{o} \in \NN)}{\accessword{\mem}{\abs{o}}}{\top} \nonumber \\
&\implies
\predmstate{\pc +1}{\siz}{\update{\stackv}{\siz-1}{\abs{v}}}{\mem}{\stor}{\cl}, \dots \} \label{rule:mload-1} \tag{M1}
\\
\instabs{\CALL}_\pc
\define \{ \,
&\predmstate{\pc}{\siz}{\stackv}{\mem}{\stor}{\cl}
~\land~ \siz > 6
\implies
\predmstate{\pc +1}{\siz-6}{\update{\stackv}{\siz-7}{\top}}{\arrayinit{\top}}{\arrayinit{\top}}{\cl}, \label{rule:call-1} \tag{C1} \\
&\predmstate{\pc}{\siz}{\stackv}{\mem}{\stor}{\cl}
~\land~ \siz > 6
\implies
\predmstate{0}{0}{\arrayinit{0}}{\arrayinit{0}}{\stor}{1} , \label{rule:call-2} \tag{C2} \\
&\predmstate{\pc}{\siz}{\stackv}{\mem}{\stor}{\cl}
~\land~ \siz > 6
~\land~ \predhalt{\stor_h}{1}
\implies
\predmstate{0}{0}{\arrayinit{0}}{\arrayinit{0}}{\stor_h}{1} , \mydots \} \label{rule:call-3} \tag{C3}
\end{align*}
\end{minipage}
}
\caption{Partial definition of $\instabs{\cdot}_\pc$: selection of abstract semantics rules.
\normalfont
For $\MLOAD$ and $\CALL$ the exception rule is omitted.
}
\label{fig:absrules}
\end{figure*}

\subsubsection{Addition}
\label{subsubsec:addition}
The abstract semantics of the addition instruction ($\ADD$) encompasses two \hornc clauses describing the successful execution and the failure case.
A prerequisite for a successful addition is the existence of a sufficient amount of arguments on the machine stack. In this case, the top stack values are extracted and the stack at the next program counter (modeled by the predicate $\pnmstate{\pc+1}$) is updated with their sum.
As the stack elements, however, range over the abstract value domain $\absdom$, the addition operation on $\NN$ needs to be lifted to $\absdom${}: Following the general intuition of $\top$ representing all potential values in $\NN$, the occurrence of $\top$ as one of the operands immediately declassifies the result to $\top$.
Similar liftings are performed for all unary, binary and comparison operators in the instruction set.
A precise definition is given in \autoref{sec:appendix-ana-def}.

In accordance to the choice of not modeling gas consumption, the \hornc clause modeling, the failure case -- which is common to the abstract semantics of all instructions -- does not have any preconditions, but the instruction reachability. This rule subsumes all other possible failure cases (such as stack over- and underflows).

\subsubsection{Memory Access}
Memory access on the level of EVM bytecode is enabled by the $\MLOAD$ instruction which takes the memory offset to be accessed as argument from the stack and pushes instead the word from the memory starting at this index. In our abstraction defined by the abstract semantic rule depicted in ~\autoref{fig:absrules} either immediately $\top$ is pushed to the stack (in case that the offset $\abs{o}$ is not a concrete value and hence the value to be loaded cannot be determined) or the word from the concrete memory offset is extracted. The extraction needs to account for the word-indexed memory abstraction that we chose and is formally defined by the function $\accessword{\cdot}{\cdot}$ depicted in~\autoref{fig:accessword}.
In case that the offset is a word address (divisible by $32$), the corresponding value can be accessed from the word memory $\mem$ by converting the byte address to the word address ({\footnotesize $\frac{p}{32}$}). Otherwise, the word at the next lower word address ({\footnotesize $\left \lfloor \frac{p}{32}\right \rfloor$}) and the word at the next higher byte address ({\footnotesize$\left \lceil \frac{p}{32}\right \rceil$}) are accessed to combine their relevant parts to a full word.

\begin{figure*}[!t]
{\small
\begin{align*}
\accessword{\mem}{p}
&\define
\begin{cases}
\select{\mem}{\frac{p}{32}} & p \bmod{32} = 0 \\
	(\extract{\select{\mem}{\left \lfloor \frac{p}{32}\right \rfloor}}
	{p \bmod 32}
	{31})
	\aconcat{p}
 (\extract{\select{\mem}{\left \lceil \frac{p}{32}\right \rceil}}
	{0}
	{(p \bmod 32) -1})
& \text{otherwise}
\end{cases}
\end{align*}}
\caption{Function extracting the word at byte offset $p$ from word-indexed memory $\mem$.
\normalfont
Here $\extract{v}{l}{r}$ denotes the value represented by $v$'s $l$th byte till $r$th byte in big endian byte representation. $v \aconcat{n} w$ is defined as in \autoref{fig:absfun}. We assume both operations to be lifted to $\absdom$.}
\label{fig:accessword}
\end{figure*}

\subsubsection{Contract Calls}
The abstraction for $\CALL$ is the most interesting. This instruction takes seven arguments from the stack that specify parameters to the call such as the target of the call or the value to be transferred along with the memory addresses specifying the location of the input and the return data.
When returning from a successful contract call, the value $1$ is written to the stack and the return value is written to the specified memory fragment. The persistent storage after a successful call contains all changes that were performed during the execution of the called contract. In the case that the contract call terminated exceptionally instead, the storage is rolled back to the point of calling and the value $0$ is written to the stack to indicate failure.

Since a contract $\CALL$ initiates the execution of another (unknown) contract, all its effects on the executions of $\cstar$ need to be modeled. More precisely, these effects are two-fold: the resuming execution of $\cstar$ on the current call level needs to be approximated, as well as the reentering executions of $\cstar$ (on a higher call level). For obtaining an analysis that is precise enough to detect real-world contracts with reentrancy protection as secure, it is crucial to model $\cstar$'s persistent  storage as accurately as possible in reentering executions. This makes it necessary to carefully study how the storage at the point of reentering relates to the one in the previous executions of $\cstar$, taking into account that (in the absence of $\DELEGATECALL$ and $\CALLCODE$ instructions in $\cstar$) only $\cstar$ can manipulate its own storage.
\autoref{fig:call-abstraction} overviews the storage propagation in the case of a contract call: To this end it shows the sequence diagram of a concrete execution of $\cstar$ that calls a contract $c'$ which again triggers several reentrancies of $\cstar$. In this course three ways of storage propagation between executions of $\cstar$ are exhibited:
1) The storage is \emph{forward propagated} from a calling execution to a reentering execution of $\cstar$ (\absruletwo{A}, \absruletwo{C})
2) The storage is \emph{cross propagated} from a finished reentering execution to another reentering execution of $\cstar$ (\absrulethree{B})
3) The storage is \emph{back propagated}  from a finished reentering execution to a calling execution of $\cstar$ (\absruleone{D}, \absruleone{E})
These three kinds of propagation are reflected in the three abstract rules for the call instruction given in \autoref{fig:absrules} and correspondingly visualized in \autoref{fig:call-abstraction}.

Rule (C1) describes how the execution of $\cstar$ (original and reentering alike) resumes after returning from the call, and hence approximates storage back propagation:
For the sake of simplicity, storage gets over-approximated in this case by $\fun{x}{\top}$. The same applies to the local memory and stack top value since those are affected by the result of the computation of the  unknown contract.
Rule (C2) captures the initiation of a reentering execution (at call level $1$) with storage forward propagation: As  contract execution always starts at program counter $0$ with empty stack and zeroed-out local memory, only abstractions (instances of the $\pnmstate{0}$ predicate) of this shape are implied. The forward propagation of storage is modeled by initializing the $\pnmstate{0}$ predicate with the storage $\stor$ at call time.
Rule (C3) models storage cross propagation: Similar to rule (C2), an abstract reentering execution in a fresh machine state is triggered. However, the storage is not propagated from the point of calling, but from another finished reentering execution whose results are abstracted by the halting predicate $\pnhalt$ at call level $1$. This rule is independent of the callee in that it is only conditioned on the reachability of some $\CALL$ instruction, but it does not depend on any of the callee's state. Instead its cyclic structure requires to extrapolate an invariant on the potential storage modifications that are computable by $\cstar$: Intuitively, when reentering $\cstar$ it needs to be considered that priorly the storage was modified by applying an arbitrary sequence of $\cstar$'s public functions. The significance of this abstraction is motivated by the example in \autoref{fig:reentrancy} where the attack is only enabled by calling \lstinline|Bank|'s \lstinline|release| function first, to release the lock before reentering.

\begin{figure}
\includegraphics[width=0.6\columnwidth]{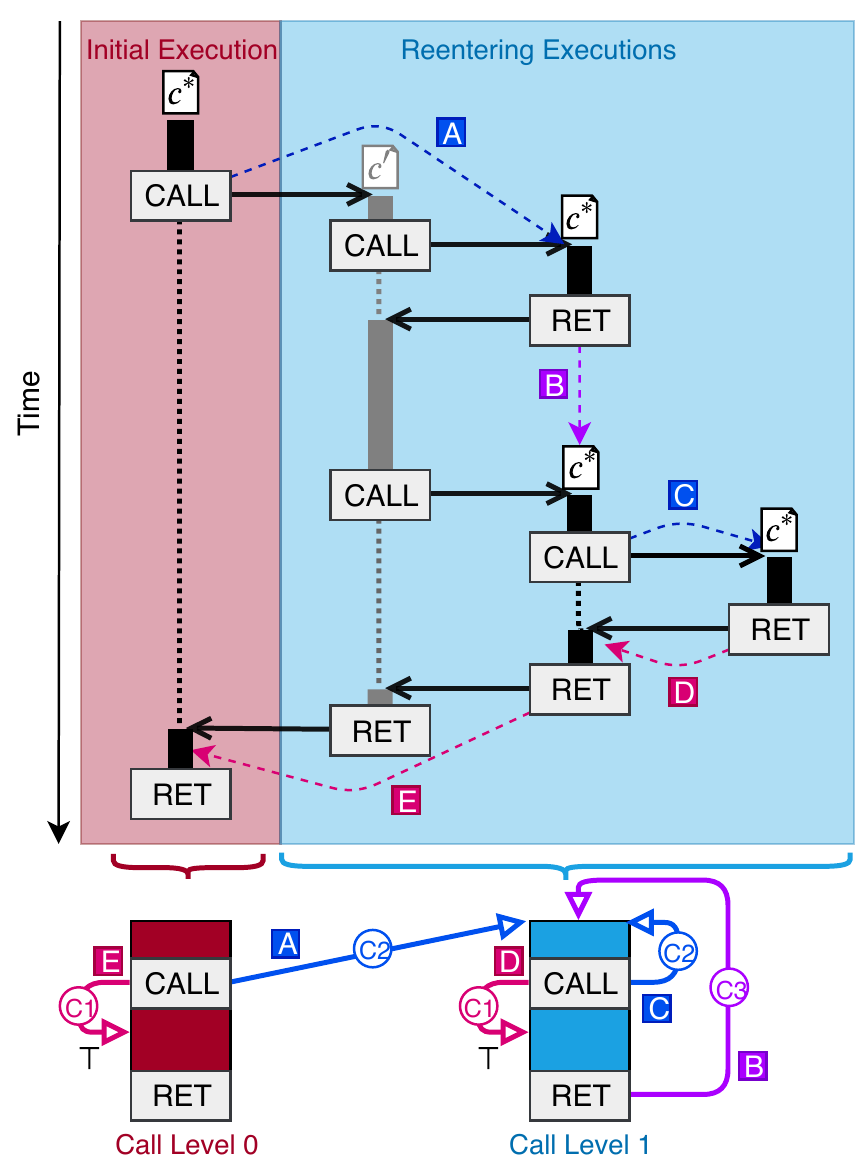}
\caption{Illustration of the different call abstractions.}
\label{fig:call-abstraction}
\end{figure}

\subsection{Scope of the analysis}
\label{subsec:scope}
Before presenting the soundness result, we discuss the scope of the analysis.
The analysis targets contracts in a stand-alone setting, which means that the behavior of all contracts that $\cstar$ might interact with is over-approximated.
This abstraction is not merely a design choice, but rather a necessity as the state of the blockchain (including the code of the contracts residing there) at execution time cannot be statically determined.
Still, we could easily accommodate the precise analysis of a set of known contracts e.g., library contracts that are already present on the blockchain. We omitted this straightforward extension in this work for the sake of clarity and succinctness in the analysis definition and the soundness claim.

Following this line of argumentation, we assume $\cstar$ not to contain $\DELEGATECALL$ and $\CALLCODE$ instructions: these instructions enable the execution of another contract code in the context of  $\cstar$, allowing for the modification of the persistent storage of $\cstar$ and even of money transfers on behalf of $\cstar$. Using $\DELEGATECALL$ or $\CALLCODE$ to call an unknown contract can therefore potentially result in the reachability of any execution states of contract $\cstar$. Consequently every property relying on the non-reachability of certain problematic contract states would be considered violated.
In a setting of multiple known contracts the restriction on $\DELEGATECALL$ and $\CALLCODE$ instructions could be relaxed to allow for  such calls that are guaranteed to target known contracts.

We now briefly illustrate the key design choices behind our abstraction, which we carefully crafted   to find the sweet spot between accuracy and practicality. The analysis is value sensitive in that concrete stack, memory, and storage values are tracked until they get abstracted due to influence of unknown components.
For local computations, the analysis is partly flow-sensitive (considering the order of instructions, but merging abstract configurations at the same program counters) and path-sensitive (being sensitive to branch conditions).
On the level of contract calls, a partial context sensitivity is given in that the storage at the time of calling influences the analysis of the subsequent call, but no other inputs to the call are tracked. In particular (due to the lack of knowledge on interactions with other contracts) all reentering calls are merged into a single abstraction, accumulating all possible storage states at the point of reentering. For this reason, the analysis of calls on level $1$ is  less precise than the one of the original execution on call level $0$, where only the restrictions of flow sensitivity apply.

\subsection{Soundness Result}
\label{subsec:soundness}

We prove, for each contract $\cstar$, that the defined \hornc-clause based abstraction soundly over-approximates the small-step semantics presented in \autoref{subsec:smallstep}. Formally, this property is stated as follows:
\begin{theorem}[Soundness]
\label{theorem:soundness}
Let $\cstar$ be a contract  whose code does not contain $\DELEGATECALL$ or $\CALLCODE$. Let $\transenv$ be a transaction environment and let $\callstack$ and $\callstack'$ be annotated callstacks such that $\size{\callstack'} > 0$. Then for all execution states $\exstate$ that are strongly consistent with $\cstar$ it holds that
\par
\nobreak
{\small
\noindent
\begin{align*}
\ssteps{\transenv}{\cons{\annotate{\exstate}{\anacontract}}{\callstack}}{\concatstack{\callstack'}{\callstack}}
\implies
\forall \absconfig_I.~ \configabs{\predsig}_{\cstar}([\annotate{\exstate}{\cstar}]) \leq \absconfig_I \\
\implies
\exists \absconfig.~ \absconfig_I, \delta(\cstar) \derives \absconfig \ ~\land~ \configabs{\predsig}_{\cstar}(\callstack') \leq \absconfig
\end{align*}
}%
\end{theorem}

The theorem states that every execution of contract $\cstar$ (modeled by a multi-step execution starting in state $\annotate{\exstate}{\cstar}$ on an arbitrary call stack $\callstack$ and ending in call stack $\concatstack{\callstack'}{\callstack}$, indicating that the original execution of $\cstar$ yielded the state as modeled by the call stack $\callstack'$), can be mimicked by an abstract execution. This means that from every abstract configuration $\absconfig_I$ that abstracts $\annotate{\exstate}{\cstar}$ (so that it is more abstract than $\configabs{\cstar}([\annotate{\exstate}{\cstar}])$) one can logically derive using the \hornc clauses in $\delta(\cstar)$ some abstract configuration $\absconfig$ abstracting $\callstack'$. As a consequence of this theorem we can soundly reason about arbitrary executions of a contract $\cstar$: if we can show that from an abstract configuration $\absconfig_I$, that abstracts a set of initial execution states of $\cstar$, it is impossible to derive using $\delta(\cstar)$ some other abstract configuration $\absconfig$, that abstracts a set of problematic execution states of $\cstar$, then this ensures that all these problematic states are not reachable with a small-step execution from any of the initial states. 

For the proof of ~\autoref{theorem:soundness} we refer the reader to~\autoref{app:theory}.

\subsection{Reachability Properties for Contract Safety}
\label{subsec:reachability}

As characterized by the soundness result, our abstraction allows for the sound analysis of reachability properties. We will illustrate in the following how such a reachability analysis is sufficient to express relevant smart contract security properties.

\subsubsection{Single-entrancy}
Some generic security properties of Ethereum smart contracts can be over-approximated by reachability properties and hence checked automatically by our static analysis.
Consider,  the single-entrancy property from \autoref{subsec:props} which has been proven to be approximated by the following reachability property in~\cite{GMS::CAV18}.

\begin{definition}[Call unreachability~\cite{GMS::CAV18}]
\label{def:callunreachability}
A contract $c$ is call unreachable if for all regular execution states $\regstatefull{\mstate}{\exenv}{\gstate}{\transeffects}$ that are strongly consistent with $c$ and satisfy $\mstate = \smstate{g}{0}{\emptymemory}{0}{\emptystack}$ for some $g \in \NN$, it holds that for all transaction environments $\transenv$ and all callstacks $\callstack$
\par
\nobreak
{\small
\noindent
\begin{align*}
\neg \exists s, \callstack. &\,
\ssteps{\transenv}{\cons{\annotate{\regstatefull{\mstate}{\exenv}{\gstate}{\transeffects}}{c}}{\callstack}}{\cons{\annotate{s}{c}}{\concatstack{\callstack'}{\callstack}}} \\
&~\land~ \size{\callstack'} > 0 ~ \land ~
\arraypos{\getcontractcode{c}}{\access{\access{s}{\mstate}}{\pc}} \in \callinstructions
\end{align*}
}%
Where the set $\callinstructions$ of call instructions is defined as
\par
\nobreak
{\small
\noindent
\[\callinstructions =  \{\CALL, \CALLCODE, \allowbreak \DELEGATECALL, \CREATE \}\]
}%
\end{definition}

Intuitively, call reachability is a valid over-approximation of single-entrancy as an internal transaction can only be initiated by the execution of a call instruction. Consequently, for excluding that an internal transaction was initiated after reentering, it is sufficient to ensure that no call instruction is reachable at this point. In addition, as all contracts start their executions in a fresh machine state (program counter  and active words set to $0$, empty stack, memory initialized to $0$) when being initially called, it is sufficient to check all executions of contract $c$ that started in such a state.

\subsubsection{Static assertion checking} The Solidity language supports the insertion of assertions into source code. Assertions shall function as pure sanity checks for developers and are enforced at runtime by the compiler creating the corresponding checks on the bytecode level and throwing an exception (using the $\INVALID$ opcode) in case of an assertion violation~\cite{solidity-doc-assert}.
However, adding these additional checks creates a two-fold cost overhead: At create time a longer bytecode needs to be deployed (the longer the bytecode the higher the gas cost for creation) and at call time the additional checks need to be executed which results in additional gas consumption.
With our static analysis technique, assertions can be statically checked by querying for the reachability of the $\INVALID$ instructions. If no such instruction is reachable, by the soundness of the analysis, the code is proven to give the same guarantees as with the assertion checks (up to gas) and those checks can safely be removed from the code resulting in shorter and cheaper contracts.\footnote{The Solidity Docs~\cite{solidity-doc-assert} discuss exactly this future use of static analysis tools for assertion checking.}
Formally, we can characterize this property as the following reachability property:
\begin{definition}[Static assertion checking]
Let $c$ be a contract and $\regstatefull{\mstate}{\exenv}{\gstate}{\transeffects}$ regular execution states such that $\regstatefull{\mstate}{\exenv}{\gstate}{\transeffects}$ is strongly consistent with $c$ and $\mstate = \smstate{g}{0}{\emptymemory}{0}{\emptystack}$ for some $g \in \NN$. Let $\transenv$ be an arbitrary transaction environment and $\callstack$ be an arbitrary callstack. Then a the static assertion check for $c$ is defined as follows:
\par
\nobreak
{ \small
\noindent
\begin{align*}
\neg \exists s, \callstack. & \,
\ssteps{\transenv}{\cons{\annotate{\regstatefull{\mstate}{\exenv}{\gstate}{\transeffects}}{c}}{\callstack}}{\cons{\annotate{s}{c}}{\concatstack{\callstack'}{\callstack}}}
~ \land ~
\arraypos{\getcontractcode{c}}{\access{\access{s}{\mstate}}{\pc}} = \INVALID
\end{align*}
}%
\end{definition}
Intuitively this property says that during an execution of contract $c$ it should never be possible to execute an $\INVALID$ instruction.

\subsubsection{Semi-automated verification of contract-specific properties}
As demonstrated by Hildebrandt et al.~\cite{hildenbrandt2017kevm}, reachability analysis can be effectively used for Hoare-Logic-style reasoning. This holds in particular for the analysis tool presented in this work: Let us consider a Hoare Logic triple $\{ P \} \textit{C} \{ Q\}$ where $P$ is the precondition (operating on the execution state), $\textit{C}$ is the contract code and $Q$ is the postcondition that should be satisfied after executing code $\textit{C}$ in an execution state satisfying $P$. Then we can intuitively check this claim by checking that a state satisfying $\neg Q$ can never be reached when starting execution in a state satisfying $P$.
More formally, we can define Hoare triples as reachability properties as follows:
\begin{definition}[Hoare triples]
Let $\cstar$ be a contract and let $\textit{C}$ be a code fragment of $\cstar$.
Let $P \in \exstates \to \BB$ be a predicate on execution states (strongly consistent with $\cstar$) that models execution right at the start of  $\textit{C}$ and similarly let $Q\in \exstates \to \BB$ be  a predicate on execution states (strongly consistent with $\cstar$) that models execution right at the point after executing  $\textit{C}$.
Then Hoare triples $\{ P \} \textit{C} \{ Q\}$ can be characterized as follows:
\par
\nobreak
{ \small
\noindent
\begin{align*}
\{ P \} \textit{C} \{ Q\}
\define
\forall \exstate. \, P(\exstate)
\implies
\neg \exists \exstate'. & \,
\ssteps{\transenv}{\cons{\annotate{\exstate}{\cstar}}{\callstack}}{\cons{\annotate{s'}{\cstar}}{\callstack}} ~\land~  \neg Q(\exstate')
\end{align*}
}%
\end{definition}

Hoare-Logic style reasoning can be used for the semi-automated verification of smart contracts given that their behavior is specified in terms of pre- and postconditions. For now it still requires a non-negligible amount of expertise to insert the corresponding abstract conditions on the bytecode-level, but by a proper integration into the Solidity compiler the generation of the initialization and reachability queries could be fully automated (cf. \autoref{subsec:queries}). We want to stress that in contrast to existing approaches, our{} analysis technique has the potential to provide fully automated pre- and postcondition checking even in the presence of loops as it leverages the fixed point engines of state-of-the-art SMT solvers~\cite{hoder2011muz}.

\section{\horst{}: A Static Analysis Language}

\label{sec:horst}
To facilitate the principled and robust development of static analyzers based on \hornc clause resolution, we designed \horst{} -- a framework consisting of a high-level specification language for defining \hornc-clause based abstractions and a compiler generating optimized \texttt{smt-lib} encodings for SMT-solvers.
The objective of \horst{} is to assist analysis designers in developing fast and robust static analyzers from clean and readable logical specifications.

Many existing practical analyzers are built on top of modern SMT-solvers such as \zz{}.
These solvers are highly optimized for performance, which causes big performance deviations on different problem instances and makes their internal workings (due to the heavy use of heuristics) opaque to the user. Handcrafting logical specifications for such solvers in their low-level input format \texttt{smt-lib} is hence not only cumbersome, error-prone, and requires technical expertise, but is also very inflexible, since the performance effects of different encodings may vary with the concrete problem instance. For tackling this issue, \horst{} decouples the high-level analysis design from the compilation to the input format: A high-level specification format allows for clear, human-readable analysis definitions while the translation process is handled by a stable and streamlined backend. On top, it allows for easily applying and experimenting with different \hornc-clause level optimizations that we can show to enhance the performance of \zz{} substantially in our problem domain. We will shortly illustrate the utilization of \horst{} in the design process of our static analyzer and discuss the most interesting optimizations performed by the \horst{} compiler.
For an introduction to the \horst{} language, we refer the reader to \autoref{sec:appendix-horst}.

\paragraph*{Designing static analyses using \horst{}}
The \horst{} language allows for writing math-like specifications of \hornc{} clauses such as those given in~\autoref{fig:absrules}. For parametrizing those clauses (e.g., by the program counters of a specific contract) an interface with a \java{} back-end can be specified that handles the domain specific infrastructure, such as contract parsing.
We overview the different steps of the analysis design process in~\autoref{fig:horst-design}.

\begin{figure}[t]
\includegraphics[width=\columnwidth]{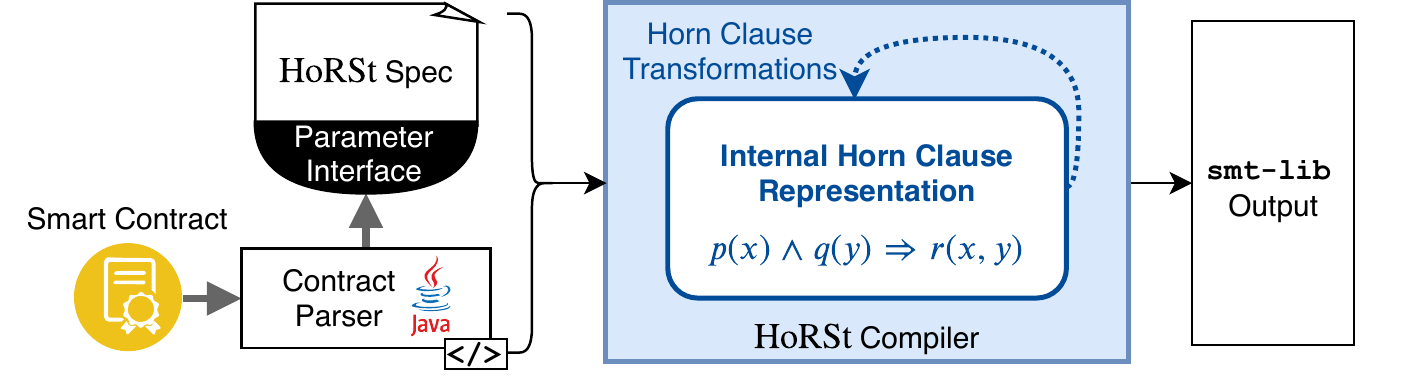}
\caption{Utilization of \horst{} for static analysis}
\label{fig:horst-design}
\end{figure}

The core of the analysis is the \horst{} specification. Using high-level programming constructs such as algebraic data types, pattern matching, and bounded iteration, a \horst{} specification describes Constrained \hornc{} clauses over user-defined predicates. \hornc{} clauses can be parametrized by (families) of sets that are specified in the parameter interface (e.g., the sets of all program counters containing a certain bytecode instruction in a specific contract). 
Given such a specification, the analysis designer needs to provide infrastructure code written in \java{}. In particular this code needs to exhibit an implementation of those sets (or functions) specified in the parameter interface. In the case of our analysis, the environment code contains the infrastructure for contract parsing and the parameter interface allows for accessing the assembled contract information (code length, positions of opcodes, etc.) in the analysis specification. 
The \horst{} compiler itself is utilized to generate (optimized) \smtlib{} output given a \horst{} specification and the parameter interface implementation: It unfolds the high-level specification into separate \hornc{} clauses over basic data types, applying the interface implementation: To this end it also resolves all high-level constructs, ensuring that the resulting \hornc{} clauses fall into the fragment that can be handled by \zz{}. On top, the \horst{} compiler (optionally) performs different optimizations and transformations on the resulting \hornc{} clauses, before translating them into the standardized SMT output format \smtlib{}. The most important of these transformations are discussed in the following.

\paragraph*{Low-level optimizations}
\label{subsubsec:folding}
One of the most effective optimizations performed by \horst{} is the predicate elimination by \emph{unfolding} \hornc clauses.
This satisfiability preserving transformation has been long-studied in the literature~\cite{burstall1977transformation,tamaki1984unfold}  and showed  beneficial for solving \hornc clauses in certain settings~\cite{hermenegildo2012overview,bjorner2015horn}.
In practice, however, the exhaustive application of this transformation can lead to an exponential blow-up in the number of \hornc clauses and hence does not necessarily yield the best results. For this reason \horst{} implements different strategies for the (partial) application of this transformation, which we call \emph{linear folding} and \emph{exhaustive folding}.

The idea behind the unfolding transformation is that a predicate $p$ can be eliminated from a set of \hornc clauses $\abss$ by unfolding the occurrences of $p$ in the premises according to the clauses that have $p$ as conclusion. 
An example is given in \autoref{fig:unfolding}. Here predicate $P_2$ is eliminated by merging the two single execution steps (modeled by the two clauses on the left) into a combined clause (on the right) summarizing the steps.

\begin{figure}
{\small
\begin{align*}
  \begin{rcases}
  \begin{aligned}
    P_1(x)\land y = x + 1&\Rightarrow P_2(y) \\
    P_2(y)\land z = y * 3&\Rightarrow P_3(z)
  \end{aligned}
  \end{rcases} P_1(x) \land y = x + 1 \land z = y * 3 \Rightarrow P_3(z)
\end{align*}}
\caption{Unfolding of $P_2$.}
\label{fig:unfolding}
\end{figure}

This intuition serves as a starting point for the unfolding strategy of \emph{linear folding}. In linear folding, all clauses representing a basic block of sequential execution steps are merged into a single clause. More precisely, the unfolding transformation is only applied to those predicates that are used linearly in $\abss$, meaning that $p$ occurs in the premises of exactly one clause in $\abss$ and in the conclusion of exactly one different clause in $\abss$. 
Linear folding has the advantage that it runs linearly in the number of clauses in $\abss$ and yields as result a reduced set of clauses $\abss'$ such that $\size{\abss'} \leq \size{\abss}$. 

In contrast, applying the unfolding transformation exhaustively on all predicates (with exception of those that are recursively used) might yield an exponential blow-up in clauses (and hence also result in exponential runtime). In practice however, the set of clauses $\abss'$ resulting from such a \emph{exhaustive folding} is often of a reasonable size. For mitigating the runtime overhead, however, it is crucial to avoid unnecessary blow-ups in the intermediate clause sets produced during the transformation: To this end, for exhaustive folding \horst{} applies linear folding first and only afterwards performs the unfoldings that multiply existing clauses.

Finally, \horst{} supports constant folding for minimizing the \texttt{smt-lib} output and value encoding to map custom data types into primitive type encodings that are efficiently solvable by \zz. We refer to \autoref{app:horst} for further details on \horst{} internals and functionalities.

\section{Implementation \& Evaluation}
\label{sec:eval}
We use \horst{} to generate the analyzer \ethor{} which implements the static analysis defined in~\autoref{sec:stat}. In the following, we shortly overview the design of \ethor{} and illustrate how \ethor{} can enhance smart contract security in practice. To this end we conduct a case study on a widely used library contract, showing \ethor{}'s capability of verifying  functional correctness properties and static assertion checks. Furthermore, we validate \ethor{}'s soundness and precision on the official EVM testsuite and run a large-scale evaluation for the single-entrancy property on a set of real-world contracts from the Ethereum blockchain, comparing \ethor{} with the state-of-the-art analyzer ZEUS~\cite{kalra2018zeus}.

\subsection{Static Analysis Tool}
The mechanics of \ethor{} are outlined in \autoref{fig:analysis-outline}: \ethor{} takes as input the smart contract  to be analyzed in bytecode-format and a \horst{}-specification parametrized by $\cstar$
.
For enhancing the performance and precision of the tool, \ethor{} performs a multi-staged analysis:
First, it approximates the contract jump destinations (based on a less precise abstract semantics) which helps the tool performance as it decouples the control flow reconstruction (which can be performed more efficiently with a less precise abstract semantics as typically no computations on jump destinations are performed, but just their flow during the stack needs to be modeled) from the more evolved abstract semantics required for precisely analyzing the properties discussed in \autoref{subsec:reachability}. As both used semantics are sound, the soundness of the overall analysis is not affected.
In a second pre-processing step, \ethor{} performs a simple partial execution of atomic program blocks in order to statically determine fixed stack values. This can be beneficial in order to, e.g., precompute hash values and results of exponentiation which would otherwise need to be over-approximated in the analysis due to the lacking support for such operations by \zz{}.
The results from the pre-analysis steps are incorporated into the analysis by a predefined interface in the \horst{}-specification.
The $\horst$ compiler then -- given the interface implementation and the specification -- creates an internal \hornc clause representation which, after optionally performing different optimizations, is translated to an \texttt{smt-lib} file on which the SMT-Solver $\zz$ is invoked.
The reconstructed control flow is obtained by a \souffle{}~\cite{jordan2016souffle} program, which was created by manually translating a \horst{} specification.
\souffle{} is a high performance datalog engine, which we plan to support as a compilation target for (a subset
of) \horst{} in the near future.
Since the problem of control flow reconstruction falls into the fragment supported by modern datalog solvers, we found \souffle{} more performant than using the general-purpose solver \zz{} in this context
\footnote{\zz{} also implements a standard datalog engine which is restricted to work with predicates over finite domains. This constraint is used to ensure that the \smtlib{}-expressible \hornc clauses do not leave the classical datalog-solvable fragment. However, \souffle{} overcomes this restriction in favor of a more liberal characterization of the solvable fragment which could also be incorporated into the \horst{} language - allowing for compilation to \souffle{} from this fragment.}.
 However, for reasoning about more involved properties, the expressiveness of \zz{} is required as we will illustrate in \autoref{subsec:case-study}.

\begin{figure}
\includegraphics[width=.91359\linewidth]{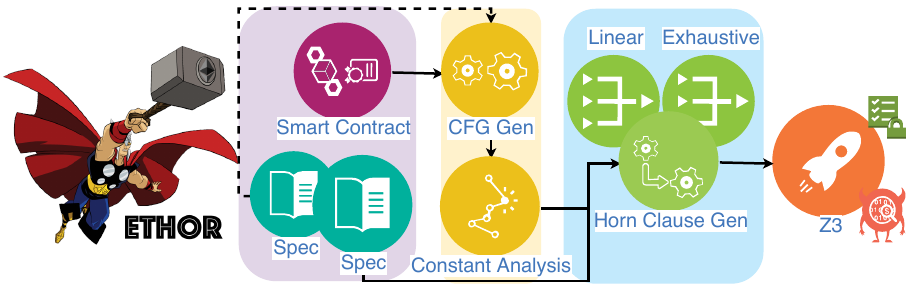}
\caption{Analysis outline.}
\label{fig:analysis-outline}
\end{figure}

\subsection{Case Study: SafeMath Library}
\label{subsec:case-study}

As a case study for functional correctness and assertion checking we chose Solidity's \lstinline|SafeMath| library~\cite{safemath}, a library implementing proper exception behavior for standard arithmetic operations.
This particularly encompasses exceptions in case of overflows in addition and multiplications, underflows in subtractions, and division or modulo by $0$.
The \lstinline|SafeMath| library is special in that it is not a classical library which is deployed as an own contract on the blockchain, but its functions  get inlined during the compilation of a contract that uses them\footnote{In Solidity, one always needs to provide definitions of the (library) contracts one is interacting with. In case that a library is only containing pure internal functions, the Solidity compiler inlines this functions instead of compiling them to \DELEGATECALL{} call instructions to an address at which the user specified the library to reside.}.
This specific behavior makes it particularly interesting to analyze the individual library functions as their concrete implementations may vary with changes in the compiler.

\myparagraph{Functional Correctness}
For our case study we compiled the functions of the \lstinline|SafeMath| library with a recent stable Solidity compiler version (0.5.0) and 
verified that they expose the desired behavior. In particular we showed that all functions 1) cannot return successfully in the problematic corner cases. 2) can return successfully with the correct result in the absence of corner cases. 3) if halting successfully in the absence of corner cases, they can return nothing but the correct result.

As these properties require to precisely relate different input values over the execution (e.g., requiring that the sum of two input values $x$ and $y$ exceeds $2^{256}$), we needed to slightly adapt our analysis by adding a corresponding representation of the initial input (as word array) to the $\pnmstate{}$ and the $\pnhalt$ predicates. This array is accessed by the \CALLDATALOAD{} operation which fetches the input data. Additionally, we need to model return values by an own predicate. For more details, we refer the reader to \autoref{sec:functional-correctness}.
Our analysis manages to prove the corresponding functional properties for each of the five functions within milliseconds, showcasing the effiency of our analysis tool.
Note that  verifying meaningful functional correctness properties, like in this case study, requires to universally quantify over potential inputs, hence making an analysis with a datalog engine (such as \souffle{}), which requires to explicitly list finite initial relations, infeasible.

\myparagraph{Static Assertion Checking}
The following code snippet shows the division function of the \lstinline|SafeMath| library:
\begin{lstlisting}
function div(uint256 a, uint256 b, string memory errorMessage) internal pure returns (uint256) {
    require(b > 0, errorMessage);
    uint256 c = a / b;
    // assert(a == b * c + a % b); // There is no case in which this doesn't hold
    return c; }
\end{lstlisting}
It testifies that the function used to contain an assertion which was deemed to be unnecessary and hence removed (probably to save gas). We reinserted this assertion and indeed could prove that the dynamic assertion check is obsolete as it can never be violated.

\subsection{Large-scale Evaluation}
We performed a series of experiments to assess the  overall performance of our tool. In particular,  we systematically evaluated \ethor{}'s correctness and precision on the official EVM testsuite and additionally conducted a large scale analysis for the single-entrancy property, comparing \ethor{} with the ZEUS~\cite{kalra2018zeus} static analyzer, using the real-world dataset introduced with the latter\footnote{We chose to compare with~\cite{kalra2018zeus} as we found it the only (claimed) sound tool to support a property comparable to single-entrancy.
\cite{securify} only supports a no-write-after-call pattern which the authors claim themselves to be different from reentrancy.
\cite{lu2019neucheck} makes use of a similar pattern.}.

\subsubsection{Automated Testing}
\label{subsec:testing}
For making a principled assessment of its correctness, we evaluated \ethor{} against the virtual machine test cases provided by the Ethereum Foundation\footnote{\url{https://github.com/ethereum/tests/}}. Being formulated as pre- and postconditions, these test cases fall in the class of properties characterized in \autoref{subsec:props} and we could automatically translate them into  \hornc clauses and queries in \horst{}.
The test suite defines 609 test cases, 604 of which specify properties that are relevant for a single contract setting (see~\autoref{subsec:automated-testing-appendix} for details).

Using a $1$ second timeout, we were able to solve $85\%$ (513) of the test cases precisely with a termination rate of $99\%$ (597).

\subsubsection{Reentrancy}
\label{subsec:reentrancy-eval}
For the call unreachability property described in \autoref{def:callunreachability}, we evaluated \ethor{} against the the set of real-world contracts presented in~\cite{kalra2018zeus}. 
The authors extracted $22493$ real-world contracts from the Ethereum blockchain over a period of three months and (after deduplication) made available a list of $1524$ contract addresses.
Due to various problems of this dataset (as described in \autoref{sec:zeus}), sanitization leaves us with  $720$ distinct bytecodes
 out of which we can label $100$ contracts to be trivially non-reentrant (because they did not contain any possibly reentering instruction) and $2$ were out of the scope of our analysis (because they contain at least one \DELEGATECALL{} or \CALLCODE{} instruction) and hence immediately classified to be potentially vulnerable. We make this sanitized benchmark available to the community, including bytecode and sources (where available)~\cite{extended}.
For $13$ contracts we failed to reconstruct the control flow graph, leaving us with $605$ distinct contracts to run our experiments on.

We ran three different experiments for evaluating \ethor{}'s performance for the single-entrancy property: one without performing any \hornc clause folding and two performing each one of the two \hornc clause folding variants described in \autoref{sec:horst}. 
The aim of this experimental set up is not only to conduct a comparison with ZEUS, but also to showcase how \ethor{}'s modular structure facilitates its performance in that \ethor{} can flexibly benefit from different optimization techniques of the \horst{} compiler.
In the comparison with ZEUS, we take into account the combined result of the three different experiments (the contracts solvable using any of the applied transformations).
For the exhaustive folding, we only considered instances where we could generate the \texttt{smt-lib} output in less than 15 minutes.\footnote{This timeout was chosen since it yielded a termination rate of $>95\%$.}
All of the experiments were conducted on a Google Cloud Server with $24$ Cores at $2.8$ GHz and $150$ GiB of RAM.
At most $30$ queries were executed at once, each with a 10 minutes timeout.
Combining the different experiments we were able to obtain results for all but $20$ contracts.

In order to assess the precision of our tool, we compared the results with~\cite{kalra2018zeus}.
Because of the existing unsoundness claims of~\cite{kalra2018zeus} in the literature~\cite{torres2018osiris,GMS::CAV18} we manually reassessed the ground truth provided
by~\cite{kalra2018zeus} for all contracts that were labeled insecure by at least one of the tools.
Since this is a challenging and time consuming task, especially in the case that no Solidity source code is available, we excluded all contracts with more than 6000 bytecodes for which we were not able to obtain the source code, which leaves us with $712$ contracts for which we assessed the ground truth.

Surprisingly, we found numerous contracts labeled non-reentrant by~\cite{kalra2018zeus} which, if analysed in a single contract setting, definitely were reentrant according to the definition of reentracy given in \autoref{def:single-entrancy} and also according to the informal definition provided in~\cite{kalra2018zeus} itself\footnote{\cite{kalra2018zeus} gives the following informal definition: `A function is reentrant if it can be interrupted while in the midst of its execution, and safely re-invoked even before its previous invocations complete execution.'}.
We assume this to be an artefact of~\cite{kalra2018zeus}'s syntactical treatment of the \texttt{call} directive on the Solidity level which is, however, insufficient to catch all possible reentrancies.
As the authors claim to exclude reentrancies introduced by the \texttt{send} directive (even though this is officially considered potentially insecure~\cite{solidity-doc-security}), for the sake of better comparability, we slightly updated our abstract semantics to account for calls that can be deemed secure following the same argument (namely that a small gas budget prevents reentrancy).
In the following we compare \ethor{} against~\cite{kalra2018zeus} on our manually established ground truth. The results are summarized in~\autoref{tab:results}.
\begin{table}
\begin{center}
{\small
\begin{tabular}{llrr}
  Measure & Definition & \ethor{} &\cite{kalra2018zeus} \\
  \hline
  termination & $\textit{terminated}/\textit{total}$ & 95.4 & 98.3 \\
  sensitivity & $tp/(tp + fn)$ & 100 & 11.4 \\
  specificity & $tn/(tn + fp)$ & 80 & 99.8 \\
  F-measure & $2*(\textit{spec}*\textit{sens}/(\textit{spec} + \textit{sens}))$ & 88.9 & 20.4 \\
\end{tabular}}
\end{center}
\caption{Performance comparison of \ethor{} and ZEUS~\cite{kalra2018zeus}.
\label{tab:results}
\normalfont
\emph{total}/\emph{terminated} denotes the total number of contracts in the data set/the number of contracts the respective tool terminated on. \emph{tp}/\emph{fp} denotes the number of true/false positives and \emph{tn}/\emph{fn} the number of true/false negatives.}
\end{table}

For achieving a termination rate comparable to~\cite{kalra2018zeus} ($95.4\%$ vs. $98.3\%$), we needed to run our tool with a substantially higher timeout (10 min. query timeout vs 1 min. contract time out for ZEUS). This difference can be explained by the fact that our analysis works on little structured bytecode in contrast to the simplified high-level representation used by~\cite{kalra2018zeus}. Additional overhead needs to be attributed to the usage of sound abstractions on bytecode level as well as to our different experimental setup that did not allow for the same amount of parallelization.
The soundness claim of~\cite{kalra2018zeus} is challenged by the experimentally assessed sensitivity of only $11.4\%$.
One possible explanation for this low value, which deviates from the numbers reported in~\cite{kalra2018zeus} on the same data set, is that the intuition guiding the manual investigation performed by~\cite{kalra2018zeus} departed from the notion of single-entrancy and the intuitive definition given by the authors. This highlights the importance of  formalizing not only the analysis technique but also the security properties to be verified.
When interpreting the high specificity of ~\cite{kalra2018zeus} (almost $100\%$) one should consider that ZEUS labels only $22$ contracts vulnerable in total out of which one is a false positive. Given that the dataset is biased towards safe contracts ($517$ safe as opposed to $195$ unsafe ones) a high specificity can be the result of a tool's  tendency to label contracts erroneously secure. Due to the proven soundness, for \ethor{} such a behavior is excluded by design.
This overall advantage in accuracy of \ethor{} over ZEUS is reflected by \ethor{}'s F-measure of $88.9\%$ as opposed to $20.4\%$ for ZEUS.

\subsubsection{Horn Clause Folding}
\begin{figure}
\centering
\includegraphics[width=\linewidth]{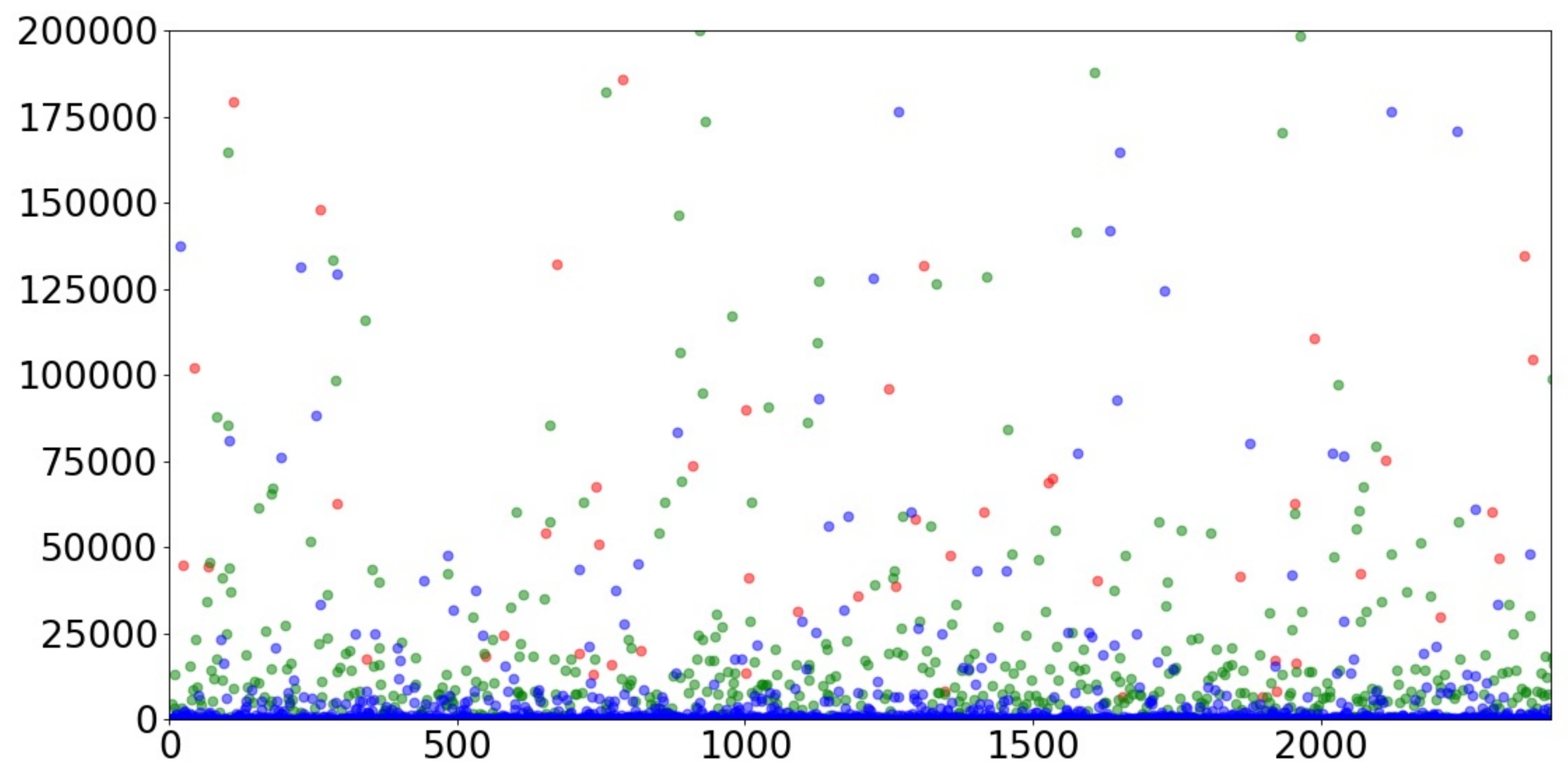}
\caption{Query runtimes in ms for the combined approach itemized by queries.
\normalfont
A red/green/blue dot denotes a query solved fastest with no/linear/exhaustive folding.}
\label{fig:query-times}
\end{figure}
Our experimental evaluation shows that, while both forms of \hornc clause folding improve the termination rate, the results of the different foldings are not directly comparable.
This is illustrated by \autoref{fig:query-times} which plots the (lowest) termination times for those queries that terminated within 200 seconds during the large-scale experiment. The different colors  indicate the kind of optimization (no/linear/exhaustive folding) that was fastest to solve the corresponding query.
The distribution of the dots shows that in the range of low query times (indicating more simply structured contracts) exhaustive folding (depicted in blue) dominates. However, for longer query times, often the linear folding (depicted in green) shows a better performance. One possible explanation for this is that for more complex contracts, the blow-up in rules created by the exhaustive folding exceeds the benefits of eliminating more predicates. Interestingly, for few instances, even applying no folding at all (indicated in red) led to the fastest termination. We can only explain this behavior by special heuristics used inside \zz{} that helped these particular cases. This shows the lacking predictability of \zz{} and thereby motivates the necessity of high-level tools like \horst{} that allow users to easily combine different optimizations in order to obtain results reliably.

\section{Related Work \& Discussion}
\label{sec:related-work}

In the last years there has been plenty of work on the automatic analysis of Ethereum smart contracts. 
These works can be classified in dynamic and static analysis approaches: 
An example of a dynamic approach is the work by Grossmann et al.~\cite{grossman2017online} which studies effectively callback freedom, a property characterizing the absence of reentrancy bugs, and provides a dynamic detection algorithm for it. Besides that the authors prove that statically showing a smart contract to be effectively callback free is indeed undecidable. 
This work serves as a motivation why the correct and precise static analysis of smart contracts with respect to relevant security properties is  challenging and requires the usage of suitable (sound) abstractions in order to be feasible, or even possible.

As a consequence, most practical static analysis tools so far focused on the heuristic detection of certain (classes of) bugs. These works do not strive for any theoretical guarantees nor do they aim for giving formal (and semantic) characterizations of security properties that the analysis targets. Important representatives of such bug-finding tools are the static analyzer Oyente~\cite{luu2016making} which was the first static analyzer for Ethereum smart contracts and the tool Osiris~\cite{torres2018osiris} which focuses on the detection of integer bugs.
Since the aim and scope of these tools differ substantially from the ones of \ethor{}, we omit a more detailed discussion and refer the interested reader to recent surveys~\cite{GMS::CAV18, di2019survey,liu2019survey} for more details.

In contrast to pure bug-finding, some recent works target a sound and automated static analysis of smart contract security properties. In particular the tools Securify~\cite{securify}, ZEUS~\cite{kalra2018zeus}, EtherTrust~\cite{GMS::CAV18}, and NeuCheck~\cite{lu2019neucheck} make such soundness claims. 

Securify implements a dependency analysis based on the reconstructed control-flow graph of contract bytecode and expresses some generic security properties in terms of these dependencies. The paper claims that the dependency patterns which they provide are either sufficient for the satisfaction (compliance patterns) or the violation (violation patterns) of a property. However, no proofs for the correctness of the control-flow-graph transformation, the soundness of the dependency analysis itself, or the relationship between the security patterns and the properties are provided. As a consequence we could empirically show that Securify's algorithm for control-flow reconstruction is unsound and give counterexamples for $13$ out of $17$ patterns (which partly also indicate a flaw in the dependency analysis itself). In~\autoref{app:soundness-issues} we give a detailed account of these issues.

NeuCheck is a tool for analysing Solidity smart contracts by searching for patterns in the contract syntax graph. The soundness claim of the work is neither substantiated by a soundness statement, nor proof. Also no semantics of Solidity is given and no (formal) security properties are formulated.
This lack of formalism makes it hard to validate any soundness claim. Despite the missing formal connections, the given patterns are clearly of syntactic nature and can be argued not to match the intuitive properties given throughout the paper which makes NeuCheck rather a bug-finding and style-checking tool. For more details we refer to~\autoref{app:soundness-issues}.

Similar to NeuCheck, ZEUS analyzes smart contracts written in Solidity. To this end, it transforms Solidity smart contracts into an abstract intermediate language and later into LLVM bitcode which allows for leveraging existing symbolic model checkers. 
The code transformations are claimed to be semantics preserving which however has already been refuted by~\cite{GMS::CAV18}. Additionally, the analyzed security properties are neither formally defined nor are they translated for model checking in a streamlined fashion: while some of them are compiled to assertions, other properties require additional code transformations which we show to be flawed in~\autoref{app:soundness-issues}. Empirical evidence for the unsoundness of ZEUS has been reported by~\cite{torres2018osiris} and is emphasized by the empirical evaluation in~\autoref{sec:eval}.

The work presented in~\cite{GMS::CAV18} surveys different approaches to static analysis and aims at illustrating design choices and challenges in sound static analysis. The work also discusses EtherTrust, a first proof of concept for a reachability analysis based on \hornc{} clauses, which however is still preliminary and exhibits soundness issues in its abstraction as discussed~\autoref{app:soundness-issues}.

For avoiding the pitfalls leading to unsoundness in the presented works, \ethor{} follows a principled design approach:
Starting from the formal EVM semantics defined in~\cite{GMS::POST18}, it formulates an abstract semantics in the specification language \horst{} which is proven sound with respect to the concrete semantics, hence covering all particularities of the EVM bytecode language. 
Based on this abstract semantic specification, a streamlined compilation process creates an SMT-encoding which is again systematically tested for soundness against the official test suite to minimize the effect of implementation bugs. The challenge of sound control flow reconstruction is solved by basing a corresponding preanalysis on a proper relaxation of the provably sound abstract semantics in the \souffle{} format, ensuring that the original soundness guarantees are inherited. For a more robust development, it is planned to also streamline this process in the future by making the \horst{} compiler support \souffle{} as additional output format for a restricted \hornc{} clause fragment. 
For providing end-to-end guarantees of the resulting static analyzer, we do not only ensure the soundness of the core analysis by proofs and testing, but also give provably sound approximations for relevant formalized semantic security properties suitable for encoding in the analysis framework.

\section{Conclusion}
\label{sec:conclusion}
We presented \ethor{}, the first automated tool implementing a sound static analysis technique for EVM bytecode, showing how to abstract the semantics of EVM bytecode into a set of \hornc clauses and to express security as well as functional properties in terms of reachability  queries, which are solved using \zz{}. In order to ensure the long-term maintenance of the static analyzer and facilitate future refinements, we designed \horst{}, a development framework for \hornc-clause based static analysis tools, which given a high-level specification of \hornc clauses automatically generates an optimized implementation in the \smtlib{} format. We successfully evaluated \ethor{} against the official Ethereum test suite to gain further confidence in our implementation and conducted a large-scale evaluation, demonstrating the practicality of our approach. Within a large-scale experiment we compared \ethor{} to the state-of-the-art analysis tool ZEUS, demonstrating that \ethor{} surpasses ZEUS in terms of overall performance (as quantified by the F-measure).

This work opens up several interesting research directions. For instance, we plan to extend our analysis as well as~\horst{} to relational properties, since some interesting security properties for smart contracts can be defined in terms of 2-safety properties~\cite{GMS::POST18}. Furthermore, we intend to further refine the analysis in order to enhance its precision, e.g., by extending the approach to a multi-contract setting, introducing abstractions for calls that approximate the account's persistent storage and local memory after calling more accurately.
Furthermore, we plan to significantly extend the scope of \horst{}. First, we intend to  make the specification of the static analysis accessible to proof assistants in order to mechanize soundness proofs. Furthermore, we intend to explore the automated generation of static analysis patterns from the specification of the concrete semantics, in order to further reduce the domain knowledge required in the design of static analyzers.

\begin{acks}
This work has been partially supported by the the European Research Council (ERC) under the European Union's Horizon 2020 research (grant agreement 771527-BROWSEC); by the Austrian Science Fund (FWF) through the projects PROFET (grant agreement P31621) and the project W1255-N23; by the Austrian Research Promotion Agency (FFG) through the Bridge-1 project PR4DLT (grant agreement 13808694) and the COMET K1 SBA; and by the Internet Foundation Austria (IPA) through the netidee project EtherTrust (Call 12, project 2158).
\end{acks}

\bibliographystyle{ACM-Reference-Format}
\bibliography{biblio.bib}

\section*{Appendix}
\label{appendix}
\renewcommand{\thesection}{\Alph{section}}
\setcounter{section}{0}
\renewcommand\theHsection{appendix.\arabic{section}}

The appendix is structured as follows: In~\autoref{app:horst} we overview the analysis specification language \horst{}. In~\autoref{app:theory} we make the theoretical foundations of our work explicit, in particular we give the soundness proof of our analysis. \autoref{app:security-properties} gives details on how the security properties discussed in the paper are implemented in \ethor{} using \horst{}. Finally, in~\autoref{app:soundness-issues} we discuss the soundness issues in related works and give concrete counterexamples for highlighting soundness flaws in existing static analyzers.

\section{\horst{}}
This section gives an introduction to the newly developed language \horst{} that allows for the high-level specification of Horn-clause based static analyses. We will first give a short primer that illustrates the main functionality of \horst{}, followed by a more detailed discussion of \horst{}'s language features.
\label{app:horst}
\subsection{\horst{} by Example}
\label{sec:appendix-horst}

For illustrating the features of \horst{} we show how to express a general rule for binary stack operations, subsuming the rule for addition presented in~\autoref{sec:stat}.
\begin{figure*}[!t]
\lstset{escapechar=$,language=HoRSt, numbers=left,xleftmargin=2em,frame=lines,framexleftmargin=1.5em}
\begin{lstlisting}
datatype AbsDom := @T | @V<int>; // Abstract Domain
datatype Opcode := @STOP | @ADD | ... | @INVALID | @SELFDESTRUCT // opcodes (shortened)

pred MState{int*int}: int * array<AbsDom> * array<AbsDom> * array<AbsDom> * bool;

op absadd(a: AbsDom, b: AbsDom): AbsDom := match (a, b) with | (@V(x), @V(y)) => @V((x + y) mod MAX) | _ => @T;
op binOp(c: Opcode, x: AbsDom, y: AbsDom): AbsDom := match c with | @ADD => absadd(x, y) | ... | _ => @T;

sel ids: unit -> [int]; // contracts to be analyzed
sel binOps: unit -> [int]; // binary stack operations
sel pcsForIdAndOpcode: int * int -> [int]; // program counters at which a specific opcode occurs in a specific contract
sel argumentsTwoForIdAndPc: int * int -> [int * int]; // results from the preanalysis for a given contract and pc

op tryConcrete{!c:int}(val:AbsDom): AbsDom := (!c = ~1) ? (val) : (@V(!c));

rule opBin :=
  for ($!op$: int) in binOps(),
      (!id: int) in ids(),
      (!pc: int) in pcsForIdAndOpcode(!id, $!op$),
      (!a:int, !b: int) in argumentsTwoForIdAndPc(!id, !pc)
   clause [?x: AbsDom, ?y:AbsDom, ?size: int, ?sa: array<AbsDom>, ?mem: array<AbsDom>, ?stor: array<AbsDom>, ?cl: bool]
     MState{!id, !pc}(?size, ?sa, ?mem, ?stor, ?cl), ?size > 1,
     ?x = tryConcrete{!a} (select ?sa (?size -1)),
     ?y = tryConcrete{!b} (select ?sa (?size -2))
     => MState{!id, !pc +1}(?size -1, store ?sa (?size -2) (binOp(intToOpCode($!op$), ?x,?y)), ?mem, ?stor, ?cl);
\end{lstlisting}
\caption{\horst{} rule describing the abstract semantics of local binary stack operations.}
\label{fig:opbinop}
\end{figure*}
\autoref{fig:opbinop} shows an excerpt of the \horst{}-specification of the presented static analysis. The abstract domain of the analysis is realized by the definition of the abstract datatype \lstinline|AbsDom|. Predicate signatures can be specified by corresponding predicate declarations as done for the case of the \lstinline|MState| predicate. \horst{} allows for parametrizing predicates and thereby specifying whole predicate families: The \lstinline|MState| predicate is parametrized by two integer values (as specified in the curly braces) that will intuitively correspond to the contract's identifier and the program counter whose state it is approximating. The arguments of the \lstinline|MState| predicate family reflect exactly those specified in~\autoref{sec:stat}.

To facilitate modular specifications, \horst{} supports non-recursive operations over arbitrary types, such as \lstinline|absadd| which implements abstract addition. In the given example, we show the flexibility of \horst{} by presenting a single rule template for generating rules for all binary stack operations. To this end, we define a function \lstinline|binOp| that given an opcode \lstinline|c| and two integer arguments applies to them the binary operation corresponding to the opcode. This function is then leveraged in the rule template \lstinline|opBin|.
Rule templates serve for generating the abstract semantics  given in the form of Horn clauses. As in our case the abstract semantics is specified as a function on a concrete contract, the generation of horn clauses in \horst{} needs to be linked to a concrete contract bytecode. In order to account for that in a generic fashion given that \horst{} cannot support facilities for reading files or parsing bytecodes, \horst{} provides an interface for interacting with custom relations generated by \java{} code. This interface is specified upfront by so called \emph{selector functions} (introduced with the key word \lstinline|sel|) which are declared, but not defined in the \horst{} specification. In the given example, we declare selector functions for accessing the identifiers of the contracts to be analyzed (\lstinline|ids|), the set of binary operations (\lstinline|binOps|) and for the program counters in a contract that hold opcodes of a specific type (\lstinline|pcForIdAndOpcode|). In addition to that, selector functions also allow for more advanced functionalities such as incorporating the results of a pre-analysis in an elegant fashion: To this end, we declare the selector function \lstinline|argumentsTwoForIdAndPc| that returns arguments to the operation that could be statically pre-computed (returning $-1$ in case of failure).
For generating Horn clauses, we can parametrize the rule over the cross product of the result of (nested) selector function applications as done in for the \lstinline|opBin| rule. This then exactly generates Horn clauses abstracting the behavior of a binary stack operation as discussed in \autoref{sec:stat}: A stack size check is performed, the two arguments are selected from the stack and finally the \lstinline|MState| predicate at the next program counter is implied with an updated stack having the operation's result as top element. The only derivation occurs due to the consideration of the pre-analysis: the operation \lstinline|tryConcrete| tries to access the statically precomputed argument, and only in case its absence performs the (more expensive) stack access. This step  however is not a necessity, but just illustrates how the interplay between different stages of a static analysis can be implemented for boosting the performance. 

\subsection{\horst{} in Detail}

In the following, we present a short overview of the features of \horst{}.

\subsubsection{Types and Operations}
For specifying the super domain $\absdoms$ of the abstraction, \horst{} provides
in addition to the primitive types \textsc{Boolean} and \textsc{Integer},
non-recursive sum types and arrays over all types. The type of abstract values $\absdom$ used in \autoref{sec:stat}
that consists of the unknown value $\top$ and concrete integer values,
can be defined as follows:
\lstset{escapechar=,language=HoRSt, numbers=left,xleftmargin=2em,frame=lines,framexleftmargin=1.5em}
\begin{lstlisting}
datatype AbsDom = @T | @V<int>;
\end{lstlisting}
In addition, \horst{} allows us to define non-recursive operations over arbitrary types. These operations
are implemented as hygienic macros on the expression level. To work with sum
types, \horst{} provides match expressions.
This mechanism can, e.g., be used to define the abstract addition operation described in \autoref{subsubsec:addition} as follows:
\begin{lstlisting}
op absadd(a: AbsDom, b: AbsDom): AbsDom :=
  match (a, b) with
  | (@V(x), @V(y)) => @V(x + y) // for two concrete values, return sum
  | _              => @T;       // else return top
\end{lstlisting}

\subsubsection{Predicates}
The abstraction's predicate signature $\predsig$ is given in terms of predicate declarations. A predicate declaration introduces a predicate symbol that ranges over arguments of arbitrary types.
\horst{} supports a mechanism for declaring a whole family of predicates with the same argument types by allowing for the specification of compile-time constants that we will from now on call \emph{parameters}.
We illustrate the syntax of predicate declarations with the predicate  $\pnmstate{\pc}$ defined in \autoref{fig:pred-signature} that models an abstract execution state:
\begin{lstlisting}
pred MState{int}: int * array<AbsDom> * array<AbsDom> * array<AbsDom>  * bool;
\end{lstlisting}
The declared predicate has one parameter of type \texttt{int} and five arguments. The parameter represents the program counter $\pc$ and should be considered part of the predicate name.
The distinction between parameters and arguments is supported by \horst{} for performance reasons: different
parameter instantiations are compiled to different predicate names in the underlying
SMT representation which leads to speed-ups in practice and additionally facilitates the folding optimization
discussed in \autoref{subsubsec:folding}.
\subsubsection{Selector Functions}
\horst{} itself provides no facilities to read files, parse bytecode, etc.
Instead, these tasks are handled by \java{} code.
\horst{} interacts with this \java{} code by an upfront-specified interface which is implemented by so-called \emph{selector functions}.
The tasks performed by selector functions can be as easy as providing an
interval of integers or as complicated as precomputing the results of certain
bytecode operations, from a \horst{} perspective we only see the interface provided by \emph{selector function declarations} that associate selector function names with their type signature. Selector functions are restricted to take a fixed number of arguments of primitive types and to return a sequence of tuples of
primitive types.

Examples of selector function declarations are given below:
\begin{lstlisting}
sel interval: int -> [int]; // integers from 0 to (n-1)
sel pcsForOpcode: int -> [int]; // program counters for given opcode
sel pcsAndValuesForOpcode: int -> [int*int]; // program counters and precomputed values
\end{lstlisting}
In general, selector functions can be seen as the bridge between the analysis specification and the parts of the software stack responsible for preprocessing (parsing, etc.) real world smart contracts. For instance, as previously discussed, the predicate signature $\predsig_{\cstar}$, the abstraction function $\configabs{}_{\cstar}$ as well as the abstract semantics $\delta(\cstar)$ are dependent on the concrete contract $\cstar$ under analysis. Selector functions allow us to implement such a parametrization (e.g., iterating over opcode sequences in order to generate rules in $\delta(\cstar)$ according to the opcode at each program counter). 

The separation of concerns introduced by selector functions helps to keep the
\horst{} specifications declarative while the technical details of providing the
actual values can be tested by unit tests.

\subsubsection{Rules}
The fundamental abstraction of \horst{} is the concept of \emph{rule}, which essentially describes a collection of \hornc clauses. It therefore can be seen as the mechanism for specifying the abstract semantics $\delta(\cstar)$.
 A rule is either a
singleton rule that is just instantiated once or may act as a template for
arbitrarily many instantiations -- hence describing a family of rules.
The second case is enabled by the use of  selector
functions which provide the sequence that the rule family ranges over.
More technically, for each tuple returned by a selector function the parameters of the rule template will be
instantiated according to the tuple values.

The rule shown in \autoref{fig:opmstore} for example will be instantiated for all program counters $\texttt{!pc}$ at which $\cstar$ holds an $\MSTORE$ instruction.
The sequence of these program counters is provided by the selector function \texttt{pcsForOpcode} that maps opcodes to their corresponding set of occurrences (identified by program counter) in $\cstar$.

\begin{figure*}[!t]
\begin{lstlisting}
op valToMemWord (v: AbsDom, mem: array<AbsDom>, o: int): array<AbsDom> :=
  for (!a: int) in interval(32): x: array<AbsDom> -> store x (o + !a) (absExtractByteL{!a}(v)), mem;
op isConcrete(a: AbsDom): bool := match a with | @T => false | _ => true;
op extractConcrete(a: AbsDom): int := match a with | @V(x) => x | _ => 0;

rule opMstore :=
  for (!id: int) in ids(),
      (!pc:int) in pcsForIdAndOpcode(!id, MSTORE),
      (!p: int, !v: int) in argumentsTwoForIdAndPc(!id, !pc)
    clause [?size: int, ?sa: array<AbsDom>, ?mem: array<AbsDom>, ?stor: array<AbsDom>, ?cl: bool, ?offset: AbsDom, ?p: int, ?v: AbsDom]
      MState{!id,!pc}(?size, ?sa, ?mem, ?stor, ?cl), ?size > 1,
      !p != ~1,
      ?v = tryConcrete{!v}(select ?sa (?size -2))
      => MState{!id, !pc +1}(?size - 2, ?sa, writeWord{!p}(?v, ?mem), ?stor, ?cl),
    clause [?size: int, ?sa: array<AbsDom>, ?mem: array<AbsDom>, ?stor: array<AbsDom>, ?cl: bool, ?pos: AbsDom,
            ?v: AbsDom, ?memn: array<AbsDom>]
      MState{!id,!pc}(?size, ?sa, ?mem, ?stor, ?cl), ?size > 1,
      !p = ~1,
      ?pos = select ?sa (?size -1),
      ?v = tryConcrete{!v}(select ?sa (?size -2)),
      ?memn = (isConcrete(?pos)) ? (writeWordEven(extractConcrete(?pos), ?v, ?mem)) : ([@T])
      => MState{!id, !pc +1}(?size - 2, ?sa, ?memn, ?stor, ?cl);
\end{lstlisting}
\caption{\horst{} rule describing the abstract semantics of the local memory write operation $\instabs{MSTORE}_\pc$ }
\label{fig:opmstore}
\end{figure*}

Within the body of rules we can define (optionally hygienic) macros that we can
use in the subsequent \emph{clauses} of the rule.
The clauses themselves (declared with keyword \texttt{clause}), describe a \hornc clause consisting of a list of premises and a conclusion ranging over free variables which need to be explicitly declared upfront.
Premises are lists of predicate applications and boolean \horst{} expressions while the conclusion may only consist of a single predicate application.
The example in \autoref{fig:opmstore} defines three clauses that exactly correspond to the \hornc clauses defined for $\instabs{\MSTORE}_\pc$ in \autoref{fig:absrules}.

\subsubsection{Sum Expressions}
Selector functions can not only be used to generate rules, but can also be used at the expression level.
So-called sum expressions exist in two different shapes: in the simple case (shown later in \autoref{fig:automated-testing}), predefined associative operations (addition, multiplication, disjunction and conjunction) are used to join expressions that may make use of the values returned by the selector function;
the generalized case can be seen in line $2$ of \autoref{fig:opmstore}.
The operation \lstinline|valToMemWord| updates 32 consecutive memory cells of \lstinline|mem| with fractions of the value \lstinline|v| starting from position \lstinline|o| --- \lstinline|mem| is the start value, \lstinline|store x (o + !a) (absExtractByteL{!a}(v))| is the iterated expression (\lstinline|x| acts as a placeholder for the last iteration step's result).

\subsubsection{Queries}
In order to check for reachability of abstract configurations, \horst{} allows for the specification of (reachability) queries that can also be generated from selector functions. The query shown in \autoref{fig:reentrancy-query} for instance, checks for reentrancy by checking if
any \CALL{} instruction is reachable with call level $\cl = 1$ (here encoded as bool). It therefore is an implementation of the reachability property introduced in \autoref{subsec:reachability}.

Note that if there is a notion of an expected outcome, we can define queries with the keyword \lstinline|test| as seen in \autoref{fig:automated-testing}.

\begin{figure}[htb]
\lstset{escapechar=,language=HoRSt, numbers=left,xleftmargin=2em,frame=lines,framexleftmargin=1.5em}
\begin{lstlisting}
query reentrancyCall
  for (!id: int) in ids(),
      (!pc:int) in pcsForIdAndOpcode(!id, CALL)
    [?sa: array<AbsDom>, ?mem: array<AbsDom>,
     ?stor: array<AbsDom>, ?size:int]
      MState{!id, !pc}(?size, ?sa, ?mem, ?stor, true);
\end{lstlisting}
\caption{\horst{}-query for reeentrancy.}
\label{fig:reentrancy-query}
\end{figure}

\section{Theoretical Foundations of \ethor{}}
In this section, we provide details on the theoretical foundations of \ethor{}. We start by formally characterizing the notion of Horn-clause based abstractions as they can be implemented in \horst{} and relate this concept to the general framework of abstract interpretation. Next, we provide missing details on the definition of the static analysis underlying \ethor{} and conclude with a detailed proof sketch of the soundness statement for this analysis.
\label{app:theory}

\subsection{Horn-clause based abstractions}
\label{sec:appendix-key}
In this section, we more formally characterize the aim and scope of this work, as well as the kind of static analyses that are realizable by \horst{}.
Generally, we focus on the reachability analysis of programs with a small-step semantics, which we over-approximate by an abstract program semantics based on \hornc clauses.
More formally, we will assume a program's small-step semantics to be a binary relation $\smalls$ over program configurations $\config \in \configs$. A \hornc-clause based abstraction for such a small-step semantics $\smalls$ is then fully specified by a tuple ${(\absdoms, \predsig, \configabs{\predsig}, \absts)}$ where $\predsig$ defines the signature of predicates with arguments ranging over (partially) ordered subsets of $\absdoms$. For a given a predicate signature $\predsig$, an abstraction function $\configabs{\predsig}: \configs \to \absconfigs$ maps concrete program configurations $\config \in \configs$ to abstract program configurations $\absconfig \in \absconfigs$ consisting of instances of predicates in $\predsig$.

Formally, a predicate signature $\predsig \in \prednames \nrightarrow \prod (\setof{\absdoms} \times (\setof{\absdoms} \times \setof{\absdoms}))$ is a partial function from predicate names $\prednames$ to their argument types (formally written as a product over the subsets of some abstract superdomain $\absdoms$, equipped with a corresponding order). We require for all $n \in \prednames$ that $(D, \leq) \in \predsig(n)$ that $(D, \leq)$ forms a partially ordered set.
Correspondingly, the set of abstract configurations $\absconfigs_{\predsig}$ over $\predsig$ can be defined as $\setof{\{ n(\vec{v}) ~|~ n \in \prednames ~\land~ \forall i \in \{1, \dots, \size{\predsig(n)}\}.\,  \proj{i}{\predsig(n)} = (D, \leq) \implies \proj{i}{\vec{v}} \in D \}}$ where $\proj{i}{\cdot}$ denotes the usual projection operator.
The abstraction of a small-step semantics is then a set of constrained \hornc clauses $\abss \subseteq \hclauses{\predsig}$ that approximates the small-step execution rules.

A constrained horn clause is a first order formula of the form
\begin{align*}
\hc{X}{\constraints}{\premises}{\conclusion}
\end{align*}
Where $X \subseteq \vars \times \setof{\absdoms}$ is a (functional) set of typed variables, and  $\constraints$ is a set of quantifier free constraints over the variables in $X$.
Conclusions $\conclusion$ are predicate applications $n(\vec{z}) \in P_X  := \{n(\vec{x}) ~|~ \size{\vec{x}} = \size{\predsig(n)} ~ \land~ \forall i \in \{1, \dots, \size{\vec{x}} \}.\, \proj{i}{\vec{x}} = x \land \proj{i}{\predsig(n)} = (D, \leq) \implies (x, D) \in X \}$ over variables in $X$ that respect the variable type.
Correspondingly, the premises $\premises \subseteq P_X$, are a set of predicate applications over variables in $X$.

We lift the suborders of $\predsig$ to an order on abstract configurations $\absconfig_1, \absconfig_2 \in \absconfigs_\predsig$ as follows:
\par
\nobreak
{\small
\noindent
\begin{align*}
n_1(\vec{t_1}) \ord{p} n_2(\vec{t_2})
&:=  n_1 = n_2  \\
& ~\land~ \forall i \in \{1, \dots, \tsize{\vec{t_1}} \}. ~ \proj{i}{\vec{t_1}} \ord{n_1, i} \proj{i}{\vec{t_2}}
\\
\text{given $\proj{i}{\predsig(n)} = (D_{n,i}, \ord{n, i}) $}
\\
\absconfig_1 \leq \absconfig_2
&:= \forall p_1 \in \absconfig_1.~ \exists p_2 \in \absconfig_2.~ p_1 \ord{p} p_2
\end{align*}
}%
Finally, we introduce the notion of soundness for a \hornc-clause based abstraction.

\begin{definition}
\label{def:soundness}
A \hornc-clause based abstraction ${(\absdoms, \predsig, \configabs{\predsig}, \absts)}$ \emph{soundly approximates} a small-step semantics $\smalls$ if
\par
\nobreak
{\small
\noindent
\begin{align}
\forall (\config, \config') \in \smalls^*.~
&\forall \absconfig.~ \configabs{\predsig} (\config) \leq \absconfig \nonumber \\
&\Rightarrow
\exists \absconfig'.~ \absconfig, \absts \derives \absconfig' \ ~\land~ \configabs{\predsig}(\config') \leq \absconfig' \label{eqn:soundness}
\end{align}
}%
\end{definition}

This statement requires that, whenever a concrete configuration $\config'$ is reachable from configuration $\config$ (meaning that $(\config, \config')$ is contained in the reflexive and transitive closure of $\smalls$, denoted as $\smalls^*$), it shall hold that from all abstractions $\absconfig$ of $\config$, the \hornc clause abstraction allows us to logically derive ($\derives$) a valid abstraction $\absconfig'$ of $\config'$. Note that $\configabs{\predsig}$ intuitively yields the most concrete abstraction of a configuration, hence to make the property hold for all possible abstractions of a configuration, we strengthen the property to hold for all abstractions that are more abstract than $\configabs{\predsig}(c)$.
The soundness theorem implies  that whenever we can show that from some abstraction $\absconfig$ of a configuration $\config$ there is no abstract configuration $\absconfig'$ derivable such that $\absconfig'$ abstracts $\config'$, then $\config'$ is not reachable from $\config$. Consequently, if it is possible to enumerate all abstractions of $c'$, checking non-derivability (as it is supported by the fixedpoint engines of modern SMT solvers) gives us a procedure for proving unreachability of program configurations.

\subsubsection{Relation to abstract interpretation}
\label{subsubsec:abstract-interpretation}
It is possible to phrase the previous characterization in terms of classical abstract interpretation notions. More precisely, we can define a Galois connection $(\alpha, \gamma)$ between sets of concrete configurations $\setof{\configs}$ (ordered by $\subseteq$) and abstract configurations $\absconfigs$ (ordered by $\ord{}$). To this end, we lift the abstraction function $\alpha$ to subsets in a canonical fashion:

\begin{align}
\alpha(C) \define \bigcup_{c \in C}{\alpha(c)} \label{eqn:alpha-lifted}
\end{align}

Next, we define the the concretization function with the help of $\alpha$:
\begin{align*}
\gamma(\absconfig) \define \{ c \in \configs ~|~ \alpha(c) \ord{} \absconfig \}
\end{align*}

\begin{lemma}
The pair of functions $(\alpha, \gamma)$ forms a Galois connection between $(\setof{\configs}, \subseteq)$ and $(\absconfigs, \ord{})$.
\end{lemma}
\begin{proof}
We need to show for all $C$ and $\absconfig$ that
\begin{align*}
\alpha(C) \ord{} \absconfig \Leftrightarrow C \subseteq \gamma(\absconfig)
\end{align*}
\begin{itemize}
\item[$\Rightarrow$:]
Let $\alpha(C) \ord{} \absconfig$. Further let $c \in C$. We show that $c \in \gamma(\absconfig)$. By the definition of $\gamma$ it is sufficient to show that $\alpha(c) \ord{} \absconfig$. Let $p_1 \in \alpha(c)$. We show that there is some $p_2 \in \absconfig$ such that $p_1 \ord{} p_2$. Since $p_1 \in \alpha(c)$ and $c \in C$, we know that $p_1 \in \alpha(C)$ and since $\alpha(C) \ord{} \absconfig$ also that there needs to be some $p_2 \in \absconfig$ such that $p_1 \ord{} p_2$ what concludes the proof.
\item[$\Leftarrow$:]
Let $C \subseteq \gamma(\absconfig)$. Further let $p_1 \in \alpha(C)$. We show that there is some $p_2 \in \absconfig$ such that $p_1 \ord{} p_2$. Since $p1 \in \alpha(C)$ there must be some $c \in C$ such that $p_1 \in \alpha(c)$. And from $C \subseteq \gamma(\absconfig)$ we can conclude that $c \in \gamma(\absconfig)$ which implies that $\alpha(c) \ord{} \absconfig$. Consequently there needs to be a $p_2 \in \absconfig$ such that $p_1 \ord{} p_2$ what concludes the proof.
\end{itemize}
\end{proof}

Now, we can define reachability on concrete configurations and derivability of abstract configurations as the least fixed points of step functions which describe a collecting semantics (with respect to some initial configuration).

\begin{align*}
\sstepfun{I}(C) \define  \{c' ~|~ \exists c \in C. \, (c, c') \in \smalls \}  \cup I
\end{align*}
\begin{align*}
\derivefun{\absconfig_I}(\absconfig) \define  \{p ~|~ \absconfig, \absts \derives p \} \cup \absconfig_I
\end{align*}

We obtain the following intuitive correspondences between the different characterizations:
\begin{align}
(\config, \config') \in \smalls^* &\Leftrightarrow c' \in \lfp{\sstepfun{\{c\}}} \label{eqn:lfp-concrete}\\
\absconfig, \absts \derives \absconfig' &\Leftrightarrow \absconfig' \subseteq \lfp{\derivefun{\absconfig}} \label{eqn:lfp-abstract}
\end{align}
where $\lfp{f}$ denotes the least fixed point of a function $f$.

To ensure that the corresponding least fixed points exists, we need to ensure that the domains $\setof{\configs}$ and $\absconfigs$ of the Galois connection form a complete lattice and that both $\sstepfun{I}$ and $\derivefun{\absconfig_I}$ are monotone.
While $\langle \setof{\configs}, \subseteq, \emptyset, \setof{\configs}, \cup, \cap \rangle$ is the canonical power set lattice,
we can easily show $\langle \absconfigs, \ord{}, \emptyset, \absconfig, \cup, \cap \rangle$ to also form a complete lattice as $\subseteq$ is a subrelation of $\ord{}$.
While it is trivial to show that $\sstepfun{I}$ is monotone, for $\derivefun{\absconfig_I}$ it becomes a proof obligation on $\absts$:
\begin{align}
\forall \absconfig ,\absconfig'. \, \absconfig \ord{} \absconfig' \land \absconfig, \absts \derives p \implies \exists p'.\,  p \ord{} p' \land \absconfig', \absts \derives p' \label{eqn:mon-local}
\end{align}

\sloppy
Using the step functions, we can characterize sound over-approximations as defined in \autoref{def:soundness} in an alternative fashion.
More precisely, we require our approximation to be a \emph{sound upper approximation}~\cite{cousot2004basic}.
\begin{lemma}
A \hornc-clause based abstraction ${(\absdoms, \{ \ord{n, i} \}_{(n, i)}, \predsig, \configabs{\predsig}, \absts)}$ \emph{soundly approximates} a small-step semantics $\smalls$ iff $\absts$ satisfies \autoref{eqn:mon-local} and for all $c \in \configs$ and all $\absconfig \geq \alpha(c)$
\begin{align*}
\alpha(\lfp{\sstepfun{\{c\}}}) \ord{} \lfp{\derivefun{\absconfig}}
\end{align*}
\end{lemma}
\begin{proof}
"$\Rightarrow$": Assume \autoref{eqn:soundness} and $f_1 \in \alpha(\lfp{\sstepfun{\{c\}}})$ for some fact $f_1$. We show that there exists some fact $f_2$ such that $f_2 \in \lfp{\derivefun{\absconfig}}$ and $f_1 \ord{} f_2$.
By \autoref{eqn:alpha-lifted}, we know that from $f_1 \in \alpha(\lfp{\sstepfun{\{c\}}})$ we can conclude that there exists some $c' \in \lfp{\sstepfun{\{c\}}}$ such that $f_1 \in \alpha(c')$. By \autoref{eqn:lfp-concrete}, we have that $(c, c') \in \smalls^*$ and hence by  \autoref{eqn:soundness} we can conclude that there exists some $\absconfig'$ such that $\absconfig, \absts \derives \absconfig'$ and $\alpha(c')\ord{} \absconfig'$. With  $f_1 \in \alpha(c')$ we get from this that there exists some $f_2 \in \absconfig'$ such that $f_1 \ord{} f_2$.
Since  $\absconfig, \absts \derives \absconfig'$, we get from \autoref{eqn:lfp-abstract} that $\absconfig' \subseteq \lfp{\derivefun{\absconfig}}$ and hence also $f_2 \in  \lfp{\derivefun{\absconfig}}$ which concludes the proof.

"$\Leftarrow$":
Assume $\alpha(\lfp{\sstepfun{\{c\}}}) \ord{} \lfp{\derivefun{\absconfig}}$ and let $(\config, \config') \in \smalls^*$ and $\alpha(\config) \ord{} \absconfig$. We show that there is some $\absconfig'$ such that $\absconfig,\absts \derives \absconfig'$ and $\alpha(c') \ord{} \absconfig'$.
By \autoref{eqn:lfp-concrete}, we get that $c' \in \lfp{\sstepfun{\{c\}}}$ and hence also $\alpha(c') \subseteq \alpha(\lfp{\sstepfun{\{c\}}})$ (by \autoref{eqn:alpha-lifted}). As $\alpha(\lfp{\sstepfun{\{c\}}}) \ord{} \lfp{\derivefun{\absconfig}}$ it follows that also $\alpha(c')  \ord{} \lfp{\derivefun{\absconfig}}$. Additionally, it follows from \autoref{eqn:lfp-abstract} immediately that $\absconfig,\absts \derives \lfp{\derivefun{\absconfig}}$. This closes our proof.
\end{proof}

Given that $\derivefun{\absconfig}$ is monotonic, $\alpha(\lfp{\sstepfun{\{c\}}}) \ord{} \lfp{\derivefun{\absconfig}}$ can be shown to be a consequence of the following one-step characterization:
\begin{align}
\alpha \compose \sstepfun{} \ord{} \derivefun{} \compose \alpha \label{eqn:approx}
\end{align}
(where $\sstepfun{} = \sstepfun{\emptyset}$ and $\derivefun{} = \derivefun{\emptyset}$).

This is as $\alpha \compose \sstepfun{} \ord{} \derivefun{} \compose \alpha$ implies for all $c \in \configs$ and all $\absconfig \geq \alpha(c)$ that $\alpha \compose \sstepfun{\{c\}} \ord{} \derivefun{\absconfig} \compose \alpha$ and by the fixed point transfer theorem~\cite{cousot2004basic} for Galois connections, this result can be lifted to least fixed points.
As a consequence for proving \autoref{theorem:soundness}, it is sufficient to show that \autoref{eqn:mon-local} and \autoref{eqn:approx} hold.
\subsection{Analysis Definition (continued)}
\label{sec:appendix-ana-def}
We overview additional details of the analysis definition introduced in \autoref{sec:stat}.

First, we formally define the orders on the abstract argument domains for the predicates defined in \autoref{fig:pred-signature}.

\begin{align*}
\ord{\absdom} &\define \{(\abs{a}, \abs{b}) ~|~ \abs{b} = \top \lor \abs{a} = \abs{b} \} \\
\ord{\NN} &\define \{(m, n) ~|~ m = n \} \\
\ord{\BB} &\define \{(a, b) ~|~ a = b \} \\
\ord{\NN \to \absdom} &\define \{(f, g) ~|~ \forall n \in \NN.~ f(n) \ord{\absdom} g(n) \} \\
\ord{\NN \times (\NN \to \absdom)} &\define
\{ ((m, f), (n, g)) ~|~  m = n ~\land~ \forall i < m.~ f(i) \ord{\absdom} g(i) \}
\end{align*}
We assume that the same orders apply to the same argument domains of different predicates.

Some of the partially ordered described by the argument domains and their corresponding order, have a supremum, as formally stated in the following lemma:
\begin{lemma}[Suprema of argument domains]
\label{lem:suprema}
The following statements hold:
\begin{itemize}
\item
$\forall \abs{a} \in \absdom.~ \abs{a} \ord{\absdom} \top$
\item
$\forall f \in \NN \to \absdom.~ f \ord{\NN \to \absdom} \fun{x}{\top}$
\end{itemize}
\end{lemma}

\subsubsection{Abstract operations}
We formally define abstract operations on values from the abstract argument domains, starting with binary operations on natural numbers:
Let $\binop \in \NN \times  \NN \to \NN$ be a binary operation.
We define abstract binary operations as follows:
\begin{align*}
\absop{\cdot} &\in (\NN \times \NN \to \NN) \to \absdom \times \absdom \to \absdom\\
\absop{\binop}(\abs{x}, \abs{y}) &\define
\begin{cases}
\binop(\abs{x}, \abs{y}) &\abs{x}, \abs{y} \in \NN \\
\top & \text{otherwise}
\end{cases}
\end{align*}

Similarly, we can define abstract comparison operators. Let $\compop \in \NN \times \NN \to \BB$ be a comparison operation on natural numbers.
We define abstract comparison operations as follows:
\begin{align*}
\absop{\cdot} &\in (\NN \times \NN \to \BB) \to \absdom \times \absdom \to \BB\\
\absop{\compop}(\abs{x}, \abs{y}) &\define
\begin{cases}
\compop(\abs{x}, \abs{y}) &\abs{x}, \abs{y} \in \NN \\
1 & \text{otherwise}
\end{cases}
\end{align*}

Next, we define the abstract operations for memory access used in \autoref{fig:accessword} and \autoref{fig:absfun}.

First we define the function for extracting a specified fraction of an integer (interpreted as $32$ byte word)
\begin{align*}
\extract{\cdot}{\cdot}{\cdot} & \in \absdom \times \NN \times \NN \to \absdom \\
\extract{\abs{v}}{l}{r} &\define
\begin{cases}
\left \lfloor \frac{\abs{v}}{256^{31 - r}} \right \rfloor \bmod 256^{r-l+1} & l \leq r ~\land~ \abs{v} \in \NN  \\
\top & \text{otherwise}
\end{cases}
\end{align*}

Next we define the append function:
\begin{align*}
\aconcat{\cdot} &\in \absdom \times \absdom \times \NN \to \absdom \\
\abs{v} \aconcat{n} \abs{w} &\define
\begin{cases}
\abs{w}* 256^{n} + \abs{v} & \abs{v}, \abs{w} \in \NN \\
\top & \text{otherwise}
\end{cases}
\end{align*}


We focused here only those operations that were used in \autoref{sec:stat}. For a full account of all abstract operations, we refer the reader to our \horst{} specification in \autoref{app:horst}.

\subsection{Proof of soundness}
\label{sec:appendix-soundness}
\setlist[itemize,1]{leftmargin=\dimexpr 15pt}
\setlist[itemize,2]{leftmargin=\dimexpr 15pt}
\setlist[itemize,3]{leftmargin=\dimexpr 15pt}
\setlist[itemize,4]{leftmargin=\dimexpr 15pt}
\setlist[enumerate,1]{leftmargin=\dimexpr 15pt}
\setlist[enumerate,2]{leftmargin=\dimexpr 15pt}
\setlist[enumerate,3]{leftmargin=\dimexpr 15pt}
\setlist[enumerate,4]{leftmargin=\dimexpr 15pt}

For proving \autoref{theorem:soundness}, we will not make an immediate use of the proof strategy presented in \autoref{subsubsec:abstract-interpretation}.  Even though we proof monotonicity (\autoref{eqn:mon-local}) separately, since this facilitates the reasoning in the individual cases, we will in the end go for a direct proof of the statement in \autoref{theorem:soundness} proceeding by complete induction of the number of small-steps. The reason for that is that for proving our abstraction sound one-step reasoning is not sufficient as we need to argue about execution steps that lie further ahead (hence the use of complete induction).

For the sake of concise presentation, we will only present a proof sketch (featuring the most interesting and challenging cases) and state, but not prove, auxiliary lemmas.

\subsubsection{Auxiliary lemmas}
For reasoning about the soundness, we first need to state some general properties of the small-step execution and the shapes of callstacks that might appear during this execution. We omit the proofs for most of these properties as those are mostly straight-forward case distinctions and simple inductions.

The following lemmas summarize some general properties of callstack evolution during the execution.
The small-step semantics is designed such that the callstack records the execution state as at the point of calling. The corresponding states only get modified when returning from an internal transaction. In this case, modification is guaranteed, since the gas for the execution is subtracted. As a consequence, an unmodified (sub) callstack indicates that the execution of the same internal transaction is still executed.
More formally this is captured by the following lemma:

\begin{lemma}[Callstack preservation during execution]
\label{lem:callstack-preservation}
Let $(\transenv, \callstack)$ be a configuration such that $\ssteps{\transenv}{\concatstack{\callstackb}{\callstack}}{\concatstack{\callstackb'}{\callstack}}$. Then the following properties hold:
\begin{itemize}
\item if $\callstackb' = \nil$ then $\callstackb = \nil$
\item if $\callstackb = \nil$ and $\callstackb' \neq \nil$ then there are $\exstate \in \exstates$, $c \in \contracts$ such that $\sstep{\transenv}{\callstack}{\cons{\annotate{\exstate}{c}}{\callstack}}$ and $\ssteps{\transenv}{\cons{\annotate{\exstate}{c}}{\callstack}}{\concatstack{\callstackb'}{\callstack}}$.
\item if $\callstackb' \neq \nil$ and $\ssteps{\transenv}{\concatstack{\callstackb}{\callstack}}{\callstack'}$ and $\ssteps{\transenv}{\callstack'}{\concatstack{\callstackb'}{\callstack}}$ then there exists $\callstackb''$ such that $\size{\callstackb''} > 0$ and $\callstack' = \concatstack{\callstackb''}{\callstack}$
\end{itemize}
\end{lemma}

We introduce the notion of a \emph{call state} for characterizing those states that invoke internal transactions.
\begin{definition}[Call states]
\label{def:callstate}
A regular execution state $\exstate$ is a call state if
$\sstep{\transenv}{\cons{\exstate}{\callstack}}{\cons{\exstate'}{\cons{\exstate}{\callstack}}}$ for some $\transenv$, $\callstack$ and $\exstate'$.
\end{definition}
Intuitively, an execution state is a call state if it satisfies all preconditions for a transaction initiating instruction.

In a regular execution all elements of a callstack but its top element are call states.
\begin{lemma}
\label{lem:callstates-but-top}
Let $\ssteps{\transenv}{\cons{\exstate}{\callstack}}{\concatstack{\cons{\exstate'}{\callstack'}}{\callstack}}$, then every execution state $\exstate'' \in \callstack'$ is a call state.
\end{lemma}

Whenever some configuration is reachable, the execution before stepped through the call states on the callstacks.
This property is formally captured by the following lemma:
\begin{lemma}
\label{lemma:step-through-call}
Let $\nsteps{\transenv}{\callstack}{\concatstack{(\concatstack{\callstack_1}{\callstack_2})}{\callstack}}{n}$.
Then there exists some $m \in \NN$ such that
$\nsteps{\transenv}{\callstack}{\concatstack{\callstack_2}{\callstack}}{m}$
and $\nsteps{\transenv}{\concatstack{\callstack_2}{\callstack}}{\concatstack{(\concatstack{\callstack_1}{\callstack_2})}{\callstack}}{n-m}$.
\end{lemma}

As previously discussed, we assume execution states to be annotated with the contracts that they are currently executing.
These annotations need to be consistent with the current execution state in the sense that they correspond to the active account of the execution state and that they present a valid contract in the global state.
\begin{definition}[Annotation consistency]
An execution state $\exstate$ is consistent with contract annotation $c$ if the following two conditions hold
\begin{enumerate}
\item
$\isregular{\exstate} \implies \access{\access{\exstate}{\exenv}}{\activeaccount} = \access{c}{\addr}$
\item
$\isregular{\exstate} \lor \ishalt{\exstate} \implies \access{\access{\exstate}{\gstate}(\access{c}{\addr})}{\code} = \access{c}{\code}$
\end{enumerate}
where $\isregular{\cdot}$ and $\ishalt{\cdot}$ are predicates on execution states indicating whether they are regular execution states or halting states, respectively.
\end{definition}
The consistency of annotations is preserved over execution.
\begin{lemma}[Preservation of annotation consistency]
Let $\exstate$ be consistent with $c$ and $\ssteps{\transenv}{\cons{\annotate{\exstate}{c}}{\callstack}}{\concatstack{\callstack'}{\callstack}}$ for some $\transenv$, $\callstack$, and $\callstack'$.
Then for all $\annotate{\exstate'}{c'} \in \callstack'$ it holds that $\exstate'$ is consistent with $c'$.
\end{lemma}

In order to prove soundness, we will need to require a stronger form of consistency for the execution states of $\cstar$ that allows to relate the contract code to the currently executed code.
\begin{definition}[Strong annotation consistency]
An execution state $\exstate$ is strongly consistent with contract annotation $c$ if it is consistent with $c$ and additionally
\begin{align*}
\isregular{\exstate} \implies \access{\access{\exstate}{\exenv}}{\code} = \access{c}{\code}
\end{align*}
\end{definition}

Contract annotations reflect the active contract that is executed. The active account of an execution state cannot be changed during execution.
Formally, this is stated by the following lemma:
\begin{lemma}[Annotation persistence]
\label{lem:annotation-persistence}
Let $\ssteps{\transenv}{\concatstack{\callstack_1}{\cons{\annotate{\exstate}{c_1}}{\callstack}}}{\concatstack{\callstack_2}{\cons{\annotate{\exstate'}{c_2}}{\callstack}}}$.
Then it holds that $c_1 = c_2$.
\end{lemma}

In order to formally state soundness, as well as some preliminary lemmas, we need to put a minor restriction on the executions that we are considering. This is as in the case of contract creations, it is theoretically possible (with negligible probability) that (due to a hash collision in the $\pkec{\cdot}{\cdot}$ function) a contract with the same address as the contract $\cstar$ under analysis is created. In this case the contract's storage as well as code will be overwritten.
If such an over-write would occur in the execution of $\cstar$ (after giving up the control flow due the call or by performing a $\CREATE$ instruction itself), this would mean that any following execution of $\cstar$ would resume in the altered storage, and (even more severely) following a new contract code. Consequently there is no way of faithfully abstracting the execution of $\cstar$ once the control flow was handed over.
In practice, however, the occurrence of such a hash collision can be neglected due to its low probability.
Formally, we give a soundness guarantee only for those executions that do not encompass a problematic hash collision.
To this end, we formally characterize such \emph{collision-free} executions:

\begin{definition}[Collision-free execution]
\label{def:coll-free}
A ($n$-step) execution $\nsteps{\transenv}{\callstack}{\callstack}{n}$ is collision-free for contract $\cstar$ (written $\ncollfreestep{n}{\cstar}$ ) if
for all $m \leq n$ and all callstacks $\callstack''$  such that $\nsteps{\transenv}{\callstack}{\callstack''}{m}$ it holds that
\begin{align*}
\forall \annotate{\exstate}{c} \in \callstack''.~ \access{c}{\addr} \neq \access{\cstar}{\addr} \lor c = \cstar
\end{align*}
\end{definition}
This definition ensures that during the execution, the address of contract $\cstar$ can never be attached to a different code.
Given that execution states of contract creation are annotated with $(\rho, \bot)$ (where $\rho$ is the address of the contract in creation), this definition in particular rules out that the creation code of a contract with address $\access{\cstar}{\addr}$ is executed.

While the occurrences of such a colliding contract creation needs to be excluded on executions (since it could be performed by arbitrary contracts), exclusion of $\DELEGATECALL$ and $\CALLCODE$ extractions only applies to the executions of $\cstar$ and can therefore be syntactically enforced on $\cstar$'s contract code.
We establish the invariant that we obtain from excluding $\DELEGATECALL$ and $\CALLCODE$ from $\cstar$'s contract code:
\begin{lemma}[Annotation agreement for $\cstar$]
\label{lem:agreement}
Let $\cstar$ be a contract such that $\{ \DELEGATECALL, \CALLCODE \} \cap \access{\cstar}{\code} = \emptyset$.
Let further $\ssteps{\transenv}{\cons{\annotate{\exstate}{\cstar}}{\callstack}}{\concatstack{\callstack'}{\callstack}}$ be a collision free execution for $\cstar$ and $\exstate$ be strongly consistent with $\cstar$. Then for all (regular) execution states $\annotate{\exstate'}{c'} \in \callstack'$ it holds that
\begin{enumerate}
\item If $\access{\cstar}{\addr} = \access{\access{\exstate'}{\exenv}}{\activeaccount}$ then $\access{\cstar}{\code} = \access{\access{\exstate'}{\exenv}}{\code}$
\item If $c' \neq \cstar$ then $\access{\access{\exstate'}{\exenv}}{\activeaccount} \neq \access{\cstar}{\addr}$
\end{enumerate}
\end{lemma}
So, in a nutshell, for contracts not containing $\CALLCODE$ and $\DELEGATECALL$, strong consistency is preserved and additionally, the contract code is persistent (cannot change over the execution).

For arguing about the call abstraction, we show the following substantial lemma that allows us to trace back the storage of a (contract) account to the (result of a) prior execution of this contract.

\begin{lemma}[Storage evolution]
\label{lem:storage-evolution}
Let $\transenv$, $\callstack$, $\callstack'$, $\exstate$, $\exstate'$, and $c' \neq \cstar$ be such that $\exstate$ is strongly consistent with $\cstar$ and
\begin{align*}
\nsteps{\transenv}{\cons{\annotate{\exstate}{c^*}}{\callstack}}{\cons{\annotate{\exstate'}{c'}}{\concatstack{\callstack'}{\callstack}}}{n}
\end{align*}
is a collision-free execution for $\cstar$.
Then one of the following holds:
\begin{enumerate}
\item \label{case:exc}
$\exstate' = \excstate$
\item \label{case:callstate}
$\begin{aligned}[t]
&\exstate' \neq \excstate ~\land~ \\
&\exists \annotate{\exstate^*}{\cstar} \in \callstack' ~ \access{\access{\exstate'}{\gstate}(\access{\cstar}{\addr})}{\stor} = \access{\access{\exstate^*}{\gstate}(\access{\cstar}{\addr})}{\stor}
\end{aligned}$
\item \label{case:return}
$\begin{aligned}[t]
&\exstate' \neq \excstate ~\land\\
&\exists~ \callstack^* ~\gstate~\lgas~d~\transeffects~m.~ \size{\callstack^*} > 0 \\
&~\land~ \nsteps{\transenv}{\cons{\annotate{\exstate}{c^*}}{\callstack}}{\cons{\haltstatefull{\gstate}{\lgas}{d}{\transeffects}}{\concatstack{\callstack^*}{\callstack}}}{m} \\
&~\land~ \nsteps{\transenv}{\cons{\haltstatefull{\gstate}{\lgas}{d}{\transeffects}}{\concatstack{\callstack^*}{\callstack}}}{\cons{\annotate{\exstate'}{c'}}{\concatstack{\callstack'}{\callstack}}}{n-m}\\
&~\land~ \access{\access{\exstate'}{\gstate}(\access{\cstar}{\addr})}{\stor} = \access{\gstate(\access{\cstar}{\addr})}{\stor}
\end{aligned}$
\end{enumerate}
\end{lemma}
This lemma allows to relate the storage of contract $\cstar$ to prior executions of $\cstar$ itself. More precisely, the storage of $\cstar$ (given that $\cstar$ does not contain $\CALLCODE$ or $\DELEGATECALL$ instructions) either needs to be as it was at the point of the last call originating from $\cstar$ or the result of some finished prior execution of $\cstar$.

We sketch the proof of this crucial lemma, arguing about the most interesting cases:
\begin{proof} (sketch)
We proceed by complete induction on the number $n \in \NN$ of small-steps.
\begin{itemize}
\item
Case $n=0$. In this case it holds that $\annotate{\exstate'}{c'} = \annotate{\exstate}{\cstar}$ and $\callstack' = \emptystack$.
Hence the assumption that $c' \neq \cstar$ is trivially violated.
\item
Case $n > 0$. In this case
$\nsteps{\transenv}{\cons{\annotate{\exstate}{c^*}}{\callstack}}{\callstack''}{n-1}$
and $\sstep{\transenv}{\callstack''}{\cons{\annotate{\exstate'}{c'}}{\concatstack{\callstack'}{\callstack}}}$ for some $\callstack''$.
We proceed by case analysis on the small-step rule being applied in the last step.
\begin{itemize}
\item[$\ADD$] (non exception case).
Then $\callstack'' = \cons{\annotate{\regstatefull{\mstate}{\exenv}{\gstate}{\transeffects}}{c'}}{\concatstack{\callstack'}{\callstack}}$ for some $\mstate$, $\exenv$, $\gstate$, and $\transeffects$ and $\exstate' = \regstatefull{\mstate'}{\exenv}{\gstate}{\transeffects}$ for some $\mstate'$. By inductive hypothesis for $n-1$ it follows that one out of options \ref{case:exc} to \ref{case:return} holds for $\regstatefull{\mstate}{\exenv}{\gstate}{\transeffects}$. As the global state $\gstate$ stays unaffected, this consequently also holds for $\regstatefull{\mstate'}{\exenv}{\gstate}{\transeffects}$ hence closing the case.
This reasoning applies to all local rules that are not changing the contracts global storage (so except for $\SSTORE$).
\item[$\SSTORE$] (non exception case). Then $\callstack'' = \cons{\annotate{\regstatefull{\mstate}{\exenv}{\gstate}{\transeffects}}{c'}}{\concatstack{\callstack'}{\callstack}}$ for some $\mstate$, $\exenv$, $\gstate$, and $\transeffects$ and $\exstate' = \regstatefull{\mstate'}{\exenv}{\gstate'}{\transeffects}$ for some $\mstate'$, $\gstate'$. Since $\SSTORE$ only modifies the storage of the active account, we can conclude that for all addresses $a$ such that $a \neq \access{\exenv}{\activeaccount}$ it holds that $\gstate(a) = \gstate'(a)$.
Since by \autoref{lem:agreement} it holds that $\access{\exenv}{\activeaccount} \neq \access{\cstar}{\addr}$, it particularly follows that
$\gstate(\access{\cstar}{\addr}) = \gstate'(\access{\cstar}{\addr})$. Hence the claim follows immediately from the application of the inductive hypothesis for $n-1$.
\item[$\CALL$]
(All preconditions satisfied, called account exists).
Then $\callstack'' =\cons{\annotate{\exstate''}{c''}}{\concatstack{\callstack'''}{\callstack}}$ and $\callstack' = \cons{\annotate{\exstate''}{c''}}{\callstack'''}$ for some regular execution state $\exstate''$, contract $c''$ and callstack $\callstack''$. We do a case distinction on $c''$:
\begin{itemize}
\item [$c'' = \cstar$]
\sloppy
In this case, condition \ref{case:callstate} is satisfied since $\annotate{\exstate''}{c''} \in \callstack'$ and the call itself does not affect the contract's storage, so  $\access{\access{s'}{\gstate}(\access{\cstar}{\addr})}{\stor} = \access{\access{s''}{\gstate}(\access{\cstar}{\addr})}{\stor}$.
\item[$c'' \neq \cstar$]
In this case the inductive hypothesis is applicable for $n-1$. Given again that the call itself does not affect storage, so $\access{\access{s'}{\gstate}(\access{\cstar}{\addr})}{\stor} = \access{\access{s''}{\gstate}(\access{\cstar}{\addr})}{\stor}$, the claim straightforwardly propagates to the case of $n$ steps.
Note that similar reasoning also applies to the cases of $\STATICCALL$, $\CALLCODE$, and $\DELEGATECALL$.
\end{itemize}
\item[Halt] (return from regular halting after $\CALL$)
\sloppy
Then $\callstack'' = \cons{\annotate{\haltstatefull{\gstate}{\lgas}{d}{\transeffects}}{\dot{c}}}{\cons{\annotate{\exstate''}{c'}}{\concatstack{\callstack'}{\callstack}}}$ for some $\gstate$, $\lgas$, $d$, $\transeffects$, $\exstate''$, and $c''$. Additionally, it holds that $\access{\exstate'}{\gstate} = \gstate$. We do a case distinction on $\dot{c}$:
\begin{itemize}
\sloppy
\item[$\dot{c} = \cstar$] In this case, condition \ref{case:return} is satisfied since
$\nsteps{\transenv}{\cons{\annotate{\exstate}{c^*}}{\callstack}}{\cons{\annotate{\haltstatefull{\gstate}{\lgas}{d}{\transeffects}}{\dot{c}}}{\cons{\annotate{\exstate''}{c'}}{\concatstack{\callstack'}{\callstack}}}}{n-1}$,
$\sstep{\transenv}{\cons{\annotate{\haltstatefull{\gstate}{\lgas}{d}{\transeffects}}{\dot{c}}}{\cons{\annotate{\exstate''}{c'}}{\concatstack{\callstack'}{\callstack}}}}{\cons{\annotate{\exstate'}{c'}}{\concatstack{\callstack'}{\callstack}}}$ and
as $\access{\exstate'}{\gstate} = \gstate$, also $\access{\access{\exstate'}{\gstate}(\access{\cstar}{\addr})}{\stor} = \access{\gstate(\access{\cstar}{\addr})}{\stor}$.
\item[$\dot{c} \neq \cstar$]
In this case again the inductive hypothesis can be applied for $n-1$, and since $\access{\exstate'}{\gstate} = \gstate$, the claim trivially carries over to the case of $n$ steps.
\end{itemize}
\item[Exc] (return from exceptional halting)
\sloppy
Then $\callstack'' = \cons{\annotate{\excstate}{\dot{c}}}{\cons{\annotate{\exstate''}{c'}}{\concatstack{\callstack'}{\callstack}}}$ and $\access{\exstate'}{\gstate} = \access{\exstate''}{\gstate}$ (as the global state is rolled back).
By \autoref{lemma:step-through-call}, we know that there exists some $m < n -1$ such that $\nsteps{\transenv}{\cons{\annotate{\exstate}{c^*}}{\callstack}}{\cons{\annotate{\exstate''}{c'}}{\concatstack{\callstack'}{\callstack}}}{m}$ and
$\nsteps{\transenv}{\cons{\annotate{\exstate''}{c'}}{\concatstack{\callstack'}{\callstack}}}{ \cons{\annotate{\excstate}{\dot{c}}}{\cons{\annotate{\exstate''}{c'}}{\concatstack{\callstack'}{\callstack}}}}{n-1-m}$.
By applying the inductive hypothesis for $m$ ($< n$), the claim straightforwardly carries over to the case of $n$ steps.
\item [$\CREATE$]
(All preconditions satisfied, created account does not exist (no hash collision).
Then $\callstack'' =\cons{\annotate{\exstate''}{c''}}{\concatstack{\callstack'''}{\callstack}}$ and $\callstack' = \cons{\annotate{\exstate''}{c''}}{\callstack'''}$ for some regular execution state $\exstate''$, contract $c''$ and callstack $\callstack''$. The same reasoning as for the $\CALL$ case applies.
\item [$\CREATE$]
(All preconditions satisfied, created account exists (hash collision).
Then $\callstack'' =\cons{\annotate{\exstate''}{c''}}{\concatstack{\callstack'''}{\callstack}}$ and $\callstack' = \cons{\annotate{\exstate''}{c''}}{\callstack'''}$ for some regular execution state $\exstate''$, contract $c''$ and callstack $\callstack''$.
Additionally, we know that $c' = (\rho, \bot)$ where $\rho$ is the newly created address.
Here, we need to make use of the assumption that the newly created address $\rho$ is not colliding with the address of $\cstar$ ($\rho \neq \access{\cstar}{\addr}$). This is ensured as otherwise the execution would not be collision-free (for $\rho = \access{\cstar}{addr}$ the condition of \autoref{def:coll-free} would be violated for $\annotate{\exstate'}{c'}$).
We do a case distinction on $c''$:
\begin{itemize}
\item [$c'' = \cstar$]
\sloppy
In this case, condition \ref{case:callstate} is satisfied since $\annotate{\exstate''}{c''} \in \callstack'$ and the contract creation does not affect the storage of address $\access{\cstar}{\addr}$ (but only the one of $\rho$, which is different by assumption). So $\access{\access{s'}{\gstate}(\access{\cstar}{\addr})}{\stor} = \access{\access{s''}{\gstate}(\access{\cstar}{\addr})}{\stor}$.
\item[$c'' \neq \cstar$]
In this case the inductive hypothesis is applicable for $n-1$. Given again that the contract creation does not affect the storage of address $\access{\cstar}{\addr}$, so $\access{\access{s'}{\gstate}(\access{\cstar}{\addr})}{\stor} = \access{\access{s''}{\gstate}(\access{\cstar}{\addr})}{\stor}$, the claim straightforwardly propagates to the case of $n$ steps.
\end{itemize}
\item[Halt] (return from regular halting after $\CREATE$)
Similar to the halting case for $\CALL$ with the only difference that instead of  $\access{\exstate'}{\gstate}=\gstate$ it only holds that $\access{\exstate'}{\gstate} = \updategstate{\gstate}{\rho}{\textit{acc}}$ for some account state $\textit{acc}$. However, as long as it is ensured that $\rho \neq \access{\cstar}{\addr}$ (which is the case due to the collision-free execution) this does not affect the reasoning.
\end{itemize}
\end{itemize}
\end{proof}

\subsubsection{Monotonicity of abstract rules}
We prove separately, that all rules in $\delta(\cstar)$ are monotone (for any $\cstar$).
This facilitates the reasoning in the individual cases of the main proof, since it allows us to argue about most concrete abstractions only.

Since monotoniticy is independent of the small-step semantics, we will in the following consider an abstract semantics specified by $(\absdoms, \predsig, \absts)$.
First, we define  monotonicity for an abstract semantics $(\absdoms,  \predsig, \absts)$ as follows:
\begin{definition}[Monotonicity of abstract Semantics]
An abstract semantics $(\absdoms,  \predsig, \absts)$ is monotone if for all abstract configurations $\absconfig_I$, $\absconfig_I'$, $\absconfig_F\in \absconfigs_\predsig$ such that $\absconfig_I \leq \absconfig_I'$ it holds that
\begin{align*}
\absconfig_I \cup \absts \derives \absconfig_F
\implies \exists \absconfig_F'.~ \absconfig_I' \cup \absts \derives \absconfig_F' \land \absconfig_F \leq \absconfig_F'
\end{align*}
\end{definition}

We will prove the following theorem:
\begin{theorem}[Monotonicity of $\delta$]
\label{thm:delta-monotonicity}
For all contracts $c$ it holds that $(\absdoms_\textit{evm}, \predsig_\textit{evm}, \delta(c))$ is monotone.
(Where $\absdoms_\textit{evm}$ is the super domain and $\predsig_\textit{evm}$ is the signature induced by the definition in \autoref{fig:pred-signature}.)
\end{theorem}

We prove this property by proving (one-step) monotonicity of the individual rules in $\delta(c)$.

We define one-step derivations of a \hornc clause $\hornclause$ from some abstract configuration $\absconfig$.
To this end, we use the notion of a variable assignment $\valuation \in \vars \to \absdoms$ that maps two variables to values of the corresponding abstract domain. We write $\valuation(n(\vec{z}))$ for $n(\vec{V(z)})$ and $\valuation(\{f_1, \dots, f_n \})$ for $\{ \valuation(f_1), \dots, \valuation(f_n) \}$.
By $\valuation \satisfies \constraints$ we denote that replacing all variables in $\constraints$ according to $\valuation$ yields a tautology.

\begin{definition}[One-step derivability from horn clause]
Let $(\absdoms, \predsig, \absts)$ be an abstract semantics and $(\hc{X}{\constraints}{\premises}{\conclusion}) \in \absts$.
Further let $f \in \absconfigs_\predsig$
Then the one-step derivability relation $\derivesone$ on abstract configurations is defined as follows:
\begin{align*}
\absconfig, (\hc{X}{\constraints}{\premises}{\conclusion}) \derivesone f
\define
\exists \valuation.~ \valuation(P) \subseteq \absconfig ~\land~ \valuation \satisfies \constraints ~\land~ f = \valuation(\conclusion)
\end{align*}
\end{definition}
Note that this intuition implicitly enforces that the valuation $\valuation$ respects the argument types of the predicates.

We extend the notion of derivability to sets of horn clauses and abstract configurations:

\begin{definition}[One-step derivability from abstract semantics]
Let $(\absdoms, \predsig, \absts)$ be an abstract semantics. Then the one-step derivability relation $\derivesone$ on $\absts$ is defined as follows
\begin{align*}
\absconfig, \abss \derivesone \absconfig'
\define
\exists f.~ \absconfig = \absconfig \cup \{ f \}
~\land~  \exists \hornclause \in \abss.~ \absconfig, \hornclause \derivesone f
\end{align*}
\end{definition}
Finally, we define $\derives$ to be the reflexive, transitive closure of $\derivesone$.

We define the monotonicity of a \hornc clause as follows:
\begin{definition}[Monotonicity of \hornc clauses]
Let $(\absdoms, \predsig, \absts)$ be an abstract semantics.
A constrained \hornc clause $\hornclause \in \absts$ is monotone if for all $\absconfig' \geq \absconfig$
\begin{align*}
\absconfig, \hornclause \derivesone f \implies \exists f'.~ \absconfig', \hornclause \derivesone f' ~\land f' \geq f
\end{align*}
\end{definition}

Evidently, the (one-step) monotonicity of all \hornc clauses in an abstract semantics implies the (multi-step) monotonicity of the abstract semantics
\begin{lemma}
Let $(\absdoms, \predsig, \absts)$ be an abstract semantics.
If all constrained horn clauses $\hornclause \in \abss$ are monotone, then so is $\abss$.
\end{lemma}

It is hence sufficient to prove the (one-step) monotonicity of all \hornc clauses in $(\absdoms_\textit{evm}, \predsig_\textit{evm}, \delta(c))$ (for arbitrary $c$).

For facilitating the proofs, we give a more syntactic characterization of \hornc clause monotonicity:
\begin{lemma}
\label{lem:rule-monotone-syntax}
Let $\hornclause = \hc{X}{\constraints}{\premises}{\conclusion}$ be a \hornc clause.
If for all variable assignments $\valuation$, $\valuation'$ with $(x, D) \in X \implies \valuation(x) \in D ~\land \valuation' (x) \in D$  it holds that
\begin{align*}
\valuation'(\premises) \geq \valuation(\premises) \land \valuation \satisfies \constraints \\
\implies \exists \valuation^* .~ \valuation^*(\premises) = \valuation'(\premises) \land \valuation^*(\conclusion) \geq \valuation(\conclusion) \land \valuation^* \satisfies \constraints
\end{align*}
then $\hornclause$ is monotone.
\end{lemma}
\begin{proof}
Assume that (1)
\begin{align*}
\valuation'(\premises) \geq \valuation(\premises) \land \valuation \satisfies \constraints \\
\implies \exists \valuation^* .~ \valuation^*(\premises) = \valuation'(\premises) \land \valuation^*(\conclusion) \geq \valuation(\conclusion) \land \valuation^* \satisfies \constraints
\end{align*}
holds for valuations as defined above. We show the monotonicity of $\hornclause= \hc{X}{\constraints}{\premises}{\conclusion}$. To this end we assume some (2) $\absconfig\ \geq \absconfig$ and (3) $\absconfig, \hornclause \derivesone f$ and show that there is some valuation $\valuation'$ such that $\valuation'(P) \subseteq \absconfig'$, $\valuation' \satisfies \constraints$ and $\valuation'(\conclusion) \geq f$.
From (3) it is known that there is some valuation $\valuation$ such that $\valuation(P) \subseteq \absconfig$, $\valuation \satisfies \constraints$ and $f = \valuation(c)$. From (2), we get that for every $p \in \valuation(P)$ there exists a $p' \in \absconfig'$ such that $p \leq p'$. Given that the variables of all premises are distinct, we can easily construct a valuation $\valuation'$ such that $\valuation'(q) = p$ for some $q \in \premises$ and consequently $\valuation'(P) \subseteq \absconfig'$ and $\valuation(P) \leq \valuation'(P)$. Using (1), we get that there is some $\valuation^*$ such that $\valuation^*(P) = \valuation'(P)$ and $\valuation^*(\conclusion) \geq \valuation(\conclusion)$ and $\valuation^* \satisfies \constraints$.
Consequently, since $\valuation^*(P) =  \valuation'(P) \subseteq \absconfig'$ and $\valuation^*(\conclusion) \geq \valuation(\conclusion) =f$, $\valuation^*$ satisfies all required conditions.
\end{proof}
This lemma reduces proving monotonicity of the constrained \hornc clause to proving the monotonicity of the clause's constraints.

\subsubsection{Abstract operations}
We exemplary show the monotonicity of the rules shown in \autoref{fig:absrules}.
To this end we will first establish some general monotonicity results on abstract operations.

\begin{lemma}[Monotonicity of abstract binary operations]
\label{lem:abs-binop-monotone}
Let $\abs{x}$, $\abs{x'}$, $\abs{y}$, $\abs{y'}$ such that $\abs{x} \leq \abs{x'}$ and $\abs{y} \leq \abs{y'}$. Then
\begin{align*}
\absop{\binop}(\abs{x}, \abs{y}) \ord{\absdom} \absop{\binop}(\abs{x'}, \abs{y'})
\end{align*}
\end{lemma}

\begin{lemma}[Monotonicity of abstract comparison operations]
Let $\abs{x}$, $\abs{x'}$, $\abs{y}$, $\abs{y'}$ such that $\abs{x} \ord{\absdom} \abs{x'}$ and $\abs{y} \ord{\absdom} \abs{y'}$. Then
\begin{align*}
\absop{\compop}(\abs{x}, \abs{y}) = 1 \implies \absop{\compop}(\abs{x'}, \abs{y'}) = 1
\end{align*}
\end{lemma}

\begin{lemma}[Monotonicity of memory access]
\label{lem:mem-access-monotone}
Let $\mem_1, \mem_2 \in \NN \to \absdom$ such that $\mem_1 \ord{\NN \to \absdom} \mem_2$ and let $p \in \NN$.
\begin{align*}
\accessword{\mem_1}{p} \ord{\absdom} \accessword{\mem_2}{p}
\end{align*}
\end{lemma}

We now give a proof sketch for \autoref{thm:delta-monotonicity}, illustrating the general proof strategy.
\begin{proof}
For showing the monotonicity of $\delta(\cstar)$ for arbitrary $\cstar$ it is sufficient to show the one-step derivability of all rules in $\instabs{\instr}_\lpc$ for all instructions $\instr$ and an arbitrary program counter $\lpc$. Hence, let $\lpc \in \NN$ be arbitrary.
The proof proceeds by case distinction on the instruction set.
\begin{itemize}
\item[$\ADD$]
We now prove the monotonicity of the rules for addition in \autoref{fig:absrules}.
Recall the definition of the clause for addition.
\begin{align*}
& \predmstate{\pc}{\siz}{\stackv}{\mem}{\stor}{\cl}
~\land~ \siz > 1 \\
&~\land~ \abs{x} = \select{\stackv}{\siz-1}
~\land~ \abs{y} = \select{\stackv}{\siz-2} \\
&\implies
\predmstate{\pc +1}{\siz-1}{\update{\stackv}{\siz-2}{\abs{x} \absop{+} \abs{y}}}{\mem}{\stor}{\cl}
\end{align*}
We prove the monotonicity using \autoref{lem:abs-binop-monotone}.
Assume that there is some variable assignment satisfying the rule constraints, meaning that there are values $(\siz, \stackv)$, $\mem$, $\stor$, $\cl$, $\abs{x}$, $\abs{y}$ satisfying $\siz > 1$, $\abs{x} = \select{\stackv}{\siz-1} $ and $\abs{y} = \select{\stackv}{\siz-2}$. We show for any values $(\siz',\stack') \dro{\NN \times (\NN\to \absdom)} (\siz, \stackv)$, $\mem' \dro{\NN \to \absdom} \mem$, $\stor' \dro{\NN \to \absdom} \stor$, $\cl' \dro{\BB} \cl$ that there are $\abs{x'}$, $\abs{y'}$ such that $\siz' > 1$, $\abs{x'} = \select{\stackv'}{\siz'-1} $ and $\abs{y'} = \select{\stackv'}{\siz'-2}$, and
$(\siz-1, \update{\stackv}{\siz-2}{\abs{x} \absop{+} \abs{y}}) \ord{\NN \times (\NN \to \absdom)} (\siz'-1, \update{\stackv'}{\siz'-2}{\abs{x'} \absop{+} \abs{y'}})$.
First we observe that $\siz = \siz'$ and (since $(\siz, \stackv) \ord{\NN \times (\NN\to \absdom)} (\siz', \stackv')$).
We pick $\abs{x'} = \select{\stackv'}{\siz-1} $ and $\abs{y'} = \select{\stackv'}{\siz-2}$ and from $(\siz, \stackv) \ord{\NN \times (\NN\to \absdom)} (\siz', \stackv')$ we know that $\select{\stackv}{\siz-1}  \ord{\absdom} \select{\stackv'}{\siz-1}$ and $\select{\stackv}{\siz-2} \ord{\absdom} \select{\stackv'}{\siz-2}$, so consequently also $\abs{x} \ord{\absdom} \abs{x'}$ and $\abs{y}\ord{\absdom} \abs{y'}$.
So we are left to show that
$(\siz -1, \update{\stackv}{\siz-2}{\abs{x} \absop{+} \abs{y}}) \ord{\NN \times (\NN\to \absdom)} (\siz'-1, \update{\stackv'}{\siz-2}{\abs{x'} \absop{+} \abs{y'}})$. Since $(\siz, \stackv) \ord{\NN \times (\NN\to \absdom)} (\siz', \stackv')$, we only need to show that $\abs{x} \absop{+} \abs{y} \ord{\absdom} \abs{x'} \absop{+} \abs{y'}$ which immediately follows from \autoref{lem:abs-binop-monotone}.
\item[$\MLOAD$]
Recall the definition of the rule for addition:
\begin{align*}
&\predmstate{\pc}{\siz}{\stackv}{\mem}{\stor}{\cl}
~\land~ \siz > 1 \\
&~\land~ \abs{o} = \select{\stackv}{\siz-1}
~\land \abs{v} = \cond{(\abs{o} \in \NN)}{\accessword{\mem}{\abs{o}}}{\top}  \\
&\implies
\predmstate{\pc +1}{\siz}{\update{\stackv}{\siz-1}{\abs{v}}}{\mem}{\stor}{\cl}
\end{align*}
We prove the monotonicity using \autoref{lem:abs-binop-monotone}.
Assume that there is some variable assignment satisfying the rule constraints, meaning that there are values $(\siz, \stackv)$, $\mem$, $\stor$, $\cl$, $\abs{o}$, $\abs{v}$ satisfying
$\siz > 0$, $\abs{o} = \select{\stackv}{\siz-1}$, and $\abs{v} = \cond{(\abs{o} \in \NN)}{\accessword{\mem}{\abs{o}}}{\top}$.
We show for any values $(\siz',\stack') \dro{\NN \times (\NN\to \absdom)} (\siz, \stackv)$, $\mem' \dro{\NN \to \absdom} \mem$, $\stor' \dro{\NN \to \absdom} \stor$, $\cl' \dro{\BB} \cl$ that there are $\abs{o'}$, $\abs{v'}$ such that $\siz' > 1$, $\abs{o'} = \select{\stackv'}{\siz'-1}$ and $\abs{v'} = \cond{(\abs{o'} \in \NN)}{\accessword{\mem}{\abs{o'}}}{\top}$, and $\update{\stackv}{\siz-1}{\abs{v}}\ord{\NN \times (\NN \to \absdom)} \update{\stackv'}{\siz'-1}{\abs{v'}}$.
First we observe $\siz = \siz'$ and $\cl = \cl'$.
We pick $\abs{o'} = \select{\stackv'}{\siz'-1}$ and $\abs{v'} = \cond{(\abs{o'} \in \NN)}{\accessword{\mem'}{\abs{o'}}}{\top}$. We know that $\select{\stackv}{\siz-1} \ord{\absdom} \select{\stackv'}{\siz'-1}$ since $(\siz',\stack') \dro{\NN \times (\NN\to \absdom)} (\siz, \stackv)$ and hence also $\abs{o} \ord{\absdom} \abs{o'}$.
For showing that $(\siz,\update{\stackv}{\siz-1}{\abs{v}})\ord{\NN \times (\NN \to \absdom)} (\siz, \update{\stackv'}{\siz'-1}{\abs{v'}})$
it is sufficient to show that $\abs{v} \ord{\absdom} \abs{v'}$.
We make a case distinction on $\abs{o} \in \NN$
\begin{itemize}
\item[$\abs{o} \in \NN$]
In this case $\abs{v} = \accessword{\mem}{\abs{o}}$. Since $\abs{o} \ord{\absdom} \abs{o'}$ we know that either $\abs{o'} = \abs{o}$ or $\abs{o'} = \top$.
\begin{itemize}
\item[$\abs{o'} = \abs{o}$]
In this case clearly $\abs{o'} \in \NN$ and hence $\abs{v'} = \accessword{\mem'}{\abs{o'}}$.
Since $\mem \ord{\NN \to \absdom} \mem'$, we know from \autoref{lem:mem-access-monotone} that $\accessword{\mem}{\abs{o}} \ord{\absdom} \accessword{\mem'}{\abs{o'}}$ and hence $\abs{v} \ord{\absdom} \abs{v'}$.
\item[$\abs{o'} = \top$]
In this case  $\abs{v'} = \top$. Since $\top$ is the top element of $\absdom$ (\autoref{lem:suprema}), trivially $\abs{v} \ord{\absdom} \abs{v'}$.
\end{itemize}
\item[$\abs{o} = \top$] In this case $\abs{v} = \top$ and since $\abs{o} \ord{\absdom} \abs{o'}$ also $\abs{o'} = \top$ and hence $\abs{v'} = \top$ and consequently $\abs{v} \ord{\absdom} \abs{v'}$.
\end{itemize}
\end{itemize}
\end{proof}

\subsubsection{Soundness of abstract operations}
In addition to their monotonicity, we are also interested in the soundness of abstract operations.
Intuitively, an abstract operation is sound, if its result is at least as abstract than the result of the concrete operation
We formally state soundness for binary operations and comparison operations.

\begin{lemma}[Soundness of abstract binary operations]
\label{lem:abs-binop-sound}
Let $x, y \in \NN$. Then
\begin{align*}
\binop(x, y) \ord{\absdom} \absop{\binop}(x, y)
\end{align*}
\end{lemma}

\begin{lemma}[Soundness of abstract comparison operations]
Let $x, y \in \NN$. Then
\begin{align*}
\compop(x, y) = 1 \implies \absop{\compop}(x, y) = 1
\end{align*}
\end{lemma}

That the memory access is sound, is captured by the following lemma:
\begin{lemma}[Soundness of memory access]
Let $\mem \in \NN \to \NN$ and $p \in \NN$.
\begin{align*}
 \mem(p) \ord{\absdom}\accessword{\towordmem(\mem)}{p}
\end{align*}
\end{lemma}

\subsubsection{Main Proof}
We slightly refine \autoref{theorem:soundness} to consider collision-free executions of $\cstar$, a detail that we omitted in the original formulation for the sake of presentation.
\begin{theorem-non}[Soundness]
Let $\cstar$ be a contract whose code does not contain $\DELEGATECALL$ or $\CALLCODE$. Let $\transenv$ be a transaction environment and let $\callstack$ and $\callstack'$ be annotated callstacks such that $\size{\callstack'} > 0$. Then for all execution states $\exstate$ that are strongly consistent with $\cstar$ such that $\ssteps{\transenv}{\cons{\annotate{\exstate}{\anacontract}}{\callstack}}{\concatstack{\callstack'}{\callstack}}$ is a collision-free execution, it holds that
\par
\nobreak
{
\noindent
\begin{align*}
\forall \absconfig_I.~ \configabs{\predsig}_{\cstar}([\annotate{\exstate}{\cstar}]) \leq \absconfig_I
\implies
\exists \absconfig.~ \absconfig_I, \delta(\cstar) \derives \absconfig \ ~\land~ \configabs{\predsig}_{\cstar}(\callstack') \leq \absconfig
\end{align*}
}%
\end{theorem-non}

We will give a proof sketch for the most interesting cases of the soundness proof, providing formal arguments for the soundness of local operations as well as transaction-initiating instructions. In particular, the proof details out the correctness argument for the abstractions of the $\CALL$ rule by covering call initiation as well as returning from contracts calls.
\begin{proof} (sketch)
By complete induction on the number $n$ of small-steps.
\begin{itemize}
\item Case $n=0$. In the case of the empty reduction sequence, we have that
$\callstack' = [\annotate{\exstate}{\anacontract}]$ and consequently the claim trivially follows by the reflexivity of $\vdash$.
\item Case $n >0$.
Let $\nsteps{\transenv}{\cons{\annotate{\exstate}{c^*}}{\callstack}}{\callstack''}{n-1}$ and $\sstep{\transenv}{\callstack''}{\concatstack{\callstack'}{\callstack}}$. By Lemma~\ref{lem:callstack-preservation}, it holds that $\callstack'' = \concatstack{\callstack^*}{\callstack}$ for some $\callstack^*$ with $\size{\callstack^*} > 0$.
By the inductive hypothesis we know that for all $\absconfig_I \geq \configabs{}_{\cstar}([\annotate{\exstate}{\anacontract}])$ there is some $\absconfig_{\callstack^*} \geq \configabs{}_{\cstar}(\callstack^*)$ such that $\absconfig_I \cup \delta(\cstar) \derives \absconfig_{\callstack^*}$.
Consequently, for proving the claim, it is sufficient to show that there is some $\absconfig_{\callstack'}\geq \configabs{}_{\cstar}(\callstack')$ such that $\absconfig_{\callstack^*} \cup \delta(\cstar) \derives \absconfig_{\callstack'}$.
As $\size{\callstack^*} > 0$, we know that $\callstack^* = \cons{\annotate{s'}{c'}}{\callstack^{**}}$ for some execution state $s'$, contract $c'$ and callstack $\callstack^{**}$.
The proof is by case analysis on the rule applied in the last reduction step. We show here exemplary the cases for arithmetic operations as well as the rule for calling.
\begin{itemize}
\item[$\ADD$] (non exception case).
\sloppy
Then $\exstate' = \regstatefull{\mstate}{\exenv}{\gstate}{\transeffects}$, $\arraypos{\access{\exenv}{\activecode}}{\access{\mstate}{\pc}} = \ADD$ and $\callstack' = \cons{\annotate{\regstatefull{\mstate'}{\exenv}{\gstate}{\transeffects}}{c'}}{\callstack^{**}}$.
We distinguish the two cases on whether the top stack element $\annotate{\exstate'}{c'}$ is translated or not ($c' = \cstar$)
\begin{itemize}
\item[$c' \neq \cstar$]
In this case $\configabs{}_{\cstar}(\callstack^*) = \configabs{}_{\cstar}(\callstack^{**})$.
As $\ADD$ is a local instruction, we know that $\callstack' = \cons{\annotate{s''}{c'}}{\callstack^{**}}$ and hence also $\configabs{}_{\cstar}(\callstack') = \configabs{}_{\cstar}(\callstack^{**})$. The claim hence follows trivially from the reflexivity of $\derives$. The same reasoning applies to all other local instructions.
\item[$c' = \cstar$]
In this case $\configabs{}_{\cstar}(\callstack^*) = \configabs{}_{s}(s', \access{\cstar}{\addr}, \cl) \cup  \configabs{}_{\cstar}(\callstack^{**}) $ for some $\cl \in \BB$.
As $s'$ is strongly consistent with $\cstar$ (by \autoref{lem:agreement}), we know that $\access{\exenv}{\code} = \access{\cstar}{\code}$ and hence $\delta(\cstar) \supseteq \instabs{\ADD}_{\access{\mstate}{\pc}}$. The claim then follows from the monotonicity of $\delta(\cstar)$ (\autoref{thm:delta-monotonicity}) and the soundness of abstract addition (\autoref{lem:abs-binop-sound}). The same argumentation applies to all other local operations.
\end{itemize}
\item[$\CALL$]
\sloppy
(all preconditions satisfied, called account exists). Then $\exstate' = \regstatefull{\mstate}{\exenv}{\gstate}{\transeffects}$, $\arraypos{\access{\exenv}{\activecode}}{\access{\mstate}{\pc}} = \CALL$ and $\callstack' = \cons{\annotate{\regstatefull{\mstate'}{\exenv'}{\gstate'}{\transeffects}}{\dot{c}}}{\callstack^*}$  such that $\mstate'$ is initial, and $\access{\gstate(a)}{\stor} = \access{\gstate'(a)}{\stor}$ for all addresses $a$.
Again we distinguish the cases whether the newly pushed callstack element $\annotate{\regstatefull{\mstate'}{\exenv'}{\gstate'}{\transeffects}}{\dot{c}}$ is abstracted by $\alpha$ or not.
\begin{itemize}
\item[$\dot{c} \neq \cstar$]
\sloppy
Then $\configabs{}_{\cstar}(\callstack') = \configabs{}_{\cstar}(\callstack^*)$ and the claim trivially holds.
\item[$\dot{c} = \cstar$]
We do another case distinction on whether $c' = \cstar$
\begin{itemize}
\item[$c' = \cstar$] In this case, we know that $\alpha_\exstate(\exstate', \access{\addr}{\cstar}, \cl) \leq \absconfig_{\callstack^*}$ (where $\cl = (\callstack^{**} \neq \nil)$). Since $\exstate'$ is strongly consistent with $\cstar$ (by \autoref{lem:agreement}), we have that $\access{\exenv}{\code} = \access{\cstar}{\code}$ and hence $ \instabs{\CALL}_{\access{\mstate}{\pc}} \subseteq \delta(\cstar)$.
Since $\exstate'$ is a call state, we have that
$\instabs{\CALL}_{\access{\mstate}{\pc}} \cup \alpha_\exstate(\exstate', \access{\addr}{\cstar}, \cl) \derives
\{ \predmstate{0}{0}{\fun{x}{0}}{\fun{x}{0}}{\access{\gstate(\access{\cstar}{\addr})}{\stor}}{1} \} $.
As $\mstate'$ is initial and $\access{\gstate(a)}{\stor} = \access{\gstate'(a)}{\stor}$, we know additionally that
$\{ \predmstate{0}{0}{\fun{x}{0}}{\fun{x}{0}}{\access{\gstate(\access{\cstar}{\addr})}{\stor}}{1} \} = \alpha_\exstate(\regstatefull{\mstate'}{\exenv'}{\gstate'}{\transeffects}, \access{\cstar}{\addr}, \cl')$ (for $\cl' = (\callstack^* \neq \nil)$).
By the monotonicity of $\delta(\cstar)$ (\autoref{thm:delta-monotonicity}), we know that there is also some $\absconfig_x \geq \alpha_\exstate(\regstatefull{\mstate'}{\exenv'}{\gstate'}{\transeffects}, \access{\cstar}{\addr}, \cl')$ such that $\absconfig_{\callstack^*}, \delta(\cstar) \derives \absconfig_x$ which concludes the proof since
\begin{align*}
\absconfig_{\callstack^*} \cup \delta(\cstar) &\derives \absconfig_x \cup \absconfig_{\callstack^*} \\
&\geq \alpha_\exstate(\regstatefull{\mstate'}{\exenv'}{\gstate'}{\transeffects}, \access{\cstar}{\addr}, \cl') \cup \alpha(\callstack^*) \\
&= \alpha(\callstack')
\end{align*}
\item [$c' \neq \cstar$]
By \autoref{lem:storage-evolution}, we know (since $\exstate'$ is a regular execution state) that either (1) there exists some $\annotate{\exstate^*}{\cstar} \in \callstack^{**}$ such that $\access{\access{\exstate'}{\gstate}(\access{\cstar}{\addr})}{\stor} = \access{\access{\exstate^*}{\gstate}(\access{\cstar}{\addr})}{\stor}$ or (2) there exist $\callstack^{\dagger}$, $\gstate^*$, $\lgas^*$, $d^*$, $\transeffects^*$, and $m < n$ such that $\nsteps{\transenv}{\cons{\annotate{\exstate}{\cstar}}{\callstack}}{\cons{\annotate{\haltstatefull{\gstate^*}{\lgas^*}{d^*}{\transeffects^*}}{\cstar}}{\concatstack{\callstack^{\dagger}}{\callstack}}}{m}$ and $\nsteps{\transenv}{\cons{\annotate{\haltstatefull{\gstate^*}{\lgas^*}{d^*}{\transeffects^*}}{\cstar}}{\concatstack{\callstack^{\dagger}}{\callstack}}}{\cons{\annotate{\regstatefull{\mstate'}{\exenv'}{\gstate'}{\transeffects}}{\dot{c}}}{\cons{\annotate{\exstate'}{c'}}{\concatstack{\callstack^{**}}{\callstack}}}}{n-1-m}$ and $\access{\access{\exstate'}{\gstate}(\access{\cstar}{\addr})}{\stor} = \access{\gstate^*(\access{\cstar}{\addr})}{\stor}$.
Additionally, we know that then $\access{\gstate^*(\access{\cstar}{\addr})}{\stor} =\access{\gstate'(\access{\cstar}{\addr})}{\stor}$.
We make a distinction on the previously mentioned cases:
\begin{enumerate}
\item
\sloppy
In this case we know that $\alpha_{\exstate}(\exstate^*, \access{\cstar}{\addr}, \cl^*) \subseteq \alpha_{\cstar}(\callstack^{**}) =\alpha_{\cstar}(\callstack^{*})$ for some $\cl^* \in \BB$. Since, we know that $\exstate^*$ is a call state (\autoref{lem:callstates-but-top}), we know that $\exstate^* = \regstatefull{\mstate^*}{\exenv^*}{\gstate^*}{\transeffects^*}$ for some $\mstate^*$, $\exenv^*$, $\gstate^*$, and $\transeffects^*$ such that the conditions in \autoref{def:callstate} are satisfied.
Since $\exstate^*$ is a call state, $\access{\exenv^*}{\code}[\access{\mstate^*}{\pc}] = \CALL$\footnote{\label{note:call}This is a simplifying assumption made here. Actually $\access{\exenv^*}{\code}[\access{\mstate^*}{\pc}] \in \{ \CALL, \STATICCALL, \CREATE \}$. Since, the abstract semantics of these instructions have the same rules (up to minor differences in the preconditions of calling), exactly the same argumentation applies as shown here for the case of $\CALL$.}.
As $\exstate^*$ is strongly consistent with $\cstar$ (by \autoref{lem:agreement}), also $\access{\exenv^*}{\code}=\access{\cstar}{\code}$ and hence $\delta(\cstar) \supseteq \instabs{\CALL}_{\pc^*}$.
In particular, the second abstract $\CALL$ rule (\ref{rule:call-2}) is applicable on $\alpha_{\exstate}(\exstate^*, \access{\cstar}{\addr}, \cl^*) \subseteq \alpha_{\cstar}(\callstack^{**})$ and hence one can derive $\predmstate{0}{0}{\fun{x}{0}}{\fun{x}{0}}{\access{\gstate^*(\access{\cstar}{\addr})}{\stor}}{1}$.
Additionally, we have that $\alpha_{\exstate}(\regstatefull{\mstate'}{\exenv'}{\gstate'}{\transeffects}, \access{\cstar}{\addr}, (\callstack^*\neq \nil)) = \predmstate{0}{0}{\fun{x}{0}}{\fun{x}{0}}{\access{\gstate'(\access{\cstar}{\addr})}{\stor}}{1}$ (since, $\mstate'$ is an initial machine state and $\callstack^*$ is non-empty).
Together with $\access{\access{\exstate'}{\gstate}(\access{\cstar}{\addr})}{\stor} = \access{\access{\exstate^*}{\gstate}(\access{\cstar}{\addr})}{\stor}$ and $\access{\gstate(a)}{\stor} = \access{\gstate'(a)}{\stor}$ for all $a$ (since the call rule does not effect a contract's storage), we can conclude that $\alpha_{\exstate}(\exstate^*, \access{\cstar}{\addr}, \cl^*) \cup \delta(\cstar) \derives \alpha_{\exstate}(\regstatefull{\mstate'}{\exenv'}{\gstate'}{\transeffects}, \access{\cstar}{\addr}, 1)$.
Due to the monotonicity of $\delta(\cstar)$ (\autoref{thm:delta-monotonicity}), we know that there is some $\absconfig_i \geq \alpha_{\exstate}(\regstatefull{\mstate'}{\exenv'}{\gstate'}{\transeffects},\access{\cstar}{\addr}, 1)$, such that
$\absconfig_{\callstack^{*}} \cup \delta(\cstar) \derives \absconfig_i$ (since $\absconfig_{\callstack^{**}} \geq \alpha(\callstack^{**}) \supseteq \alpha_{\exstate}(\exstate^*, \access{\cstar}{\addr}, \cl^*)$.
Consequently:
\begin{align*}
\absconfig_{\callstack^*} \cup \delta(\cstar) &\derives \absconfig_{\callstack^*} \cup \absconfig_i \\
&\geq \alpha(\callstack^{*}) \cup \alpha_{\exstate}(\regstatefull{\mstate'}{\exenv'}{\gstate'}{\transeffects}, \access{\cstar}{\addr}, 1) \\
&= \alpha(\cons{\regstatefull{\mstate'}{\exenv'}{\gstate'}{\transeffects}}{\callstack^{*}}) \\
&= \alpha(\callstack')
\end{align*}
\item In this case, we get from the inductive hypothesis for $m$ (since $m < n$) that there exists some $\absconfig_H$
such that $\absconfig_H \geq \alpha(\cons{\annotate{\haltstatefull{\gstate^*}{\lgas^*}{d^*}{\transeffects^*}}{\cstar}}{\callstack^\dagger})$ and $\absconfig_I \cup \delta(\cstar) \derives \absconfig_H$, and additionally $\size{\callstack^\dagger} > 0$.
Consequently also $\absconfig_H \geq  \{ \predhalt{\access{\gstate^*(\access{\cstar}{\addr})}{\stor}}{1} \} \cup \alpha(\callstack^\dagger)$
Additionally we know that $\callstack^{**} =  \concatstack{\callstack_1}{[\annotate{\exstate^1}{\cstar}]}$ for some $\callstack_1$ and some $\exstate^1$ from \autoref{lem:annotation-persistence} (since the first state on top of $\callstack$ needs to be annotated with $\cstar$). Additionally we can conclude from \autoref{lem:callstates-but-top} that $\exstate^1$ is a call state.
From \autoref{lem:agreement}, we know that $\exstate^1$ is strongly consistent with $\cstar$ and hence $\access{\access{\exstate^1}{\exenv}}{\code} = \access{\cstar}{\code}$. As $\exstate^1$ is a call state, hence also $ \access{\cstar}{\code}[\access{\access{\exstate^1}{\mstate}}{\pc}] = \CALL$ and consequently $\delta(\cstar) \supseteq \instabs{\CALL}_{\access{\access{\exstate^1}{\mstate}}{\pc}}$\footnote{See \autoref{note:call}}. In addition we have that $\alpha_{\exstate}(\exstate^1, \access{\cstar}{\addr}, 0) \leq \absconfig_{\callstack^{**}}$ and since $\exstate^1$ is a call state all pre conditions of rule \ref{rule:call-3} in $\instabs{\CALL}_{\access{\access{\exstate^1}{\mstate}}{\pc}}$ are satisfied. More precisely
$\alpha_{\exstate}(\exstate^1, \access{\cstar}{\addr}, (\callstack^\dagger \neq \nil)) \cup \delta(\cstar) \cup  \{ \predhalt{\access{\gstate(\access{\cstar}{\addr})}{\stor}}{1} \}\derives \{ \predmstate{0}{0}{\arrayinit{0}}{\arrayinit{0}}{\access{\gstate^*(\access{\cstar}{\addr})}{\stor}}{1}\}$ (since $\size{\callstack^\dagger} > 0$).
By the monotonicity of $\delta(\cstar)$ (\autoref{thm:delta-monotonicity}) hence there is some $\absconfig_x$ such that $\absconfig_x \geq \{ \predmstate{0}{0}{\arrayinit{0}}{\arrayinit{0}}{\access{\gstate(\access{\cstar}{\addr})}{\stor}}{1}\}$ and
$\absconfig_H \cup \absconfig_{\callstack^*} \cup \delta(\cstar) \derives \absconfig_x$.
Since $\access{\gstate^*(\access{\cstar}{\addr})}{\stor} = \access{\gstate'(\access{\cstar}{\addr})}{\stor}$ and as $\regstatefull{\mstate'}{\exenv'}{\gstate'}{\transeffects}$ is an initial state we know that $\alpha_\exstate(\regstatefull{\mstate'}{\exenv'}{\gstate'}{\transeffects}, \access{\cstar}{\addr}, (\callstack^* \neq \nil)) =  \{ \predmstate{0}{0}{\arrayinit{0}}{\arrayinit{0}}{\access{\gstate^*(\access{\cstar}{\addr})}{\stor}}{1}\}$ which concludes the proof since
\begin{align*}
\absconfig_{\callstack^{*}} \cup \delta(\cstar)
&\derives  \absconfig_{\callstack^{*}} \cup \absconfig_H \cup \delta(\cstar) \\
&\derives \absconfig_{\callstack^{*}} \cup \absconfig_x \\
&\geq \alpha(\callstack^*) \cup \alpha_\exstate(\regstatefull{\mstate'}{\exenv'}{\gstate'}{\transeffects}, \access{\cstar}{\addr}, (\callstack^* \neq \nil)) \\
&= \alpha(\callstack')
\end{align*}
\end{enumerate}
\end{itemize}
\end{itemize}
\item[$Halt$] (returning from regular halting). Then $\exstate' = \haltstatefull{\gstate'}{\transeffects'}{\lgas'}{d'}$, $\callstack^{**} = \cons{\annotate{\exstate''}{c''}}{\callstack^\dagger}$ and $\callstack' = \cons{\annotate{\exstate'''}{c''}}{\callstack^\dagger}$. We make a case distinction on $c'' = \cstar$:
\begin{itemize}
\item[$c'' \neq \cstar$] In this case clearly  $\alpha(\callstack^*) \supseteq \alpha(\callstack^\dagger)$ and $\alpha(\callstack') = \alpha(\callstack^\dagger)$ and consequently $\absconfig_{\callstack'} \geq \alpha(\callstack')$ and hence the claim trivially follows by the reflexivity of $\derives$.
\item[$c'' = \cstar$] In this case $\alpha(\callstack^*) \supseteq \alpha_\exstate(\exstate'', \access{\cstar}{\addr}, \cl'') \cup \alpha(\callstack^\dagger)$ and $\alpha(\callstack') = \alpha_\exstate(\exstate''', \access{\cstar}{\addr}, \cl'') \cup \alpha(\callstack^\dagger)$.
From \autoref{lem:callstates-but-top}, we know that $\exstate''$ is a call state. With \autoref{lem:agreement}, we additionally have that $\access{\access{\exstate''}{\exenv}}{\code} = \access{\cstar}{\code}$ and hence also $\access{\cstar}{\code}[\access{\access{\exstate''}{\mstate}}{\pc}] =\CALL$\footnote{See \autoref{note:call}}. Consequently $\delta(\cstar) \supseteq \instabs{\CALL}_{\access{\access{\exstate''}{\mstate}}{\pc}}$. In addition we have that $\alpha_\exstate(\exstate'', \access{\cstar}{\addr}, \cl'') \leq \absconfig_{\callstack^*}$ and since $\exstate''$ is a call state, all preconditions of rule \autoref{rule:call-1} in $\instabs{\CALL}_{\access{\access{\exstate''}{\mstate}}{\pc}}$ are satisfied. More precisely $\alpha_\exstate(\exstate'', \access{\cstar}{\addr}, \cl'') \cup \delta(\cstar) \derives \pmstate{\access{\access{\exstate''}{\mstate}}{\pc} + 1}{\size{\access{\access{\exstate''}{\mstate}}{\stack}} - 6}{\update{\stacktoarray(\access{\access{\exstate''}{\mstate}}{\stack})}{\size{\access{\access{\exstate''}{\mstate}}{\stack}} - 7}{\top}}{\fun{x}{\top}}{\fun{x}{\top}}{\cl''} = p$.
\sloppy
We know additionally that $\alpha_\exstate(\exstate''', \access{\cstar}{\addr}, \cl'') =
\pmstate{\access{\access{\exstate'''}{\mstate}}{\pc}}{\allowbreak \size{\access{\access{\exstate'''}{\mstate}}{\stack}}}{\allowbreak\stacktoarray(\access{\access{\exstate'''}{\mstate}}{\stack})}{\allowbreak\towordmem(\access{\access{\exstate'''}{\mstate}}{\mem})}{\allowbreak\access{\access{\access{\exstate'''}{\mstate}}{\gstate}(\access{\cstar}{\addr})}{\stor}}{\allowbreak\cl''}$
Since $\access{\access{\exstate'''}{\mstate}}{\pc} = \access{\access{\exstate''}{\mstate}}{\pc}+1$, $\size{\access{\access{\exstate'''}{\mstate}}{\stack}} = \size{\access{\access{\exstate''}{\mstate}}{\stack}}-6$ and for all $i \in \{0, \dots, \size{\access{\access{\exstate''}{\mstate}}{\stack}}-8\}$ we have $\access{\access{\exstate''}{\mstate}}{\stack}[i] = \access{\access{\exstate''}{\mstate}}{\stack}[i]$, it holds that $p \geq \alpha_\exstate(\exstate''', \access{\cstar}{\addr}, \cl'')$ (since $\fun{x}{\top}$ is the top element for mappings $f \in \NN \to \absdom$ and $\top \geq \access{\access{\exstate'''}{\mstate}}{\stack}[0]$, cf. \autoref{lem:suprema}).
So since $\alpha(\callstack^*) \cup \delta(\cstar) \derives p$ there is by the monotonicity of $\delta(\cstar)$ (\autoref{thm:delta-monotonicity}) some $\absconfig_p$ such that
$\absconfig_{\callstack^*} \cup \delta{\cstar} \derives \absconfig_p$ and $\absconfig_p \geq p$.
Consequently we can conclude the proof:
\begin{align*}
\absconfig_{\callstack^*} \cup \delta(\cstar) &\derives \absconfig_p \cup \absconfig_{\callstack^*} \\
&\geq \{p \} \cup \alpha(\callstack^*)\\
&\geq \alpha_\exstate(\exstate''', \access{\cstar}{\addr}, \cl'') \cup \alpha(\callstack^\dagger) \\
&= \alpha(\callstack')
\end{align*}
\end{itemize}
The same arguments apply for returning from exceptional halting.
\end{itemize}
\end{itemize}
\end{proof}

\section{Checking Security Properties with \ethor{}}
In this section, we discuss how the security properties presented in~\autoref{subsec:props} are implemented in \ethor{} using \horst{}. In particular, we explain how reachability properties can be abstracted as queries using the example of the call reachability property. Afterwards , we illustrate the infrastructure for proving functional correctness queries as well as the one for automated soundness and precision testing.

\label{app:security-properties}
\subsection{From reachability properties to queries}
\label{sec:appendix-queries}
\label{subsec:queries}
All reachability properties introduced in \autoref{subsec:reachability} can be seen as instances of properties of the following form:
\par
\nobreak
{ \small
\noindent
\begin{align*}
\prop{P}{\bad} \define
\forall \exstate. \, P([\exstate])
\implies
\neg \exists \callstack',  \,
\ssteps{\transenv}{\cons{\annotate{\exstate}{\cstar}}{\callstack}}{\concatstack{\callstack'}{\callstack}}
~\land~   \bad(\callstack')
\end{align*}
}%
where $\exstate$ is assumed to be strongly consistent with $\cstar$ and $\callstack'$ is assumed to be non-empty. We will refer to properties of this form in the following as \emph{unreachability} properties.

For the sake of presentation, we will in the following interpret predicates $P, \bad$ as the sets of those elements satisfying these predicates.
Additionally, we will overload the abstraction function $\configabs{}$ to operate on sets of configurations hence writing
$\configabs{}_{\cstar}(\bad)$ for
$\bigcup_{\callstack' \in \bad}{\configabs{}_{\cstar}(\callstack')}$ and
$\configabs{}_{\cstar}(P)$ for
$\bigcup_{[\annotate{\exstate}{\cstar}] \in P}{\configabs{}_{\cstar}([\annotate{\exstate}{\cstar}])}$.

Following \autoref{theorem:soundness} for proving such properties it is sufficient to give some set $\absconfig_{P}$ such that $\absconfig_P \geq \configabs{}_{\cstar}(P)$ and to show, for any set $\absconfig_{\bad}$ over-approximating
$\configabs{}_{\cstar}(\bad)$
that $\absconfig_{P} \not \derives \absconfig_{\bad}$.
Instead of showing this property for all possible sets $\absconfig_{\bad}$, it is sufficient to find a query set $\queryset$ that shares at least one element with all possible sets $\absconfig_{\bad}$:
\par
\nobreak
{\small
\noindent
\begin{align}
\forall \absconfig_{\bad}.~ \absconfig_{\bad} \geq \configabs{}_{\cstar}(\bad)
\implies \absconfig_{\bad} \cap \queryset \neq \emptyset \label{eq:querysetprop}
\end{align}
}
Proving $\absconfig_\bad, \delta(\cstar) \not \derives \queryset$ then implies that $\prop{P}{\bad}$ holds.

Under certain conditions, such a set can be easily constructed from $\bad$ as follows:
\par
\nobreak
{\small
\noindent
\begin{align}
\queryset(\bad) \define
\{ p' ~|~ \exists p.~ p \in \configabs{}_{\cstar}(\bad) ~\land~ p \ord{p} p' \}
\label{eq:queryset}
\end{align}
}%
Intuitively, it is sufficient to query for the most concrete abstraction (as given by $\configabs{}_\cstar$) of the concrete configurations in $\bad$ and all predicate-wise ($\ord{p}$) coarser abstractions of those.
The set $\queryset(\bad)$ however is only a valid query set for $\bad$ if for some $\callstack' \in \bad$ it holds that $\configabs{}_{\cstar}(\callstack')$ is non-empty. Otherwise \autoref{eq:querysetprop} is trivially violated.
Intuitively this means that that only postconditions $\bad$ that make some restrictions on those callstack components that are modeled by the analysis (namely executions of contract $\cstar$) can be reasonably analyzed using this technique.
We formally state this property in the following lemma:
\begin{lemma}
\label{thm:query}
Let $\smalls \subseteq \configs \times \configs$ be a small step semantics and $(\absdoms, \predsig, \configabs{\predsig}, \abss)$ a sound abstraction thereof.
Furthermore let $P$, $\bad \subseteq \configs$ be predicates on configurations and $\absconfig_P$ be an abstract configuration such that $\absconfig_P \geq \configabs{}(P)$. Then if there is some $c' \in \bad$ such that  $\configabs{}(c') \neq \emptyset$ it holds that
\par
\nobreak
{ \small
\noindent
\begin{align*}
\absconfig_P, \abss \not \derives \queryset(\bad)
\implies
\prop{P}{\bad}
\end{align*}
}%
\end{lemma}

As a consequence, it is generally sufficient to query for the reachability of $\queryset(\bad)$ in order to prove an unreachability property $\prop{P}{\bad}$.

We will next show how this theoretical result can be used in practice and in particular at the level of \horst{}.

\subsubsection{Initialization}
For checking an unreachability property $\prop{P}{\bad}$, we need to show the non-derivability of a valid query set $\queryset$ from some abstract configuration $\absconfig_P \geq \configabs{}(P)$. Hence we need to axiomatize such an abstract configuration $\absconfig_P$.
This can be easily done in \horst{} by providing rules having \lstinline|true| as a single premise.
For axiomatizing that the execution starts in an initial machine state as required for the call unreachability property defined in \autoref{def:callunreachability} we can add the following rule to the analysis specification:
\begin{lstlisting}
rule initOp :=
	clause
		true => MState{0}(0, [@V(0)], [@V(0)], [@T], false);
\end{lstlisting}
As the precondition $P$ of the call unreachability property requires the top state $\exstate$ (that also serves as the zero-bar for the call level) to be initial, $\configabs{}(\exstate)$ can contain only predicate applications of the form $\predmstate{0}{0}{\fun{x}{0}}{\fun{x}{0}}{\mem}{0}$ where $\mem$ is some memory mapping. However, $\fun{x}{\top}$ (corresponding to \lstinline|[@T]|) over-approximates all memory arrays and hence $\absconfig_P = \{ \predmstate{0}{0}{\fun{x}{0}}{\fun{x}{0}}{\fun{x}{\top}}{0} \} \geq \configabs{}(P)$.

\subsubsection{Queries}
In addition to syntax for writing an analysis specification, \horst{} also provides mechanisms for the interaction with the underlying SMT-solver.
More precisely it supports syntax for specifying \emph{queries} and \emph{tests}.
Syntactically, queries consist of a list of premises (as in a clause). A query leads to the invocation of the  SMT solver to test whether conjunction of those premises is derivable from the given initialization using the specified rules. The query will result in \texttt{SAT} in case that all premises are derivable and in \texttt{UNSAT} in case that the conjunction of premises can be proven to be non-derivable.

In order to check for reachability of abstract configurations, \horst{} allows for the specification of (reachability) queries that can also be generated from selector functions. The query shown in \autoref{fig:reentrancy-query} for instance checks for reentrancy by checking if
any \CALL{} instruction is reachable at call level $1$\footnote{To be fully correct there are also corresponding queries for the other relevant call instructions $\CREATE$ and $\STATICCALL$.}.
It therefore is an implementation of the reachability property introduced in \autoref{subsec:reachability}.
This query can be obtained from the call unreachability property defined in \autoref{def:callunreachability} which is of the form $\prop{P}{\bad}$ with $\bad:= \{\cons{\annotate{\exstate}{\cstar}}{\callstack'} ~|~ \size{\callstack'} > 0 ~\land~ \cstar.\code[\exstate,\mstate.\pc]  \in \callinstructions\}$.
Intuitively, we can split this property into a set of different properties $\prop{P}{\bad_i}$ where $i$ ranges over the set of $\CALL$ instructions in $\cstar$. More precisely, let $\bad_i := \{\cons{\annotate{\exstate}{\cstar}}{\callstack'} ~|~ \size{\callstack'} > 0 ~\land~ \exstate,\mstate.\pc = i \}$ then it holds that 
\par
\nobreak
{\small
\noindent
\begin{align*}
\prop{P}{\bad} \Leftrightarrow \forall i \in \{ i ~|~ \cstar.\code[i]  \in \callinstructions \}. \, \prop{P}{\bad_i}
\end{align*}
}%
Then each instance of the \lstinline|reentrancyCall| query specifies one query set $\queryset^i$ that satisfies \autoref{eq:querysetprop} for $\bad_i$.
Hence showing the underivability of all those sets from $\absconfig_P$ proves the claim.
Intuitively, $\queryset^i$ satisfies \autoref{eq:querysetprop} for $\bad_i$ because $\configabs{}_\cstar(\bad_i)$ contains an application of a predicate $\pnmstate{i}$ with argument $\cl = 1$ and so it needs to contain all abstractions of $\absconfig_{\bad_i} \geq \configabs{}_\cstar(\bad_i)$ as the $\cl$ component ranging over $\BB$ cannot further be abstracted. Consequently, the set $\queryset^i$, which contains all predicates of that form, has a trivial intersection with $\absconfig_{\bad_i}$.

\subsection{Functional correctness}
\label{sec:functional-correctness}
For checking functional correctness, some modifications to the abstract semantics are necessary.

This is as the different contract executions need to be bound the the corresponding input data of the call and since we want to reason about return data.
We will in the following shortly overview the relevant changes and motivate that similar modifications can easily be incorporated for reasoning about other dependencies with the execution or blockchain environment. We will present the relevant modifications in \horst{} syntax so that the explanations serve as a guide to the enhanced version of the semantics ~\cite{extended}.

First, the relevant predicates need to be enriched with a corresponding representation of the call data.
We decided to represent call data as a word array with the particularity that the array's first element represents only 4 bytes. This is due to the call conventions enforced by the Solidity compiler which interpret the first 4 bytes of input data as the hash of the called function's signature to properly dispatch function calls.
In addition to the call data, we introduce a new predicate representing the return data of a call.

Formally, we arrive at the following predicate definitions:
\begin{lstlisting}
datatype CallData := @D<int*array<AbsDom>>;
pred MState{int*int}: int * array<AbsDom> * array<AbsDom> * array<AbsDom> * bool * CallData;
pred Exc{int}: bool;
pred Halt{int}: array<AbsDom> * AbsDom * bool * CallData;
pred ReturnData{int}: int * AbsDom * bool * CallData;
\end{lstlisting}
Note that we represent call data as a tuple of its size and an array of abstract words. Also, we added to the \lstinline|Halt| predicate an argument representing the return data size. This argument stems from the abstract domain with \lstinline|@T| indicating that the concrete size of the return data is unknown.
The \lstinline|ReturnData| predicate maps the positions of the return data (word) array to the corresponding values that it holds.

The existing rules simply propagate the the call data array with the only addition that the $\CALLDATALOAD$ instruction now accesses the call data array instead of over-approximating the loaded value.
The new rule for $\CALLDATALOAD$ is depicted in \autoref{fig:rule-calldataload}.
\begin{figure*}
\begin{lstlisting}
rule opCallDataLoad :=
	for (!id: int) in ids(), (!pc: int) in pcsForIdAndOpcode(!id, CALLDATALOAD), (!a: int) in argumentsOneForIdAndPc(!id, !pc)
	clause [?x: AbsDom, ?size: int, ?sa: array<AbsDom>, ?mem: array<AbsDom>, ?stor: array<AbsDom>, ?cl: bool, ?p: int, ?v: AbsDom, ?cdata: CallData]
		MState{!id, !pc}(?size, ?sa , ?mem, ?stor, ?cl, ?cdata), ?size > 0,
		!a != ~1, // in case that the position could be pre-computed, use it for accessing the position more precisely
		?v = accessWordCalldata{!a}(?cdata) // accesses word at the corresponding position of the call data
		=> MState{!id, !pc +1}(?size, store ?sa (?size -1) (?v), ?mem, ?stor, ?cl, ?cdata),
	clause [?x: AbsDom, ?size: int, ?sa: array<AbsDom>, ?mem: array<AbsDom>, ?stor: array<AbsDom>, ?cl: bool, ?cdata: CallData, ?p: int, ?v: AbsDom]
		MState{!id, !pc}(?size, ?sa , ?mem, ?stor, ?cl, ?cdata), ?size > 0,
		!a = ~1, // if the argument could not be preecomputed, extract the argument from stack
		?x = select ?sa (?size - 1),
		?v = (isConcrete(?x)) ? (accessWordCalldataEven(extractConcrete(?x), ?cdata)) : (@T) // if the offset is concrete, try to access the word at the given position. This will only result in a concrete result if the value is a word position
		=> MState{!id, !pc +1}(?size, store ?sa (?size -1) (?v), ?mem, ?stor, ?cl, ?cdata);
\end{lstlisting}
\caption{Rule for $\CALLDATALOAD$ in the enhanced abstract semantics.}
\label{fig:rule-calldataload}
\end{figure*}

The $\CALLDATALOAD$ operation takes as argument a value from the stack that specifies the byte position starting from which one word of the call data byte array shall be loaded (to the stack).
The rule is split into two clauses for taking advantage of the pre-analysis.
More precisely, in case that the position of call data to is known upfront, the call data array \lstinline|?call| can be assessed more precisely. Since we model the call data as a word instead of a byte array (similar to our memory abstraction), either a word loaded from it consists of a full word in the word array or needs to be composed out of two neighboring words. Composing to integers (interpreting them as byte arrays) however requires exponentiation  as defined in the append function in~\autoref{sec:appendix-ana-def}. \zz{} is not able to handle general exponentiation - for this reason we can only compute such exponentiations (by unfolding to multiplications) whose exponent is known upfront.
Consequently, the first rule in \autoref{fig:rule-calldataload} handles the case where the argument to the call is known upfront: the \lstinline|accessWordCallData| function expects the position as a parameter and computes the accessed word precisely from the call data array since exponentiation can be unrolled. The second rule handles the case where the argument to the call is not known upfront. In case that during the analysis it can be detected to be concrete (by the function \lstinline|isConcrete|), the \lstinline|accessWordCalldataEven| function is used to access the call data at the corresponding position. This function however only yields a precise result in case that the provided position corresponds to the beginning of a word in the calldata array, otherwise it over-approximates the result as \lstinline|T|.

The \lstinline|ReturnData| predicate is inhabited by the rules that model regular halting. We exemplarily show the rule of the $\RETURN$ opcode depicted in \autoref{fig:rule-return}.

{
\begin{figure*}[]
{
\begin{lstlisting}
rule opHaltOnReturn :=
    for (!id: int) in ids(), (!pc: int) in pcsForIdAndOpcode(!id, RETURN) 
    let 	
      macro #StackSizeCheck := MState{!id,!pc}(?size, ?sa, ?mem, ?stor, ?cl, ?cdata), ?size > 1
    in
    clause [?sa: array<AbsDom>, ?mem: array<AbsDom>, ?stor: array<AbsDom>, ?size:int, ?cl: bool, ?cdata: CallData, ?length: AbsDom]
        #StackSizeCheck,
        ?length = select ?sa (?size-2)
        => Halt{!id}(?stor, ?length, ?cl, ?cdata),
    clause [?sa: array<AbsDom>, ?mem: array<AbsDom>, ?stor: array<AbsDom>, ?size:int, ?cl: bool, ?offset: AbsDom, ?length: AbsDom, ?o: int, ?l:int, ?p:int, ?v: AbsDom, ?cdata: CallData]
        #StackSizeCheck,
        ?offset = select ?sa (?size-1), // select top values on the stack
        ?length = select ?sa (?size-2),
        isConcrete(?offset),
        isConcrete(?length), 
        ?o = extractConcrete(?offset),
        ?l = extractConcrete(?length),
        ?p >= 0, 
        (?p * 32) < ?l, // write all words that are still within the length 
        ?v = accessWordMemoryEven(?o + ?p, ?mem)
        => ReturnData{!id}(?p, ?v, ?cl, ?cdata), // careful: the Return data predicate is also inhabited in words!
    clause [?sa: array<AbsDom>, ?mem: array<AbsDom>, ?stor: array<AbsDom>, ?size:int, ?cl: bool, ?offset: AbsDom, ?length: AbsDom, ?o: int, ?l:int, ?p:int, ?v: AbsDom, ?cdata: CallData]
        #StackSizeCheck, 
        ?offset = select ?sa (?size-1), // select top values on the stack
        ?length = select ?sa (?size-2),
        ~isConcrete(?offset), // if we don't know the offset, but only the length, we write top at the places in the specified range
        isConcrete(?length), 
        ?l = extractConcrete(?length),
        ?p >= 0, 
        ?p * 32 < ?l
        => ReturnData{!id}(?p, @T, ?cl, ?cdata),
    clause [?sa: array<AbsDom>, ?mem: array<AbsDom>, ?stor: array<AbsDom>, ?size:int, ?cl: bool, ?offset: AbsDom, ?length: AbsDom, ?o: int, ?l:int, ?p:int, ?v: AbsDom, ?cdata: CallData]
        #StackSizeCheck,
        ?length = select ?sa (?size-2),
        ~isConcrete(?length), 
        ?p >= 0 
        => ReturnData{!id}(?p, @T, ?cl, ?cdata);
\end{lstlisting}
\caption{Rule for $\RETURN$ in the enhanced abstract semantics.}
\label{fig:rule-return}
}
\end{figure*}
}

The $\RETURN$ instruction in EVM reads a memory offset and length from the stack and returns the corresponding memory fragment as byte array. In our abstraction the return data is modeled by an own predicate that holds words instead of bytes. This design choice follows the one made for the word-indexed memory and the call data array which hold words instead of bytes as well for performance reasons.
The $\RETURN$ semantics is closely reflected in the abstract $\RETURN$ rule: the first clause of the rule inhabits the \lstinline|Halt| predicate, reading the size of the return data from the stack. The next three clauses inhabit the \lstinline|ReturnData| predicate, differentiating depending on how much information on the return data (size and memory offset) are known: If both memory offset and length of the data are known, for each word position \lstinline|?p| the corresponding memory word is read from the memory array \lstinline|?mem| (using the function \lstinline|accesswordMemoryEven|) and written into the \lstinline|ReturnData| predicate.
The next clause describes the case where the memory offset is unknown, but the size of the return data is known. In this case we cannot know which (concrete values) form the return data, but can only approximate all possible return data words (as determined by the size of the array) with \lstinline|@T|.
The last clause covers the case where the length of the return data is not known. Since it is fully unclear in this case whether data should be returned in the first place (since the length could be $0$), all potential positions of the return data array are over-approximated by \lstinline|@T|.

Finally, the functional correctness queries for the addition function of the SafeMath library can be posed as follows:
{
\begin{lstlisting}
op callAdd(x: int, y: int): CallData := 
	@D(68, store (store (store [@T] 0 @V(1997931255)) 1 (@V(x))) 2 (@V(y)));


test addOverflowNoHalt expect UNSAT
	for (!id: int) in ids() 
      [?x:int, ?y: int, ?z:int, ?p:int, ?stor: array<AbsDom>, ?rdsize:AbsDom]
         ?x >= 0,
         ?y >= 0,
         ?x < MAX, 
         ?y < MAX,
         ?x + ?y >= MAX,
         Halt{!id}(?stor, ?rdsize, false, callAdd(?x, ?y));

test addNoOverflowCorrect expect SAT
    for (!id: int) in ids() 
      [?res: AbsDom, ?x:int, ?y: int, ?z:int, ?rdsize:AbsDom, ?stor: array<AbsDom>]
         ?x >= 0,
         ?y >= 0,
         ?x + ?y < MAX,
         ReturnData{!id}(0, ?res, false, callAdd(?x, ?y)),
         Halt{!id}(?stor, ?rdsize, false, callAdd(?x, ?y)), 
         abseq(?rdsize, @V(32)),
         abseq(?res, @V(?x + ?y));

test addNoOverflowHalt expect UNSAT
    for (!id: int) in ids() 
      [?res: AbsDom, ?x:int, ?y: int, ?z:int, ?rdsize: AbsDom, ?stor: array<AbsDom>]
         ?x >= 0,
         ?y >= 0,
         ?x + ?y < MAX,
         Halt{!id}(?stor, ?rdsize, false, callAdd(?x, ?y)),
         ?rdsize != @V(32);

test addNoOverflowUnique expect UNSAT
    for (!id: int) in ids() 
      [?res: AbsDom, ?x:int, ?y: int, ?z:int, ?rdsize: AbsDom, ?stor: array<AbsDom>]
         ?x >= 0,
         ?y >= 0,
         ?x + ?y < MAX,
         ReturnData{!id}(0, ?res, false, callAdd(?x, ?y)), 
         ?res != @V(?x + ?y);
\end{lstlisting}
}

We first specify the call data for a call to the \texttt{add} function of the SafeMath library as an operation \lstinline|callAdd| returning an \lstinline|CallData| element when being provided with the arguments to the call. Since the \texttt{add} function expects 2 integer arguments the \lstinline|callAdd| function return a call data of size $68$ ($4 + 2*32)$ bytes) where the a \lstinline|T| array is initialized with the hash of the corresponding function signature as first element (which represents the first 4 bytes of the call data) and the arguments \lstinline|x| and \lstinline|y| as following to elements.
Note that the hash of the function signature and its hash is provided by Solidity compilers via the so called Ethereum Contract ABI (Contract Application Binary Interface). In the future we plan to automatically generate an infrastructure for functional correctness queries on Solidity contracts from the contract's ABI.
The first functional correctness test \lstinline|addOverflowNoHalt| requires that it is impossible to reach a \lstinline|Halt| state (which indicates regular halting) from a call to to the \texttt{add} function in case that the summands \lstinline|?x| and \lstinline|?y| provided as arguments produce an overflow.

The second functional correctness test (\lstinline|addNoOverflowCorrect|) checks whether it is possible (in case that no overflow occurs) to compute the expected result (or an over-approximation thereof) in the first place. Here \lstinline|abseq| is the function implementing an equality test on the abstract domain, hence considering every concrete element to be potentially equal to \lstinline|@T|. 
By the soundness of the analysis, if this query would turn out to be unsatisfiable, it would be impossible for the function to produce the correct result under any circumstances. This query of course does not prove that the function will always provide a result: This indeed is and should not be provable, since any smart contract can always halt exceptionally when running out of gas. This test case serves as a sanity check that only becomes meaningful in conjunction with the following tests.
The third and fourth functional correctness tests (\lstinline|addNoOverflowHalt| and \lstinline|addNoOverflowUnique|) prove that given non-overflowing arguments, if the function execution halts successfully, nothing but the correct result can be produced. In other words, it is impossible to halt successfully without producing the correct result. This property is composed out of two queries since it needs to be shown that 1) It is impossible for the function to halt without returning a result of length $32$ (corresponding to one word)  as recorded in the \lstinline|Halt| predicate and 2) It is impossible that the actual return value (as recorded in the \lstinline|ReturnData| predicate) differs from the sum of the two arguments.

The functional correctness tests for the other functions of the SafeMath library follow the same pattern.

\subsection{Automated Testing in \horst{}}
\label{subsec:automated-testing-appendix}
The setup for automated testing (see \autoref{subsec:testing}) shown in \autoref{fig:automated-testing} presents a use case for the Hoare-Logic-style reasoning capabilities of \ethor{} and furthermore provides us with the opportunity to showcase some further features of \horst{}.

We first shortly overview the form of the test cases in the official EVM test suite:
Test cases come in two flavors: the first group consists of 490 test cases specifying a storage configuration as postcondition, the second group, consisting of 108 test cases, lacks a post condition (which we interpreted as requiring an exceptional halt).

To account for this test structure we declare four additional selector functions:
The selector functions \lstinline|preStorageForId| and \lstinline|postStorageForId| provide tuples of storage offsets and values which specify the storage contents before and after the execution of the contract.
\lstinline|emptyListIfNoPostConditionForId| and \lstinline|dummyListIfNoPostConditionForId| generate an empty list, respectively a list with one element, depending on if there is a postcondition specified or not.
Since rules are generated for the cross product of their selector functions return values, we can use these functions to generate different rules for different test cases while still using the same \horst{} inputs.

The rule for initialization, \lstinline|initOp|, differs in one value from the definition used in the other experiments. In line 11, we populate an array with the values returned by \lstinline|preStorageForId|, starting from an array containing only zeroes.

In case we want to check for the reachability of a certain storage configuration, we generate the two queries \lstinline|correctValues| and \lstinline|uniqueValues|.
\lstinline|correctValues| is successful, if a \lstinline|Halt| predicate is reachable whose storage contains 1) values abstractly equal to the values returned by \lstinline|postStorageForId| at the offsets returned by \lstinline|postStorage| and 2) a value abstractly equal to $0$ for all offsets not returned by \lstinline|postStorageForId|.
\lstinline|uniqueValues| is successful, if no \lstinline|Halt| predicate is reachable whose storage contains any value abstractly unequal to the values returned by \lstinline|postStorageForId|. This is only the case, if every value in the memory is concrete.
Summing up, such a test case is considered to be solved \emph{correctly} if \lstinline|correctValues| is successful and considered to be solved \emph{precisely} if \lstinline|correctValues| and \lstinline|uniqueValues| are successful.

In case we want to check for exceptional halting, we just query for the unreachability of a regular \lstinline|Halt| predicate (see \lstinline|irregularHalt|).
Such a query is considered solved precisely on success and imprecisely on failure, since reaching additional program states (\lstinline|Halt| in this instance), which are not reachable in the concrete execution, is a sign of over-approximation.

\begin{figure*}[!t]
\lstset{escapechar=,language=HoRSt, numbers=left,xleftmargin=2em,frame=lines,framexleftmargin=1.5em}
\begin{lstlisting}
sel preStorageForId: int -> [int*int];
sel postStorageForId: int -> [int*int];
sel emptyListIfNoPostConditionForId: int -> [bool];
sel dummyListIfNoPostConditionForId: int -> [bool];

rule initOp :=
  for (!id:int) in ids()
  clause
    true
    => MState{!id, 0}(0, [@V(0)], [@V(0)],
                      for (!offset: int, !value:int) in preStorageForId (!id): x: array<AbsDom> -> store x !offset @V(!value), [@V(0)],
                      false);

test correctValues expect SAT
    for (!id: int) in ids(),
        (!b: bool) in emptyListIfNoPostConditionForId(!id)
    [?stor: array<AbsDom>, ?i: int]
    for (!offset: int, !value:int) in postStorageForId(!id): && abseq(select ?stor !offset,@V(!value)),
    (for (!offset: int, !value:int) in postStorageForId(!id): || ?i = !offset) ? (true) : (abseq(select ?stor ?i,@V(0))),
    Halt{!id}(?stor, false);

test uniqueValues expect UNSAT
    for (!id: int) in ids(),
        (!b: bool) in emptyListIfNoPostConditionForId(!id)
    [?stor: array<AbsDom>]
    for (!offset: int, !value:int) in postStorageForId(!id): || absneq(select ?stor !offset,@V(!value)),
    Halt{!id}(?stor, false);

test irregularHalt expect UNSAT
    for (!id: int) in ids(),
        (!b: bool) in dummyListIfNoPostConditionForId(!id)
    [?stor: array<AbsDom>]
    Halt{!id}(?stor, false);
\end{lstlisting}
\caption{Setup for automated testing.}
\label{fig:automated-testing}
\end{figure*}

\section{Soundness Issues in Related Work}
This section reviews the soundness problems of other works on static smart contract analysis.
We thereby focus on those works that make soundness claims.
We first overview soundness problems in the reconstruction of smart contracts' control flow graphs (which particularly affects the Securify analyzer~\cite{securify}) and afterwards successively discuss the issues in the analyses performed by~\cite{securify},~\cite{lu2019neucheck},~\cite{GMS::CAV18}, and~\cite{kalra2018zeus}. Where possible, we provide reproducible evidence in form of concrete counter-examples for the spotted sources of unsoundness.

\label{app:soundness-issues}
\subsection{Control Flow Reconstruction}
\label{sec:appendix-cfg}
Most tools that analyze Ethereum smart contracts at the level of bytecode base their analysis on the contract's control flow graph (CFG). However, the design of the EVM bytecode language does not allow for an easy reconstruction of a contract's control flow since jump destinations are not statically fixed, but might be dynamically computed.
More precisely, in EVM bytecode jump destinations are read from the stack and hence can be subject to prior computations. Even though the set of potential jump destinations is statically determined (since only such program counters with a $\JUMPDEST$ instruction constitute valid jump destinations), the concrete destination of a jump instruction might only be dispatched at runtime. The challenge hence lies in statically narrowing down the set of possible jump destinations for each branch instruction ($\JUMP$ or $\JUMPI$).
To this end, the state-of-the-art analyzer~\cite{securify} employs a custom algorithm, another popular solution~\cite{trailofbits-manticore} uses an external open-source tool~\cite{cfgbuilder} for control flow graph reconstruction. While correctness for both of them has never been discussed, flaws in the CFG reconstruction can lead to catastrophic consequences: An unsound reconstruction that erroneously excludes possible jump destinations, can deem parts of the contract code unreachable that carries critical and potentially unsafe functionality (e.g., reentrant calls).

When reviewing the algorithms used in~\cite{securify} and~\cite{cfgbuilder}, we found soundness issues in both approaches as we will discuss in the following. In ~\autoref{fig:cfg-example} we show a compact example of a smart contract's control flow that is recovered incorrectly by~\cite{securify, cfgbuilder} with no errors reported.
\begin{figure}[h!]
\centering
\includegraphics[width=\linewidth]{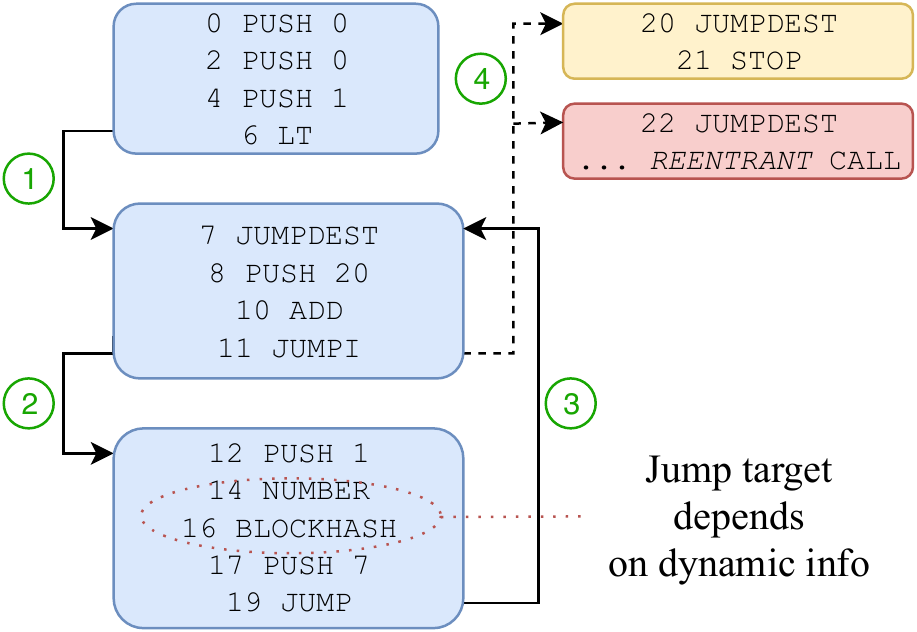}
\caption{Problematic Control Flow Example.}
\label{fig:cfg-example}
\end{figure}
Intuitively, the control flow of this contract should not be fully recoverable because one of its jump destinations depends on some blockchain information (the block hash and the block number) which cannot be statically predicted, but will only be fixed once the contract has been published on the blockchain. 

The smart contract is structured into five basic blocks.
The first block (starting at program counter $0$), initializes the local machine check with two $0$ values and continues with the execution of the second block starting at program counter $7$ ({\color{darkgreen}\circled{1}}). The second block can be entered via a jump (since it starts with a $\JUMPDEST$ instruction). It intuitively takes two stack values as arguments, the first one functioning as jump offset and the second being the jump condition: it computes the next jump destination as the sum of $20$ and the top stack element and conditionally jumps to this destination based on the second stack value. In the first iteration since both of these values are $0$ (and so particularly the condition is $0$), no jump is performed, but instead the execution proceeds with block three (starting at program counter $12$) with the empty stack ({\color{darkgreen}\circled{2}}).
This block pushes the current block number and hash to the stack and jumps back to the second block ({\color{darkgreen}\circled{3}}). Since at this point the input to the second block are values that are not statically determinable, it needs to be assumed that the jump condition as well as the jump offset could have any value.
It is hence possible during the real execution to jump to arbitrary jump destinations from program counter $10$ ({\color{darkgreen}\circled{4}}). This includes the block starting  at program counter $20$ where the execution of the contract is stopped and most importantly the block starting at program counter $22$ that executes a reentrant call. Thus, if this jump destination is undiscovered, false correctness results for reentrancy can be produced in a subsequent analysis.

There are two sound approaches for handling the usage of unpredictable information in jump destination reconstruction: Conservatively, a smart contract can be rejected by the analysis and hence be considered potentially vulnerable in this case (which is our approach) or the analysis could assume that all $\JUMPDEST$ instructions of the contract are potentially reachable.
The tools that we reviewed, however, did not follow any of these options, but produced the following results: ~\cite{cfgbuilder} correctly discovers the basic blocks, but cannot recover jumps to the targets $20$ and $22$ ({\color{darkgreen}\circled{4}}). The result of ~\cite{securify} is even more surprising: the algorithm does not manage to recover any of the blocks shown in ~\autoref{fig:cfg-example}, but reports as CFG of this contract a single block consisting of a modulo instruction followed by the $\STOP$ opcode. Consequently, all analyses that use either of these CFG reconstruction solutions will consider the reentrant call of the example contract to be unreachable and will based on that label the contract as safe to reentrancy attacks.

\subsection{Securify}
\label{sec:securify}
The Securify tool~\cite{securify} encodes dependencies inferred from a contract's control flow graph as logical facts and specifies security properties in terms of compliance and violation patterns using these facts. It is claimed that the satisfaction of a compliance pattern is sufficient for proving a security property, while matching a violation pattern guarantees that a security property is indeed violated.
We will in the following review most of the provided patterns and give counterexamples, showing that most of these patterns indeed are not sound. We validated as far as possible the patterns reported in the paper with the provided online tool (\url{https://securify.chainsecurity.com})\footnote{We accessed the website January 19th and validated all properties with Solidity Compiler version 0.4.25.}
. Unfortunately, some of the patterns introduced in~\cite{securify} were changed or renamed in the online tool. We will note this when discussing the corresponding pattern. Also it should be noted that the online tool only reports security problems. More precisely, an alarm (red) is produced, if a violation pattern is matched, a warning (orange) is produced if neither a violation nor a compliance pattern is matched. The lack of a report for a certain security property indicates that the property's compliance pattern was matched.

\subsubsection{Ether Liquidity}
The LQ (Ether liquidity) property ensures that a property cannot lock Ether (for this reason it is called Locked Ether in the online tool) meaning that for all the contract's executions either leave the contract's balance unaffected or there is a trace that allows to reduce the contract's balance.

The property formulates three different compliance patterns.
The first two compliance patterns ensure that all halting instructions are preceded by a successful conditional check on the value given to the call being $0$. This ensures that only such executions can complete successfully that have been guaranteed to have gotten no money transferred. These patterns are probably sufficient to guarantee that a contract can never receive money and hence for showing the LQ property.
The third compliance pattern checks whether there is a call that is reachable while at the same time (meaning that the call and an exception opcode are not reachable from the same conditional branch) no exception is reachable and this call transfers non-zero value or a value that is settable by the environment. This shall ensure that the contract has at least one way of successfully transferring money.

This compliance pattern is not sufficient for ensuring Ether Liquidity. Despite the problem that the corresponding Ether transferring call could be restricted to a certain address which can never initiate such a call (as it belongs to a contract without a functionality to call other contracts), the pattern also does not consider that an exception that reverts the transaction might not only occur conditionally.

Consider the following contract:
{
\lstset{escapechar=,language=JavaScript, numbers=left,xleftmargin=2em,frame=lines,framexleftmargin=1.5em}
\begin{lstlisting}
contract Bob {
  function sendMoney(address c) {
    c.send(2);
    throw;
  }
  function receive() payable {
  }
}
\end{lstlisting}
}
This contract is labeled not to lock Ether even though it can receive money (via the \lstinline|receive| function) and every Ether transfer to another contract (via \lstinline|sendMoney|) will be reverted.
Note that the absence of Solidity's \lstinline|payable| will be translated to a conditional check on the call value, and cause a revert once if the value given to the call is non-zero.

The violation pattern for LQ requires that there is no \lstinline|CALL| instruction that transfers a non-zero amount of Ether and that there is some halting instruction such that if its reachability is dependent on a conditional branching, this condition can be determined by the transaction data, hence can be enabled by the transaction initiator. This shall ensure that there is at least one execution trace that does not halt exceptionally and hence reverts the execution effect.

However, the following contract is reported to lock Ether (matches the violation pattern):
{
\lstset{escapechar=,language=JavaScript, numbers=left,xleftmargin=2em,frame=lines,framexleftmargin=1.5em}
\begin{lstlisting}
contract Bob {
  function receive(uint x) payable {
    if (x > 0 || x <= 0) {
      throw;
    }
  }
}
\end{lstlisting}
}
This contract clearly cannot lock Ether since it cannot receive any Ether. Its only function \lstinline|receive| throws an exception depending on a conditional which is always true. Still, the dependency analysis labels this condition to be determined by transaction data, so that the pattern is matched. This can be also considered as a soundness flaw in the Securify's definition of determinability since in this case clearly for different values of transaction data the value of the condition is the same (which contradicts ~\cite{securify}'s definition of determinability.)

\subsubsection{No writes after calls}
The NW (No writes after calls) property says that a contract's storage when terminating the execution should always be the same as at the point of a previous contract call (so the contract shall not be altered between a $\CALL$ instruction and the contract's successful termination).
The online tool does not implement a property with such a name, but instead implements similar patterns for a property called Gas-dependent Reentrancy. This property uses the same intuition, but puts an additional requirement that the amount of gas given to the call shall be dependent of the remaining gas.
One should note that it is very misleading that this property is in the online tool called Gas-dependent Reentrancy, even though~\cite{securify} explicitly claims that the NW property is different from reentrancy. We will detail out later why the NW property indeed is not a sound or complete approximation of the single-entrancy property.

The corresponding compliance pattern requires that the $\CALL$ instruction may not be followed by an $\SSTORE$ instruction.
This pattern however does not consider that there are other ways of modifying the storage than the $\SSTORE$ instruction, e.g., by using $\DELEGATECALL$ for calling a library function that alters the storage.

Consider the following example:
{
\lstset{escapechar=,language=JavaScript, numbers=left,xleftmargin=2em,frame=lines,framexleftmargin=1.5em}
\begin{lstlisting}
library Lib {
  struct Data { bool bvalue; }
  function write(Data storage self, bool value) {
    self.bvalue = value;
  }
}

contract Bob {
  Lib.Data sent;

  constructor () {
    sent.bvalue = false;
   }

  function ping(address c) {
    if (!(sent.bvalue)) {
      if (!c.call.value(2)()) {
        throw;
      }
      Lib.write(sent, true);
    }
  }
}
\end{lstlisting}
}
This example clearly matches the compliance pattern (since the call to the library will be translated to a $\DELEGATECALL$ instruction, hence no $\SSTORE$ instruction appears in the first place). This can also be verified with the online tool which does not report a violation of the Gas Dependent Reentrancy property.

The violation pattern for the NW property requires that there is a $\CALL$ that must be proceeded by an $\SSTORE$ instruction.
This pattern is also not sufficient as illustrated by the following example:
{
\lstset{escapechar=,language=JavaScript, numbers=left,xleftmargin=2em,frame=lines,framexleftmargin=1.5em}
\begin{lstlisting}
contract Alice {
  bool sent = false;

  function ping(address c) {
    if (!sent) {
      sent = true;
      c.call.value(2)();
      sent = sent;
    }
  }
}
\end{lstlisting}
}

This contract clearly does not violate the property (since the contract storage at the point of terminating is not altered as compared to the point of calling). Still it matches the violation pattern (and is reported by the online tool) indicating a guaranteed property violation.

Next, we shortly discuss why the NW property (independently of the fact that the patterns are not sufficient) is neither sound nor complete for single-entrancy.

We will first give an example of a contract satisfying the NW property while still being reentrant.
{
\lstset{escapechar=,language=JavaScript, numbers=left,xleftmargin=2em,frame=lines,framexleftmargin=1.5em}
\begin{lstlisting}
contract Bank {
  address a; address b;
  uint balA; uint balB;

  function setBalA (uint v) {
    balA = v;
  }

  function drainA(address ben) {
    if (msg.sender != a) { throw; }
    if (balA > 0) {
      uint v = balA;
      setBalA(0);
      ben.call.value(v)();
    }
  }
}
\end{lstlisting}
}
This contract implementing a simple bank functionality for two parties (identified by their addresses \lstinline|a| and \lstinline|b|) is vulnerable to a reentrancy attack even though no writes after the call are performed. Similar to the initial example in \autoref{fig:reentrancy}, given that \lstinline|a| is the address of a malicious contract, this contract can use the public \lstinline|setBalA| function in a reentering execution to disable the guard (here \lstinline|balA|) before reentering the contract's \lstinline|drainA| function to retransfer money that \lstinline|a| does not own.

For an example of a contract that does not satisfy the NW property, but that is still safe, we give a contract with a simple locking functionality (similar to the example in \autoref{fig:reentrancy}).
{
\lstset{escapechar=,language=JavaScript, numbers=left,xleftmargin=2em,frame=lines,framexleftmargin=1.5em}
\begin{lstlisting}
contract Bank {
  uint lock;
  mapping (address => uint) bal;

  function drain(address a) {
    if (lock == 1) { throw; }
    lock = 1;
    a.call.value(bal[msg.sender])();
    bal[msg.sender] = 0;
    lock = 0;
  }
}
\end{lstlisting}
}
The locking ensures that whenever the function is reentered an exception occurs and hence no further call can be performed. Still since the \lstinline|lock| needs to be released at the end of the execution, clearly the NW property is violated.

\subsubsection{Restricted write}
The RW (restricted write) property requires that all write accesses are restricted, meaning there is at least one address that when initiating the call cannot reach the corresponding write access.

This property needs to be questioned in its semantic definition in that this definition explicitly requires that $\SSTORE$ instructions are not reachable even though (as discussed before), the $\SSTORE$ instruction is not the only way of manipulating storage.

So for example when analyzing the following contract, there is no RW violation or warning produced for the \lstinline|Bob| contract even though the \lstinline|ping| functions allows to set contract's \lstinline|data| filed containing the owner to be set to an arbitrary value by anyone.

{
\lstset{escapechar=,language=JavaScript, numbers=left,xleftmargin=2em,frame=lines,framexleftmargin=1.5em}
\begin{lstlisting}
library Lib {
  struct Data { address owner; }
  function write(Data storage self, address value) {
    self.owner = value;
  }
}

contract Bob {
  Lib.Data data;

  function ping(address c) {
    Lib.write(data, c);
  }
}
\end{lstlisting}
}

The given compliance pattern requires that the storage offset specified in a $\SSTORE$ instruction needs to be determined by the caller of the contract. This pattern might indeed be sufficient for the semantic property only considering $\SSTORE$ instructions.

The violation pattern requires that the reachability of $\SSTORE$ instructions as well as the offset given to them may not depend on the caller of the contract.

However, in the following contract an unrestricted write is detected:
{
\lstset{escapechar=,language=JavaScript, numbers=left,xleftmargin=2em,frame=lines,framexleftmargin=1.5em}
\begin{lstlisting}
contract Test {
  bool test = false;

  function flipper () {
    if (msg.sender != 0)
      flip();
  }
  function flip () internal {
    test = !test;
  }
}
\end{lstlisting}
}
This contract should be safe with respect to the semantic definition since \lstinline|flip|, the only function containing a write access is an internal function, meaning that it can only be invoked within the contract. Given that the only place where it is invoked (in the \lstinline|flipper| function), it is done with a restriction on the caller (\lstinline|msg.sender|), also this storage access is restricted. However, the contract is reported to match the violation pattern. A reason for that could be an unsoundness in the underlying dependency analysis.

\subsubsection{Restricted transfer}
The RT (restricted transfer property) excludes that Ether transfers (via $\CALL$) cannot be invoked by any user. Again one could criticize that the property does not consider other ways of transferring money (e.g., by $\CALLCODE$).
The following contract, for example, is considered safe by this definition:
{
\lstset{escapechar=,language=JavaScript, numbers=left,xleftmargin=2em,frame=lines,framexleftmargin=1.5em}
\begin{lstlisting}
contract Bob {
  function sendMoney(address c) {
    c.callcode.value(5)();
  }
}
\end{lstlisting}
}

The corresponding compliance pattern requires that all calls transfer $0$ Ether. Given that the property only considers $\CALL$ instructions, this pattern is probably sufficient.

There are two violation patterns for the RT property, the first one requires that there is a $\CALL$ instruction transferring a non-zero amount and whose reachability may be dependent on the caller.
We can give a counterexample similar to the one for the RW violation pattern:
{
\lstset{escapechar=,language=JavaScript, numbers=left,xleftmargin=2em,frame=lines,framexleftmargin=1.5em}
\begin{lstlisting}
contract Test {
  function sendMoney() {
    if (msg.sender != 1)
      sendM();
  }
  function sendM () internal {
     msg.sender.send(1);
  }
}
\end{lstlisting}
}
Again, the tool does not detect that effectively the money transfer is restricted since the internal function \lstinline|sendM| can only be invoked in a restricted fashion.

The second violation pattern for the RT property requires instead of the transferred value to be non-zero, that the value is determined by the input to the call while at the same point the input might not affect the reachability of the $\CALL$ instruction.

We can again give a simple counterexample similar to the previous one:
{
\lstset{escapechar=,language=JavaScript, numbers=left,xleftmargin=2em,frame=lines,framexleftmargin=1.5em}
\begin{lstlisting}
contract Test {
  function sendMoney(uint x) {
    if (msg.sender != 1)
      sendM(x);
  }
  function sendM (uint y) internal {
     msg.sender.send(y);
  }
}
\end{lstlisting}
}
This example is detected as insecure while having only restricted money transfers.

\subsubsection{Handled exception}
The HE property (Handled exception) is not semantically defined, but intuitively shall ensure that exceptions that occurred in function calls shall be handled.
Due to the lack of a formal definition, it is hard to argue to which extend the given patterns really are sufficient, but we give here examples of proper/problematic exception handling which are wrongly classified.

The compliance pattern requires that every call must be followed by some branching instruction whose condition is determined by the call's return value.
Clearly, the following contract is matched by this pattern even though it does not perform a proper exception handling.
{
\lstset{escapechar=,language=JavaScript, numbers=left,xleftmargin=2em,frame=lines,framexleftmargin=1.5em}
\begin{lstlisting}
contract SimpleBank {
  mapping(address => uint) balances;
  uint successes;

  function withdraw() {
    bool success = msg.sender.send(balances[msg.sender]);
    if (success) { successes++; }
    balances[msg.sender] = 0;
  }
}
\end{lstlisting}
}
Even though this contract branches on the return value of the call, this branching does not influence the critical instruction, namely the following storage update that assumes a successful call.

The violation pattern for HE requires that all branching instructions following a $\CALL$ instruction do not have a condition that depends on the outcome of the call.
We give an example of a contract matching this pattern that however implements a useful form of exception handling:
{
\lstset{escapechar=,language=JavaScript, numbers=left,xleftmargin=2em,frame=lines,framexleftmargin=1.5em}
\begin{lstlisting}
library Lib {
  function toInt(bool b) returns (uint n) {
    if (b)
      return 1;
    else
      return 0;
  }
}

contract SimpleBank {
  mapping(address => uint) balances;

  function withdraw() {
    bool success = msg.sender.send(balances[msg.sender]);
    balances[msg.sender] = Lib.toInt(success) * balances[msg.sender];
  }
}
\end{lstlisting}
}
This contract uses the return value of the call to update the callee's balance after the call depending on that. Since the branching on the return value is outsourced to the library function \lstinline|toInt|, it can not be captured by the corresponding pattern.

In general, it is hard to imagine how proper exception handling should be generically defined, since this is a property which depends in the end on the contract's desired functionality.

\subsubsection{Transaction ordering dependency}
The TOD (Transaction ordering dependency) property is again not formally defined, but requires that the order of other transactions shall not influence the calls of the contract. More precisely, calls shall not depend on state that can be altered by other transactions.
The paper says that actually different types of dependency will be considered distinguishing whether the amount to be transferred (TA), the receiver (TR) or the reachability of the $\CALL$ as a whole are affected (TT).
However, it seems that TT is not implemented since not even the following straight forward TT violating contract is detected by the online tool:
{
\lstset{escapechar=,language=JavaScript, numbers=left,xleftmargin=2em,frame=lines,framexleftmargin=1.5em}
\begin{lstlisting}
contract SimpleGame {
  uint counter = 0;

  function play() {
    counter = 10;
  }

  function getReward() {
    if (counter > 0) {
      msg.sender.send(10);
    }
  }
}
\end{lstlisting}
}

The compliance pattern for TOD requires that calls shall not depend on the contract's storage or balance.
Again this property does not consider that there are different ways of calling, e.g., using $\CALLCODE$.

The following contract is considered secure:
{
\lstset{escapechar=,language=JavaScript, numbers=left,xleftmargin=2em,frame=lines,framexleftmargin=1.5em}
\begin{lstlisting}
contract Bob {
  uint price;

  function setPrice(uint v) {
    price = v;
  }

  function sendMoney(address c) {
    c.callcode.value(price)();
  }
}
\end{lstlisting}
}
This contract transfers an amount of money (\lstinline|price|) that it reads from the storage and that could have been modified by another transaction before. Still, no warning about a TOD violation is triggered by the online tool.

The violation pattern requires that there is a $\CALL$ which depends on a read of a constant storage cell that can be written.

Consider the following example contract:
{
\lstset{escapechar=,language=JavaScript, numbers=left,xleftmargin=2em,frame=lines,framexleftmargin=1.5em}
\begin{lstlisting}
contract Bob {
  uint price = 5;

  function sendMoney(address c) {
    price = price;
    c.send(price);
  }
}
\end{lstlisting}
}
This contract is labeled to be TOD even though the transferred amount is constantly $5$ and cannot be influenced by any other transaction.

\subsubsection{Validated arguments}
The VA (Validated arguments) property is again not semantically specified, but shall ensure that arguments to a function are checked for meeting desired preconditions.
Similarly to the HE property, it is unclear how such a goal should be captured by a generic property.

The compliance pattern requires that such values that depend on input value may only be written to the global storage if they have previously been checked, meaning that must have been a conditional branching before whose condition depended on the argument.

The following contract is an easy example of a contract matching the compliance pattern while not performing proper argument validation:
{
\lstset{escapechar=,language=JavaScript, numbers=left,xleftmargin=2em,frame=lines,framexleftmargin=1.5em}
\begin{lstlisting}
contract Test {
  uint test;
  uint count = 0;

  function setTest (uint x) {
    if (x < 10) {
      count++;
    }
    test = x;
  }
}
\end{lstlisting}
}
Even though the write to storage of argument variable \lstinline|x| is preceded by a corresponding conditional branch, this check does not influence whether the variable is indeed written to storage.

The violation pattern for VA requires that there is a storage instruction writing a value dependent on an argument that is not preceded by a corresponding conditional branch with a condition dependent on the argument.

The following contract performs a proper argument validation, but is still matched by the violation pattern.
{
\lstset{escapechar=,language=JavaScript, numbers=left,xleftmargin=2em,frame=lines,framexleftmargin=1.5em}
\begin{lstlisting}
library Lib {
  function validateArgument(uint i) {
    if (!(i >= 0 && i < 100))
      throw;
  }
}

contract Test {
  uint test;

  function setTest (uint x) {
    Lib.validateArgument(x);
    test = x;
  }
}
\end{lstlisting}
}

Since the validation is performed by the library function \lstinline|validateArgument|, the conditional branch which performs the validation cannot be detected.

\subsection{NeuCheck}
\label{sec:neucheck}
The tool NeuCheck\cite{lu2019neucheck} analyses Ethereum Smart contracts written in the Solidity by checking the contract's syntax graph for specific patterns. The tool is claimed to be sound even though no concrete soundness claim is formulated.

The formulated patterns are purely syntactic and can be rather seen as a check for the compliance with certain style guidelines.
Take as an example the access control pattern: this pattern checks whether all functions have modifiers such as \lstinline|private| or \lstinline|internal| defined which restrict general access.
It is unclear which semantic property should be implied by this pattern, and clearly there are safe usages of public functions as well as incorrect access control (e.g., due to a wrong party being allowed to call a certain contract function) even though a contract function is restricted by some modifier.

For illustrating the issues of this syntactic pattern-based approach further, we will in the following review the reentrancy pattern as this is particularly interesting for our case:
The reentrancy pattern checks for the occurrences of Solidity's \lstinline|call| function that does not have a gas limit checks and checks whether this occurrence is followed by the assignment of a state variable.
First, the absence of checking a gas limit does not ensure that reentrancy attacks are not possible.
If not gas limit is set, the gas given to the call is computed with respect to the remaining gas.
So one could easily set a high gas limit or even set all remaining gas of the execution as a gas limit (if the amount specified exceeds the remaining gas, the same gas as in the case of a lacking specification is given to the call).
On top of this, as discussed for Securify, setting a variable assignment is not the only way of changing the state (this can also be done via a library call).
The counterexample for Securify's NW compliance pattern would also be a valid counterexample for this case. Similarly, the Securify's NW violation pattern would serve as a counterexample for the patterns completeness.
As a consequence, matching the reentrancy pattern clearly does not guarantee the absence of a reentrancy attack.

Unfortunately, we could not experimentally assess the unsoundness of the provided patterns, since we did not find a way to build the tool from the provided sources\footnote{The sources are made available at \url{https://github.com/Northeastern-University-Blockchain/NeuCheck}. We contacted the authors at the end of November for clarification of the building process but received no reply as of January 20th, 2020.}.
The assessment of the tool is further aggravated by the fact that the paper gives the corresponding patterns in PseudoCode that leaves many crucial details (in particular how the dependency structure between syntactic constructs is established) undefined.

\subsection{EtherTrust}
\label{sec:ethertrust}
The approach to sound smart contract analysis presented in~\cite{GMS::CAV18} exhibits an unsoundness when it comes to modeling reentering executions. 
More precisely, the proposed abstraction assumes that the contract's storage at the point of reentering is the same as at the point of calling. This is not necessarily the case, since another (malicious) contract might in the meanwhile manipulate the storage of the corresponding contract by invoking other state-changing functions of it. 
An example is the \lstinline|Bank| contract depicted in \autoref{fig:reentrancy}. This contract would be deemed secure according to the abstraction presented in~\cite{GMS::CAV18} since the described attack requires to change the value of the contract's \lstinline|lock| variable by an invocation of the \lstinline|release| function prior to reentering the \lstinline|drain| function. If it is assumed that the storage at the point of reentering the \lstinline|drain| function is the same as at the point of invoking the \lstinline|call| method, the contract would be secure since the \lstinline|lock| variable is always set to \lstinline|1| when calling hence preventing to reach the \lstinline|call| method when reentering.

\subsection{ZEUS}
\label{sec:zeus}
A recently published work is the analysis tool ZEUS~\cite{kalra2018zeus} that analyses smart contracts written in Solidity using symbolic model checking.
The analysis proceeds by translating Solidity code to an abstract intermediate language that again is translated to LLVM bitcode. Finally, existing symbolic model checking tools for LLVM bitcode are leveraged for performing the analysis.
The security properties are defined in terms of XACML style policies that are translated to state reachability assertions in the intermediate language (and finally to assertions in LLVM bitcode). The authors evaluate their tool for generic security properties (such as reentrancy) which are however not expressed in terms of policies (which are contract specific), but by an informal description of how to add specific assertions to contracts of interest. For some properties, e.g. reentrancy, the insertion of assertions is not sufficient and additional program modifications need to be applied to the original contracts.
The authors claim their tool to be sound which they support by a proof sketch and empirical results. This claim however has several shortcomings:
\begin{itemize}
\item There is no formal soundness statement made. In particular, there is no formal relation between the policy compliance of Solidity contracts and the analysis results established and also not covered in the proof sketch. 
\item The proof is sketchy and exhibits several holes and at least two flaws: While there is an intuitive argument why given the translation from Solidity to the abstract intermediate language are correct and adding assertions does not influence semantics, there is no proof provided for the statement that the translation from the intermediate language to LLVM bitcode preserves soundness. That this property does not hold is (indirectly) admitted by the authors as they discuss that the compiler optimizations on LLVM bitcode remove relevant contract behavior. Consquently, assuming that compiler optimizations on LLVM bitcode are semantics preserving, this clearly contradicts that the translation from the intermediate language preserves semantics. For one particular optimization, a fix is hard coded, but there is no formal argument given that this particular fix is sufficient for establishing soundness.
Also the claim that the provided translation from Solidity to the intermediate language is faithful can be clearly contradicted. This is due to a clear deviation in the call semantics of the intermediate language from the Solidity semantics. The mechanism underlying Solidity's call functionalities is the one of the $\textsf{CALL}$ instructions in EVM bytecode. In particular, this mechanism determines that the failing of a contract call causes the revocation of the global state to the point of calling. The proposed semantics of the intermediate language however does not allow for such a revocation (even by design). Grishchenko et al.~\cite{GMS::POST18} spotted a similar issue in the semantics used in Oyente~\cite{luu2016making}.
\item The final results for the predefined properties (such as reentrancy) are not covered by the soundness claim as there is no (formal) argument made that the performed program modifications are sound.
In particular the presented method for detecting same-function reentrancies is faulty: For detecting same-function reentrancy of a function $f$, $f$ is replicated (resulting in $f'$) and the Solidity's \lstinline|call| construct in $f$ is replaced by a call to $f'$ whose occurrence of \lstinline|call| is preceded by a false assertion for proving the unreachability of the corresponding call. This treatment is problematic in several ways: First, the use of the \lstinline|call| construct is not the only way of calling another contract, indeed it is way more common to use direct calls.
Second, similar to the problem discussed for~\cite{GMS::CAV18}, such an abstraction fails to detect the example of \autoref{fig:reentrancy}, even though this is clearly a case of same-function reentrancy. The problem is that for the used approach of replacing calls by invocations to $f'$ it is assumed that a call can at most be preceded by a direct invocation of $f$ without any other state changing function calls being happening in the meanwhile.  Consequently, single-entrancy (and even same-function single-entrancy) is a property that cannot be assessed by considering certain contract parts in isolation.
Consider the following to contracts:
{
\lstset{escapechar=,language=JavaScript, numbers=left,xleftmargin=2em,frame=lines,framexleftmargin=1.5em}
\begin{lstlisting}
contract Bank{
  uint lock;
  mapping (address => uint) bal;

  function take () {
    lock = 1;
  }

  function release () {
    lock = 0;
  }

  function drain(address a) {
    if (lock == 1) { throw; }
      lock = 1;
      a.call.value(bal[msg.sender])();
      bal[msg.sender] = 0;
      lock = 0;
   }
}
\end{lstlisting}
\begin{lstlisting}
contract Bank{
  uint lock;
  mapping (address => uint) bal;

  function drain(address a) {
    if (lock == 1) { throw; }
      lock = 1;
      a.call.value(bal[msg.sender])();
      bal[msg.sender] = 0;
      lock = 0;
  }
}
\end{lstlisting}
}
Even though the implementation of the \lstinline|drain| function is identical in both contracts, the first contract allows for a (same-function) reentrancy attack while the second does not. ZEUS, however, would label both of these contracts to be safe.
\end{itemize}
Unfortunately, we were not able to conduct an empirical evaluation of the described issues since no sources for ZEUS are made available. Our request to the authors of~\cite{kalra2018zeus} to provide us with sources or binaries that would allow us to experimentally access ZEUS has been denied. For this reason we were forced to conduct our comparison with ZEUS on the publicly available dataset for which~\cite{kalra2018zeus} reports numbers. We further discuss this dataset in the following.

\subsubsection{Problems in the ZEUS dataset}
\label{sec:appendix-dataset-problems}
While comparing \horst{} against the dataset used in \cite{kalra2018zeus}\footnote{\url{https://docs.google.com/spreadsheets/d/12_g-pKsCtp3lUmT2AXngsqkBGSEoE6xNH51e-of_Za8}} we encountered several problems.
The dataset is a list of $1524$ contracts with the classification provided by ZEUS and the assessment whether the authors consider this classification correct. No source or bytecode is provided.

Of these $1524$ contracts, $21$ have a name that does not resemble a Ethereum address (e.g. \texttt{Code\_3\_fdf6d\_faucet}).
Of the remaining $1503$, $397$ actually have a truncated address (i.e., $39$ instead of $40$ hexadecimal digits).
The remaining $1106$ addresses contain duplicates.
After removing them we arrive at $1033$ addresses.
For $286$ of these addresses we were not able to obtain the bytecode: $53$ have been self-destructed according to \url{https://etherscan.io} which makes retrieving their bytecode non-trivial, $232$ have no recorded transaction (in particular no transaction that created them) and 1 is an external account (i.e., an address with no code deployed).
This leaves us with $747$ addresses.
After removing contracts with the same bytecode we arrive at $720$ contracts\footnote{Note that the authors of \cite{kalra2018zeus} deduplicated their dataset on the source level, therefore it may well be that these same bytecodes were produced by different source codes}.
We contacted the authors of \cite{kalra2018zeus} on July 16th 2019 about these problems and received no answer as of January 20th 2020.

\end{document}